\documentclass[15pt]{caltech_thesis}
\usepackage[hyphens]{url}
\usepackage{amsmath,amsfonts}
\usepackage{amsthm}
\usepackage{float}
\usepackage{graphicx}
\usepackage{subcaption}
\usepackage[document]{ragged2e}
\usepackage[english]{babel}
\usepackage[colorlinks,urlcolor=black,linkcolor=black,citecolor=black]{hyperref}
\usepackage{titlesec}
\usepackage{todonotes}
\usepackage[ruled,vlined,linesnumbered,resetcount,algochapter]{algorithm2e}
\usepackage{enumitem}
\usepackage{setspace}

\usepackage[utf8]{inputenc}
\usepackage[T1]{fontenc}
\usepackage{txfonts}

\usepackage[
    backend=biber,natbib,
    style=numeric,
    sorting=none,
    maxnames=10,
    url=false,
    eprint=false,
    firstinits=true,
    hyperref=true,
    ]{biblatex}

\DeclareFieldFormat{journaltitle}{#1}
\DeclareFieldFormat[article]{title}{"\mkbibemph{#1}"}
\DeclareFieldFormat[article]{volume}{\textbf{#1}}
\DeclareFieldFormat[article]{issn}{}
\DeclareFieldFormat[article]{pages}{#1}
\renewbibmacro{in:}{}
\renewbibmacro{pp}{}
\AtEveryBibitem{\clearfield{month}}
\AtEveryCitekey{\clearfield{month}}
\AtEveryBibitem{\clearfield{issue}}
\AtEveryCitekey{\clearfield{issue}}
\AtEveryBibitem{\clearfield{number}}
\AtEveryCitekey{\clearfield{number}}

\newcommand{\mbf}{\mathbf}
\newcommand{\mrm}{\mathrm}
\newcommand{\ha}{\hat{a}}
\newcommand{\hc}{\hat{c}}

\begin{document}

\title{Finite Temperature Simulations of Strongly Correlated Systems}
\author{Chong Sun}

\degreeaward{Doctor of Philosophy}                 
\university{California Institute of Technology}    
\address{Pasadena, California}                     
\unilogo{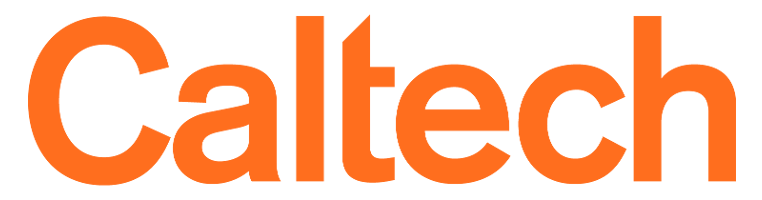}                                 
\copyyear{2021}  
\defenddate{December 1, 2020}          

\orcid{0000-0002-8299-9094}

\rightsstatement{All rights reserved}
\maketitle[logo]

\begin{acknowledgements} 	 
	I am deeply blessed as a member of the Caltech community. 
	Caltech provided me opportunities to participate in advanced research projects
	via collaborations with excellent researchers. I am truly thankful to the
	institute and every member of this big family.

	My advisor, Garnet Kin-Lic Chan, has provided me
	invaluable guidance and support throughout my graduate studies.
	From him, 
	I learned to always check my 
	hypothesis carefully against data. Whenever the data not seem reasonable,
	I should question my code first before questioning the theory or algorithm.
	He also provided plenty of
	opportunities for me to attend conferences and communicate with researchers
	in the field. Garnet is and will be the role model as a scientist to me 
	for the rest of my life.

	I would like to thank my committee, 
	Professor Mitchio Okumura, 
	Professor Thomas Miller and Professor Austin Minnich. They provided
	many helpful advices during my graduate career and tried to 
	bring the best out of me. I am also thankful to Professor Lu Wei who
	was really patient and helpful with my many questions about stimulated
	Raman spectroscopy. 

	Thank you, CCE administrative staff, in particular Alison and Elizabeth. Without
	the help from you, I would not have been able to fully focus on research without 
	worrying about many tough errands.

	I am grateful to be a member of the Chan group and work with so many awesome 
	colleagues. Everyone in the Chan group is nice and always willing to help. 
	I have been working closely with Zhihao, Mario, Ushnish, and Reza, from whom
	learned useful knowledge and skills.
	I am still close friends with previous group members such as Boxiao, Zhendong,
	and Mario, who constantly provide valuable suggestions to me when I need
	help.
	I also enjoyed group activities.
	Before the pandemic, we hung out monthly and tried many good or mediocre
	restaurants. We also had trips to Yosemite, Sequoia, and Universal Studios.
	My graduate school life is full of fun because of the Chan group villagers.

	Lastly, I would like to thank my family. My parents are the best parents I 
	could ever dream of. They did not have opportunities for good education,
	but they value education for my sister and I, and fully support 
	my career as a scientist. I was lucky to have my young sister as my close
	friend since childhood when most of my friends are only children in their 
	families. Having a smart and aggressive sibling was overall helpful to 
	push me to work harder. My husband, James has always been there to make 
	me laugh when I was unhappy with my research progress. Theoretically, my cat Jujube
	should thank me for providing her a home and food, but I know she 
	does not think in that way. I am thankful for her company and for not making loud
	noises when I have Zoom meetings.

\end{acknowledgements}

\newpage
\begin{vplace}[0.7]
\hspace*{\fill}{\large\emph{To my beloved parents}}
\end{vplace}

\justify
\begin{abstract}
This thesis describes several topics related to finite temperature studies
of strongly correlated systems: finite temperature density matrix embedding
theory (FT-DMET), finite temperature metal-insulator transition, and quantum
algorithms including quantum imaginary time evolution (QITE), quantum
Lanczos (QLanczos), and quantum minimally entangled typical thermal states (QMETTS)
algorithms.

While the absolute zero temperature is not reachable, studies of physical 
and chemical problems at finite temperatures, especially 
at low temperature, is essential for understanding the quantum 
behaviors of materials in realistic conditions. Here we define low 
temperature as the temperature regime where the quantum effect is not 
largely dissipated due to thermal fluctuation. Treatment of systems
at low temperature is specially difficult compared to both 
high temperature - where classical approximation can be applied - and zero temperature
where only the ground state is required to describe the system of interest.
FT-DMET is a
wavefunction-based embedding scheme which can handle finite temperature 
simulations of a variety of strongly correlated problems. The 
"high-level in low-level" framework enables FT-DMET to tackle large bulk sizes
and capture the majority of the entanglement at the same time. FT-DMET 
formulations and implementation details for both model systems and 
\textit{ab initio} problems are provided in Chapter~\ref{chp:dmet} and 
Chapter~\ref{chp:hlatt}.

Metal-insulator transition is a common but important phase transition in many
strongly correlated materials. The widely accepted scheme to distinguish
an insulator from a metal is band structure theory based on a single-particle
picture. However, insulating phases caused by disorder or strong correlation
cannot be explained merely with the band structure. In Chapter~\ref{chp:cp},
we demonstrate that electron locality/mobility is a more general criteria
to detect metal-insulator transition. We further introduce complex polarization
as the order parameter to reflect the electron locality/mobility and 
provide a formalism based on thermofield theory to evaluate the complex 
polarization at finite temperature.

Quantum algorithms are designed to perform simulations on a quantum device.
The infrastructure of a quantum processing unit (QPU) utilizes the 
superposition property of quantum bits (qubits), and thus can potentially
outplay the classical simulations in  computational scaling
for certain problems. In Chapter~\ref{chp:qite}, we introduce the
QITE algorithm, which can be applied to quantum
simulations of both ground state and finite temperature problems. We
further introduce a subspace method, QLanczos algorithm,
and a a finite temperature quantum algorithm, QMETTS, where QITE is used
as a building block for the two algorithms. We demonstrate above quantum
algorithms with simulations on both classical computers and quantum
computers.

\end{abstract}

\begin{publishedcontent}[iknowwhattodo]
\nocite{bMotta2020_own,Cui2020_own} 
\nocite{aSun2020_own}
\end{publishedcontent}

\chapter*{Contents}
\begin{center}
\SingleSpacing
\vskip 35pt
\SingleSpacing 
\boolfalse {citerequest}\boolfalse {citetracker}\boolfalse {pagetracker}\boolfalse {backtracker}\relax 
\defcounter {refsection}{0}\relax 
\contentsline {chapter}{{Acknowledgements}}{iii}{chapter*.1}%
\defcounter {refsection}{0}\relax 
\contentsline {chapter}{{Abstract}}{vi}{chapter*.2}%
\defcounter {refsection}{1}\relax 
\contentsline {chapter}{{Published Content and Contributions}}{viii}{chapter*.3}%
\defcounter {refsection}{0}\relax 
\contentsline {chapter}{Contents}{ix}{section*.4}%
\defcounter {refsection}{0}\relax 
\contentsline {chapter}{List of Figures}{xii}{section*.5}%
\defcounter {refsection}{0}\relax 
\contentsline {chapter}{List of Tables}{xviii}{section*.6}%
\defcounter {refsection}{0}\relax 
\contentsline {chapter}{\chapternumberline {1}Introduction}{1}{chapter.1}%
\defcounter {refsection}{0}\relax 
\contentsline {section}{\numberline {1.1}Finite temperature algorithms}{3}{section.1.1}%
\defcounter {refsection}{0}\relax 
\contentsline {subsection}{Direct evaluation of the trace}{4}{equation.1.1.3}%
\defcounter {refsection}{0}\relax 
\contentsline {subsection}{Imaginary time evolution}{7}{Item.13}%
\defcounter {refsection}{0}\relax 
\contentsline {section}{\numberline {1.2}Summary of research}{13}{section.1.2}%
\defcounter {refsection}{0}\relax 
\contentsline {chapter}{\chapternumberline {2}Finite temperature density matrix embedding theory}{18}{chapter.2}%
\defcounter {refsection}{0}\relax 
\contentsline {section}{\numberline {2.1}Abstract}{18}{section.2.1}%
\defcounter {refsection}{0}\relax 
\contentsline {section}{\numberline {2.2}Introduction}{18}{section.2.2}%
\defcounter {refsection}{0}\relax 
\contentsline {section}{\numberline {2.3}Theory}{20}{section.2.3}%
\defcounter {refsection}{0}\relax 
\contentsline {subsection}{Ground state DMET}{20}{section.2.3}%
\defcounter {refsection}{0}\relax 
\contentsline {subsubsection}{DMET bath construction}{21}{equation.2.3.1}%
\defcounter {refsection}{0}\relax 
\contentsline {subsubsection}{Embedding Hamiltonian}{23}{equation.2.3.7}%
\defcounter {refsection}{0}\relax 
\contentsline {subsubsection}{Self-consistency}{24}{equation.2.3.10}%
\defcounter {refsection}{0}\relax 
\contentsline {subsection}{Ground-state expectation values}{25}{equation.2.3.14}%
\defcounter {refsection}{0}\relax 
\contentsline {subsection}{Finite temperature DMET}{25}{equation.2.3.15}%
\defcounter {refsection}{0}\relax 
\contentsline {subsubsection}{Finite temperature bath construction}{26}{equation.2.3.15}%
\defcounter {refsection}{0}\relax 
\contentsline {subsubsection}{Thermal observables}{29}{equation.2.3.21}%
\defcounter {refsection}{0}\relax 
\contentsline {section}{\numberline {2.4}Results}{29}{section.2.4}%
\defcounter {refsection}{0}\relax 
\contentsline {subsection}{Computational details}{29}{section.2.4}%
\defcounter {refsection}{0}\relax 
\contentsline {subsection}{1D Hubbard model}{30}{section.2.4}%
\defcounter {refsection}{0}\relax 
\contentsline {subsection}{2D Hubbard model}{33}{figure.caption.15}%
\defcounter {refsection}{0}\relax 
\contentsline {section}{\numberline {2.5}Conclusions}{39}{section.2.5}%
\defcounter {refsection}{0}\relax 
\contentsline {chapter}{\chapternumberline {3}\textit {Ab initio} finite temperature density matrix embedding theory}{41}{chapter.3}%
\defcounter {refsection}{0}\relax 
\contentsline {section}{\numberline {3.1}Abstract}{41}{section.3.1}%
\defcounter {refsection}{0}\relax 
\contentsline {section}{\numberline {3.2}Introduction}{41}{section.3.2}%
\defcounter {refsection}{0}\relax 
\contentsline {section}{\numberline {3.3}\textit {Ab initio} FT-DMET}{43}{section.3.3}%
\defcounter {refsection}{0}\relax 
\contentsline {subsection}{Orbital localization}{43}{section.3.3}%
\defcounter {refsection}{0}\relax 
\contentsline {subsection}{Bath truncation and finite temperature bath}{44}{equation.3.3.1}%
\defcounter {refsection}{0}\relax 
\contentsline {subsection}{Embedding Hamiltonian}{47}{figure.caption.20}%
\defcounter {refsection}{0}\relax 
\contentsline {subsection}{Impurity solver}{49}{figure.caption.21}%
\defcounter {refsection}{0}\relax 
\contentsline {subsection}{Thermal observables}{52}{figure.caption.21}%
\defcounter {refsection}{0}\relax 
\contentsline {section}{\numberline {3.4}Results}{54}{section.3.4}%
\defcounter {refsection}{0}\relax 
\contentsline {section}{\numberline {3.5}Conclusion}{58}{section.3.5}%
\defcounter {refsection}{0}\relax 
\contentsline {chapter}{\chapternumberline {4}Finite temperature complex polarization and metal-insulator transition}{59}{chapter.4}%
\defcounter {refsection}{0}\relax 
\contentsline {section}{\numberline {4.1}Abstract}{59}{section.4.1}%
\defcounter {refsection}{0}\relax 
\contentsline {section}{\numberline {4.2}Introduction}{60}{section.4.2}%
\defcounter {refsection}{0}\relax 
\contentsline {section}{\numberline {4.3}Ground state complex polarization and electron localization }{62}{section.4.3}%
\defcounter {refsection}{0}\relax 
\contentsline {subsection}{Electron localization}{63}{equation.4.3.5}%
\defcounter {refsection}{0}\relax 
\contentsline {subsection}{Complex polarization for independent electrons}{65}{equation.4.3.17}%
\defcounter {refsection}{0}\relax 
\contentsline {section}{\numberline {4.4}Finite temperature complex polarization}{66}{section.4.4}%
\defcounter {refsection}{0}\relax 
\contentsline {section}{\numberline {4.5}Tight binding model}{69}{section.4.5}%
\defcounter {refsection}{0}\relax 
\contentsline {section}{\numberline {4.6}Hydrogen chain}{75}{section.4.6}%
\defcounter {refsection}{0}\relax 
\contentsline {section}{\numberline {4.7}Conclusion}{80}{section.4.7}%
\defcounter {refsection}{0}\relax 
\contentsline {chapter}{\chapternumberline {5}Quantum imaginary time evolution and quantum thermal simulation}{81}{chapter.5}%
\defcounter {refsection}{0}\relax 
\contentsline {section}{\numberline {5.1}Abstract}{81}{section.5.1}%
\defcounter {refsection}{0}\relax 
\contentsline {section}{\numberline {5.2}Introduction}{81}{section.5.2}%
\defcounter {refsection}{0}\relax 
\contentsline {section}{\numberline {5.3}Quantum imaginary-time evolution}{83}{section.5.3}%
\defcounter {refsection}{0}\relax 
\contentsline {section}{\numberline {5.4}Quantum Lanczos algorithm}{88}{section.5.4}%
\defcounter {refsection}{0}\relax 
\contentsline {section}{\numberline {5.5}Quantum thermal averages}{90}{section.5.5}%
\defcounter {refsection}{0}\relax 
\contentsline {section}{\numberline {5.6}Results}{92}{section.5.6}%
\defcounter {refsection}{0}\relax 
\contentsline {subsection}{Benchmarks}{95}{equation.5.6.22}%
\defcounter {refsection}{0}\relax 
\contentsline {section}{\numberline {5.7}Conclusions}{98}{section.5.7}%
\defcounter {refsection}{0}\relax 
\contentsline {appendix}{\chapternumberline {A}Appendix for Chapter~\ref {chp:dmet} and Chapter~\ref {chp:hlatt}}{99}{chapter.A}%
\defcounter {refsection}{0}\relax 
\contentsline {section}{\numberline {A.1}Proof of the finite temperature bath formula}{99}{section.A.1}%
\defcounter {refsection}{0}\relax 
\contentsline {section}{\numberline {A.2}Analytic gradient of the cost function for correlation potential fitting in DMET at finite temperature}{101}{section.A.2}%
\defcounter {refsection}{0}\relax 
\contentsline {section}{\numberline {A.3}Davidson diagonalization}{102}{section.A.3}%
\defcounter {refsection}{0}\relax 
\contentsline {appendix}{\chapternumberline {B}Appendix for Chapter~\ref {chp:qite}}{104}{chapter.B}%
\defcounter {refsection}{0}\relax 
\contentsline {section}{\numberline {B.1}Representing imaginary-time evolution by unitary maps}{104}{section.B.1}%
\defcounter {refsection}{0}\relax 
\contentsline {section}{\numberline {B.2}Proof of correctness from finite correlation Length}{105}{section.B.2}%
\defcounter {refsection}{0}\relax 
\contentsline {section}{\numberline {B.3}Spreading of correlations}{110}{section.B.3}%
\defcounter {refsection}{0}\relax 
\contentsline {section}{\numberline {B.4}Parameters used in QVM and QPUs simulations}{112}{section.B.4}%
\defcounter {refsection}{0}\relax 
\contentsline {chapter}{Bibliography}{114}{section*.47}%

\end{center}
\listoffigures
\listoftables

\mainmatter

\chapter{Introduction\label{chp:intro}}

We live in an era where the computational power is one of the main driving 
forces for science and technology development. The hardware breakthroughs in
supercomputers, graphical processing unit (GPU) and quantum computers made
heavy computational tasks possible. The development in machine learning
algorithms and artificial intelligence changed the way people live 
tremendously. Many new materials and drugs are discovered via computational 
simulations, saving hundreds of laboratory hours. We believe in the 
computational power to bring us new knowledge and concepts, as well as to 
solve fundamental problems that remain unclear for decades. 
In quantum chemistry and condensed matter physics, those hard problems 
include the phase diagram of high-temperature superconductors (HTSC)
~\citep{Dagotto1994,Nikolay2010},
 the mechanism of nitrogen fixation~\citep{Hoffman2014,Cherkasov2015}, 
protein folding~\citep{Englander2014}, etc.
The barrier for efficient simulations of the above problems is usually 
either the system size is too big or the interaction is too complicated. 
The strongly correlated systems, unfortunately, have both of the above two barriers.
The hallmark of strongly correlated systems is localized orbitals such
as $d$ and $f$ orbitals, where electrons experience strong Coulomb repulsion. 
For instance, transition metal compounds usually contain strong correlations
due to the localized $3d$ orbitals. 
 Strongly correlated materials attract
tremendous interest of both experimental and theoretical researchers
because they exhibit a plethora of exotic phases or behaviors: HTSC, 
spintronic materials~\citep{Hirohata202016}, Mott insulators~\citep{Hubbard1963}, etc. Those
strongly correlated behaviors evoked novel applications such as
quantum processing units~\citep{Ladd2010}, superconducting magnets~\citep{Wilson1983,Chester1967},
and magnetic storage~\citep{Comstock2002}. Being able to simulate strongly correlated
problems and thus understand the physics behind them has been a key 
task for theoretical and computational chemists.

This thesis focuses on developing theoretical and computational approaches
to simulate strongly correlated problems at finite temperature. 
While ground state simulations provide basic information on the system
such as ground state energy and band gap, finite temperature is where
the real-life phase transitions happen. The complexity of a quantum
many-body problem can be described by a term called \textit{entanglement}.
At ground state away from the critical point, the entanglement is bounded
by the area law~\citep{Eisert2010}. However, at finite temperature, 
especially low temperature where the quantum effect is not fully dissipated 
by thermal fluctuation, the area law is no longer valid. One would expect
the entanglement strength to decay while the 
 entanglement length to grow with temperature. The interplay between the 
entanglement strength and entanglement length decides the complexity of the
system. Normally one would expect more computational efforts for finite 
temperature calculations than ground state calculations.

The complexity of finite temperature calculations can also be understood 
in the ensemble picture. Most of the physical and chemical systems can be
seen as open systems, where the thermodynamic statistics is described by
the grand canonical ensemble. In the grand canonical ensemble picture, both
energy fluctuations and particle number fluctuations are involved. 
The system at temperature $T$ is fully described by the density matrix 
\begin{equation}
\hat{\rho}(T) = e^{-(\hat{H} - \mu\hat{N})/k_BT},
\end{equation}
where $\hat{H}$
is the Hamiltonian, $\mu$ is the chemical potential, $\hat{N}$ is the
number operator and $k_B \approx 1.38\times 10^{-23} \mathrm{J}\cdot\mathrm{K}^{-1}$ is the Boltzmann constant. The partition function
 is defined as the trace of the density matrix:
 $\mathcal{Z} = \text{Tr}(\hat{\rho})$. If one choose the
eigenstates of the Hamiltonian $\hat{H}$ as the basis to perform
the trace summation, each eigenstate would participate in the statistics with
probability 
\begin{equation}\label{eq:intro_prob}
P(n, i) = e^{-(\varepsilon_i^n - \mu n)/k_BT}/\mathcal{Z},
\end{equation}
where $\varepsilon_i^n$ is the eigenvalue of the $i$th eigenstate in the
Fock space of $n$ particles. If $\varepsilon_i^n < \mu n$, $P(n,i)$ 
decreases to $1/\mathcal{N}$ as temperature rises; if 
$\varepsilon_i^n > \mu n$, $P(n,i)$ increases to $1/\mathcal{N}$ 
as temperature rises, where $\mathcal{N}$ is the total number of eigenstates. 
At $T=0$, only the ground state is involved;
as one raises the temperature, the contribution from the ground state
drops and excited states enter the ensemble. Eventually at infinite 
temperature, all states are equally involved with a probability 
$1/\mathcal{N}$. The inclusion of many excited states is the source of 
the high complexity of finite temperature simulations. For instance,
for an electronic structure problem with $L$ orbitals, where each orbital
can take four states: $|0\rangle$, $|\uparrow\rangle$, $|\downarrow\rangle$,
and $|\uparrow\downarrow\rangle$. The total number of states is $\mathcal{N} = 4^L$,
which scales exponentially with $L$.  

Albeit the high computational cost of finite temperature simulations, there
exist a variety of finite temperature algorithms that can fulfill different
computational tasks. 
Section~\ref{sec:ftalgos} presents a detailed review of current finite 
temperature algorithms. We hope this review could be helpful to researchers
who are interested in learning about or using finite temperature algorithms.
Section~\ref{sec:sumsec} provides an outline for the rest of the chapters in
this thesis.

\section{Finite temperature algorithms\label{sec:ftalgos}}
At finite temperature $T$, the grand canonical ensemble average of an operator
$\hat{O}$ is evaluated by
\begin{equation}\label{eq:average_intro}
\langle \hat{O}\rangle (T) =  \frac{\mrm{Tr}\left(e^{-(\hat{H}- \mu \hat{N})/k_BT } \hat{O}\right)}
{\mrm{Tr}\left(e^{-(\hat{H} -\mu\hat{N})/k_BT}\right)}.
\end{equation}
There are generally two approaches to design a finite temperature algorithm:
(i) directly evaluate the trace in Eq.~\eqref{eq:average_intro} by summation
over the expectation values under an orthonormal basis; (ii) imaginary 
time evolution from infinite temperature. Theoretically the two approaches
are all based on Eq.~\eqref{eq:average_intro}, so one could argue that 
there is no big difference between the two approaches. Technically, however,
the first approach usually involves exact or approximate diagonalization
of the Hamiltonian, while the latter approach does not. In the following,
we will discuss the two approaches with some example algorithms. 

\subsection{Direct evaluation of the trace}
We first discuss the non-interacting case. For a non-interacting Hamiltonian,
only one-body terms are involved, and the Hamiltonian can be simply written 
as an $L\times L$ matrix, where $L$ is the number of orbitals in the system.
For most cases, this $L\times L$ Hamiltonian matrix can be directly 
diagonalized, with eigenvalues $\varepsilon_i$ and eigenvectors 
$|\phi_i\rangle$ (molecular orbitals, MOs). A direct implementation of 
Eq.~\eqref{eq:average_intro} is to construct Slater determinants of
all possible particle numbers and evaluate the traces, where the number
of Slater determinants in the summation scales exponentially with $L$.
Luckily, for non-interacting electrons, the grand canonical density matrix 
can be evaluated by Fermi-Dirac equation
\begin{equation}\label{eq:fd_intro}
\rho = \frac{1}{1+e^{(H-\mu\mathbb{I})/k_BT}},
\end{equation}
where $\mathbb{I}$ is the identity matrix. The occupation numbers on 
MOs are the diagonal terms of the density matrix:  
$n_i = 1/(1+e^{(\varepsilon_i - \mu)/k_BT})$. Thus Eq.~\eqref{eq:average_intro}
can be rewritten as
\begin{equation}
\langle \hat{O}\rangle_{NI} (T) = \sum_{ij} \rho_{ij} \langle \phi_j |\hat{O}|\phi_i\rangle,
\end{equation}
where the subscript "NI" stands for "non-interacting". 

Finite temperature Hartree-Fock is an example of the above approach, with the
algorithm summarized in Algorithm~\ref{alg:fthf}.

\begin{algorithm}[h]
\SetAlgoLined
\vspace{0.2em}
\begin{description}[topsep=0pt,itemsep=-1ex,partopsep=1ex,parsep=1ex,leftmargin=*]
\item[] Construct the Fock matrix $F$ from the Hamiltonian.
Define $F'$ as identity.\\
\While{$F\neq F'$}{
\begin{enumerate}[topsep=0pt,itemsep=-1ex,partopsep=1ex,parsep=1ex,leftmargin=*]
\item Store the Fock matrix into $F' = F$;
\item Diagonalized $F$ to get MO energies and coefficients;
\item Calculate the chemical potential $\mu$ by minimizing 
$(N_{\text{elec}} -  \sum_i n_i)^2$, where $N_{\text{elec}}$ is the
target electron number and $n_i$ is the occupation number of the $i$th MO;
\item Calculate density matrix $\rho$ from Eq.~\eqref{eq:fd_intro} by 
substituting $H$ with $F$;
\item Evaluate the new Fock matrix $F$ from the density matrix $\rho$ as in
ground state Hartree-Fock algorithm;
\end{enumerate}
}
\item[] Evaluate thermal observables with converged $\rho$.
\end{description}
\caption{Finite temperature Hartree-Fock algorithm}\label{alg:fthf}
\end{algorithm}

Note that in above algorithm, the convergence criteria can also be the 
 density matrix or MO energies.

For the interacting case, a naive approach is exact diagonalization
(ED), where all eigenstates of the Hamiltonian $\hat{H}$ are explicitly
calculated and the thermal average of an observable $\hat{A}$ is evaluated 
by
\begin{equation}
\langle \hat{O}\rangle (T) = \frac{\sum_{n,i}\langle \psi_i^n | \hat{O} 
e^{-(\varepsilon_i^n - \mu n)/k_BT}|\psi_i^n\rangle}
{\sum_{n,i}\langle \psi_i^n | e^{-(\varepsilon_i^n - \mu n)/k_BT}|\psi_i^n\rangle},
\end{equation}
where $|\psi_i^n\rangle$ is the $i$th eigenstate in the Fock space with 
$n$ particles. The algorithm of ED is described in 
Algorithm~\ref{alg:intro_ed}. The expense of ED scales exponentially
with the number of orbitals $L$, and thus is only limited to small systems. 
For electronic systems with two spins, the maximum $L$ is $\sim 8$. Therefore,
nearly no meaningful calculations can be done with ED.

\begin{algorithm}[H]
\SetAlgoLined
\vspace{0.2em}
\begin{description}[topsep=0pt,itemsep=-1ex,partopsep=1ex,parsep=1ex,leftmargin=*]
\item[] $\mathcal{Z} = 0, O = 0$;
\item[] \For{$n_a$ in $[0, L]$}{
\For{$n_b$ in $[0, L]$}{
\begin{enumerate}[topsep=0pt,itemsep=-1ex,partopsep=1ex,parsep=1ex]
\item Construct Hamiltonian $H(n_a, n_b)$;
\item Diagonalize $H(n_a, n_b)$ to get eigenvalues $\varepsilon_i^{n_a, n_b}$ and
eigenstates $\{|\psi_i^{n_a, n_b}\rangle\}$;
\item Evaluate $\mathcal{Z}^{n_a, n_b} = \sum_i e^{-(\varepsilon_i^{n_a, n_b} - \mu (n_a + n_b))/k_BT}$ and $O^{n_a, n_b} = \sum_i e^{-(\varepsilon_i^{n_a, n_b} - \mu (n_a + n_b))/k_BT} \langle \psi_i^{n_a, n_b} | \hat{O} | \psi_i^{n_a, n_b}\rangle$;
\item $\mathcal{Z}$ += $\mathcal{Z}^{n_a, n_b}$ ; $O$ += $O^{n_a, n_b}$;
\end{enumerate}
}
}
\item[] $\langle O \rangle (T) = O/\mathcal{Z}$ 
\end{description}
\caption{Finite temperature exact diagonalization}\label{alg:intro_ed}
\end{algorithm}

One could reduce the computational cost by only including low-lying states 
in the ensemble.
Davidson diagonalization~\citep{Davidson1975} and Lanczos algorithm~\citep{Lanczos1950} are two
methods that construct a smaller subspace of the Hilbert space containing
the low-lying states.
In the Lanczos algorithm, starting with a normalized vector $|\phi_0\rangle$, one 
could generate a set of orthonormal Lanczos vectors $\{|\phi_m\rangle, m = 0, ..., M\}$ to span the Krylov space $\{|\phi_0\rangle, \hat{H} |\phi_0\rangle, ..., \hat{H}^M |\phi_0\rangle\}$ with the following steps:
\begin{enumerate}
\item Apply $\hat{H}$ to $|\phi_0\rangle$ and split the resulting vector into
$a_0|\phi_0\rangle$ and $b_1 |\phi_1\rangle$ with $|\phi_1\rangle \perp |\phi_0\rangle$
\begin{equation}
\hat{H}|\phi_0\rangle = a_0 |\phi_0\rangle + b_1 |\phi_1\rangle,
\end{equation}
where $a_0 = \langle \phi_0 | \hat{H}|\phi_0\rangle$ and $b_1$ is chosen so
that $|\phi_1\rangle$ is normalized.
\item Iteratively apply $\hat{H}$ to $|\phi_i\rangle, i = 1,...,M$ to get
\begin{equation}
|\phi_i\rangle = b_i |\phi_{i-1}\rangle + a_i |\phi_i\rangle + b_{i+1}
|\phi_{i+1}\rangle,
\end{equation} 
where the iteration stops at $i=M$ with $b_{M+1} = 0$ or when $b_i = 0$ with
$i < M$.
\item Construct the matrix representation of the Krylov space Hamiltonian as
\begin{equation}
H' = \begin{bmatrix}
a_0 & b_1 & 0 & \cdots & 0 \\
b_1 & a_1 & b_2 & \cdots & 0\\
0 & b_2 & a_2 & \cdots & 0\\
 &  &  & \ddots &  \\
0 & 0 & 0 & \cdots & a_M
\end{bmatrix},
\end{equation}
where we choose $b_i$ to be real numbers. 
\item Diagonalize the Krylov Hamiltonian $H'$ to get the eigenvalues and
eigenvectors in the basis of $\{|\phi_i\rangle, i = 0, ..., M\}$.
Note that the 
$H'$ is a tridiagonal matrix, and the typical cost to diagonalize an 
$M\times M$ symmetric tridiagonal matrix is $\mathcal{O}(M^2)$, while the 
cost of diagonalizing a random symmetric $M\times M$ matrix is $\mathcal{O}(M^3)$.
\end{enumerate}
The quality of the Krylov space depends heavily on the initial state 
$|\phi_0\rangle$. For instance, if $|\phi_0\rangle$ has zero overlap
with the ground state, then the leading part of the trace summation at low
temperature is missing and the result is inaccurate. One could 
sample initial states and take the average of the sample to get a better approximation. Note that the above routine is for a system with fixed particle numbers,
so to fulfill the grand canonical ensemble, one should also sample the
Fock spaces with all possible particle numbers. For low temperature 
simulation, sampling particle numbers near the targeted electron number
 is usually enough. We also provide a summary of the Davidson algorithm in 
Appendix~\ref{sec:apdx_davidson}

\subsection{Imaginary time evolution}
The imaginary time evolution operator is defined as
$e^{-\beta \hat{H}}$, where $\beta$ is called the imaginary time. This 
approach can be used in both ground state search and the finite 
temperature calculations. In the latter case, $\beta$ has a physical 
meaning: the
inverse temperature $\beta = 1/k_B T$. At $\beta = 0$ (infinite temperature),
the density matrix $\hat{\rho}(\beta = 0)$ is proportional to the identity operator
 and the
system is maximally entangled. 
Differentiating $\hat{\rho}(\beta) = e^{-\beta \hat{H}}$ with respect to $\beta$ is described by the Bloch equation
\begin{equation}\label{eq:bloch_eq_intro}
\frac{\mrm{d}\hat{\rho}}{\mrm{d}\beta} = -\hat{H}\hat{\rho}
= -\frac{1}{2}(\hat{H}\hat{\rho} + \hat{\rho}\hat{H}),
\end{equation}
where the last equal sign used $[\hat{H}, e^{-\beta \hat{H}}] = 0$. The
solution to Eq.~\eqref{eq:bloch_eq_intro} can also be written in a
symmetrized form
\begin{equation}\label{eq:dm_evolve_intro}
\hat{\rho}(\beta) = e^{-\beta \hat{H}/2} \hat{\rho}(\beta = 0) e^{-\beta \hat{H}/2}.
\end{equation}

Density matrix quantum Monte Carlo (DMQMC)~\citep{Blunt2014,Petras2020} is an example of the above
approach. We introduce an energy shift $\Delta E$ to the original Hamiltonian
$\hat{H}$, and Eq.~\eqref{eq:bloch_eq_intro} turns into
\begin{equation}
\frac{\mrm{d}\hat{\rho}}{\mrm{d}\beta}= -\frac{1}{2}(\hat{T}\hat{\rho} + \hat{\rho}\hat{T}),
\end{equation}
where $\hat{T} = \hat{H} - \Delta E\hat{\mathbb{I}}$, and $\Delta E$ is 
slowly adjusted to control the population. A similar concept
of $\Delta E$ is also employed in diffusion Monte Carlo (DMC)~\citep{Hammond1994,Foulkes2001}
and full configuration interaction quantum Monte
Carlo (FCIQMC)~\citep{Booth2009,Booth2013}.

The general form of $\hat{\rho}(\beta)$ can be written as a linear combination
\begin{equation}
\hat{\rho}(\beta) = \sum_{ij}\rho_{ij}(\beta) |\psi_i\rangle\langle\psi_j|,
\end{equation}
where $\{|\psi\rangle\}$ forms a complete orthonormal basis of the Hilbert
space. Here we choose $\{|\psi\rangle\}$ to be Slater determinants.
$\{|\psi_i\rangle\langle\psi_j|\}$ forms a basis for operators in this
Hilbert space, denoted as $\{X_{ij}\}$ for simplicity. Here we introduce 
a term "psips"~\citep{Anderson1975,Anderson1976}: each psip resides on a particular basis operator $X_{ij}$ or
site $(i,j)$ with "charge" $p_{ij} = \pm 1$. The imaginary time evolution is divided into $N_{\beta}$
tiny steps: $\tau = \beta / N_{\beta}$. For each step, DMQMC
loops over the sample of psips and perform the following steps:
\begin{enumerate}
\item \textbf{Spawning along columns of the density matrix}. Starting from a
 psip on site $(i,j)$, calculate the transition probabilities
$\frac{1}{2}|T_{ik}|\tau $ to spawn onto sites $(k,j)$ with $T_{ik}\neq0$
and $i\neq k$. If the spawning attempt is accepted, a psip is born 
on site $(k,j)$ with charge $q_{kj} = \mrm{sign}(T_{ik})q_{ij}$.
\item \textbf{Spawning along rows of the density matrix}. Repeat the above step to 
spawn psips from site $(i,j)$ onto sites $(i,k)$.
\item \textbf{Psips replication and death}. Evaluate the diagonal sum $d_{ij} = T_{ii} + T_{jj}$ for site $(i,j)$: if $d_{ij} > 0$, a copy of the psip on 
site $(i,j)$ is added to the pool with probability $p_d = \frac{1}{2} 
|d_{ij}|\tau$; if $d_{ij} < 0$, the psip on site $(i,j)$ is removed
with probability $p_d$.
\item \textbf{Annihilation}. Pairs of psips on the same site with opposite charges
are removed from the pool.
\end{enumerate}
The distribution of the psips generated by repeating $N_{\beta}$ times the
above procedure provides an approximation of the unnormalized density matrix 
at $\beta$. The thermal average of an observable $\hat{O}$ is then calculated
by
\begin{equation}
\langle \hat{O}\rangle (\beta) = \frac{\sum_{ij}\bar{q}_{ij}O_{ji}}
{\sum_i \bar{q}_{ii}},
\end{equation}
where $\bar{q}$ is an average of density matrices evaluated from a large 
number of repeats of the above imaginary time evolution process.

The main concern of the above approach is the size of the density matrix. The 
number of independent elements in the density matrix is 
$\sim \mathcal{N}(\mathcal{N}+1)/2$, where $\mathcal{N}$ is the 
Hilbert space size which grows exponentially with the system size.
Even with heavy parallelization, DMQMC still suffers from considerable
computational cost. Moreover, the accuracy of DMQMC becomes worse as 
the temperature lowers, limiting this method to applications for intermediate
or high temperature calculations.

One could circumvent evolving a density matrix by artificially constructing 
an enlarged space in which the density matrix of the original system can
be obtained by partial trace from the pure state solution of the enlarged 
system. The above approach is called purification~\citep{Palser1998}. The idea of purification
is the following: suppose a system $\mathcal{S}$ can be bipartitioned into 
two smaller systems $\mathcal{A}$ and $\mathcal{B}$; then a state 
$|\Psi\rangle$ in $\mathcal{S}$ can be written as 
\begin{equation}
|\Psi\rangle = \sum_{ij} c_{ij}|A_i\rangle|B_j\rangle,
\end{equation}
where $\{|A_i\rangle\}$ and $\{|B_i\rangle\}$ are orthonormal bases
of $\mathcal{A}$ and $\mathcal{B}$ respectively, and $\sum_{ij}|c_{ij}|^2 = 1$.
 The density matrix of
the total system is $\hat{\rho}_{\mathcal{S}} = |\Psi\rangle\langle\Psi|$,
and the density matrix of $\mathcal{A}$ can be obtained by
\begin{equation}\label{eq:purify_intro}
\begin{split}
\hat{\rho}_{\mathcal{A}} &= \text{Tr}_{\mathcal{B}}\left(\hat{\rho}_{\mathcal{S}}\right) \\
&= \sum_k \langle B_k| \left(\sum_{ij} c_{ij} |A_i\rangle|B_j\rangle \right)
\left(\sum_{i'j'} c^*_{i'j'} \langle A_{i'}|\langle B_{j'}|\right)
| B_k\rangle\\
&= \sum_{ii'} \left(\sum_k c_{ik}c^*_{i'k}\right)|A_i\rangle\langle A_{i'}|\\
&= \sum_{ii'} w_{ii'} |A_i\rangle\langle A_{i'}|.
\end{split}
\end{equation}
Eq.~\eqref{eq:purify_intro} has the form of a density matrix operator, with 
matrix elements $w_{ii'}$.  
The matrix $\mathbf{w}$ has the following properties: (i) Hermitian;
(ii) diagonal terms $w_{ii} = \sum_{k} |c_{ik}|^2 \geq 0$; 
and (iii) $\sum_i w_{ii} = 1$. Based on the above properties, we confirm that
$\mathbf{w}$ is a density matrix. 

Given a density matrix $\hat{\rho}_{\mathcal{A}}$ and basis $\{|A_i\rangle\}$, 
one could also find a set of $\{|B_i\rangle\}$ to construct a state $|\Psi\rangle$ such that $\hat{\rho}_{\mathcal{A}}$ can be derived from the partial trace
of $|\Psi\rangle\langle\Psi|$ with $\{|B_i\rangle\}$. The above procedure is 
called purification. Note that for a system $\mathcal{A}$, there exist
more than one purified state $|\Psi\rangle$, and one could choose certain
$\{|B_i\rangle\}$ and $|\Psi\rangle$ for their convenience. At infinite 
temperature, the density matrix of subspace $\mathcal{A}$ can be written as
\begin{equation}
\hat{\rho}_{\mathcal{A}}(\beta = 0) = \frac{1}{N_{\mathcal{A}}}\sum_{i}
|A_i\rangle\langle A_i|,
\end{equation}
where $N_{\mathcal{A}}$ is the size of $\mathcal{A}$. One could introduce
a set of ancillary orbitals $\{|\tilde{A}_i\rangle\}$ which are copies of 
$\{|A_i\rangle\}$ and define the purified state as
\begin{equation}
|\Psi(\beta = 0)\rangle = \frac{1}{\sqrt{N_{\mathcal{A}}}}\sum_{i}
|A_i\rangle|\tilde{A}_i\rangle.
\end{equation}
It is easy to prove that $\hat{\rho}_{\mathcal{A}}(\beta = 0) $ can be derived as the partial trace of $|\Psi(\beta = 0)\rangle\langle \Psi(\beta = 0)|$ 
with  $\{|\tilde{A}_i\rangle\}$. 

Now one could apply imaginary time evolution onto $|\Psi(\beta = 0)\rangle$
instead of $\hat{\rho}_{\mathcal{A}}(\beta = 0)$,
\begin{equation}
|\Psi(\beta)\rangle \propto e^{-\beta(\hat{H}\otimes \hat{\mathbb{I}})}
|\Psi(\beta = 0)\rangle,
\end{equation}
where $\hat{H}$ is the original Hamiltonian on $\mathcal{A}$ and 
$\hat{\mathbb{I}}$ is the identity operator on $\tilde{\mathcal{A}}$. 
The thermal average of operator $\hat{O}$  in $\mathcal{A}$ is simply evaluated as 
\begin{equation}
\langle \hat{O} \rangle (\beta) = \langle \Psi(\beta)| \hat{O}\otimes \hat{\mathbb{I}} |\Psi(\beta)\rangle.
\end{equation}
The most time consuming step in the above procedure is applying $e^{-\beta\hat{H}}$
onto $|\Psi\rangle$. A commonly accepted way to deal with $e^{-\beta \hat{H}}$
is Trotter-Suzuki decomposition. Again we divide $\beta$ into $N_{\beta}$
tiny steps $\tau = \beta/N_{\beta}$, and 
$e^{-\beta\hat{H}} = \left(e^{-\tau \hat{H}}\right)^{N_{\beta}}$, where
we assumed that $\hat{H}$ does not change with temperature. 
Suppose $\hat{H}$ can be decomposed into $\hat{H} = \hat{H}_1 + \hat{H}_2
+ \cdots + \hat{H}_n$, according to Trotter-Suzuki approximation
\begin{equation}
e^{-\tau \hat{H}} = e^{-\tau \hat{H}_1/2}e^{-\tau \hat{H}_2/2} \cdots 
e^{-\tau \hat{H}_2/2} e^{-\tau \hat{H}_1/2} + \mathcal{O}(\tau^3).
\end{equation}

Another more accurate approach is the 4th order Runge-Kutta (RK4) algorithm, which
is based on solving the differentiation form of the imaginary time evolution
\begin{equation}
\frac{\mrm{d}|\Psi\rangle}{\mrm{d}\beta} = -\hat{H} |\Psi\rangle.
\end{equation}
Let $t_m = m\tau$, then one update step in RK4 algorithm is
\begin{equation}
|\Psi(t_{m+1})\rangle = |\Psi(t_{m})\rangle = \frac{1}{6}\tau 
(k_1 + 2k_2 + 2k_3 + k_4),
\end{equation}
with initial condition $t_0 = 0$ and $|\Psi(t_0)\rangle = 
|\Psi(\beta=0)\rangle$. $k_i (i=1,2,3,4)$ are defined from the $m$th step
values
\begin{equation}
\begin{split}
k_1 &= -\hat{H} |\Psi(t_{m})\rangle, \\
k_2 &= -\hat{H} \left(|\Psi(t_{m})\rangle + \frac{\tau}{2}k_1 \right),\\
k_3 &= -\hat{H} \left(|\Psi(t_{m})\rangle + \frac{\tau}{2}k_2 \right),\\
k_4 &= -\hat{H} \left(|\Psi(t_{m})\rangle + \tau k_3 \right).\\
\end{split}
\end{equation}
The error of one RK4 iteration scales as $\mathcal{O}(\tau^5)$, and 
the accumulated error is $\mathcal{O}(\tau^4)$.

An example which adopted the purification approach is the finite temperature
density matrix renormalization group (FT-DMRG)~\citep{Feiguin2005} algorithm. The matrix product
state (MPS) is defined with alternating physical and ancillary sites, as shown
in Fig.~\ref{fig:ftdmrg_mps}. The operators are arranged in the same 
alternating manner. The imaginary time evolution routine then follows
the same procedure as previously developed time-evolving block decimation
(TEBD)~\citep{Verstraete2004,Vidal2004}.

\begin{figure}
\centering
\justify
\includegraphics[width=1\textwidth]{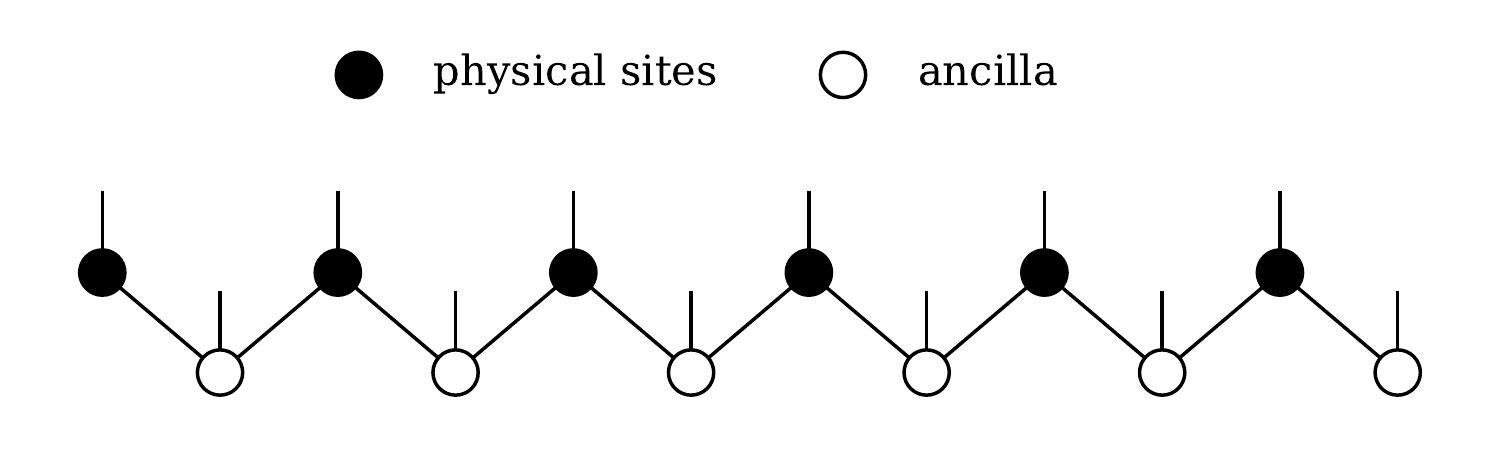}
\caption{Structure of the matrix product states used in the purification approach
of the finite temperature density matrix renormalization group algorithm. 
 }\label{fig:ftdmrg_mps}
\index{figures}
\end{figure}

In addition to the examples mentioned above, there exist several other finite
temperature algorithms. Minimally entangled typical thermal states (METTS)
algorithm~\citep{White2009,Stoudenmire2010} which will be mentioned in Chapter~\ref{chp:qite} is 
another fulfillment of finite temperature DMRG based on importance sampling.
Compared to the purification approach, METTS requires a smaller bond
dimension and the statistical error decreases as the temperature lowers.
However, METTS has only been applied to spin systems, because the original
 formulation does not allow the variation of electron numbers and thus is
limited to canonical ensemble. One could potentially adapt METTS for
a grand canonical ensemble by sampling the electron numbers or introducing
a set of initial states which do not preserve the electron numbers.
Determinantal quantum Monte Carlo (DQMC)~\citep{Blankenbecler1981} 
and finite temperature auxiliary field quantum Monte Carlo (FT-AFQMC)~\citep{Liu2018,He2019}
are two other finite temperature algorithms based on importance sampling
of Slater determinants. Both of the two QMC methods utilizes
Hubbard-Stratonovich transformation to transform the many-body 
imaginary time evolution operator to single-particle operators 
expressed as free fermions coupled to auxiliary fields. AFQMC applies
a constrained path to alleviate the sign problem, yet the computational
cost is still non-negligible to reach low enough temperatures with large
system sizes. The dynamical mean-field theory (DMFT)~\citep{Georges1996,Kotliar2006} is an embedding
 method which maps a many-body lattice problem to a many-body local problem.
Since DMFT evaluates the frequencies, it can be naturally extended to 
finite temperature calculations with a finite temperature impurity solver.
As most embedding methods, DMFT results are affected by the 
finite size effect, and extrapolation to thermodynamic limit (TDL) is 
needed to remove the artifact from the finite impurity size. All the above numerical
algorithms have their pros and cons, and one could make their choices 
based on the properties of the system and evaluate the results by 
careful benchmarking.

\section{Summary of research}\label{sec:sumsec}
\begin{figure}[t!]
\centering
\begin{subfigure}[t]{0.45\textwidth}
\includegraphics[width=\textwidth]{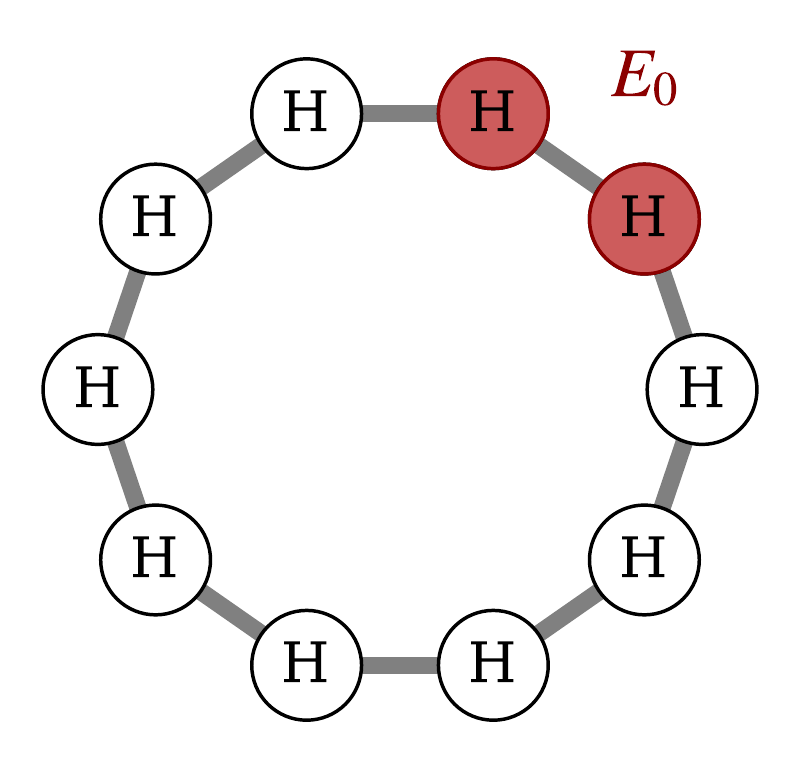}
\caption{$E_{\text{tot}} = 5E_0$}
\end{subfigure}
\hfill
\begin{subfigure}[t]{0.45\textwidth}
\includegraphics[width=\textwidth]{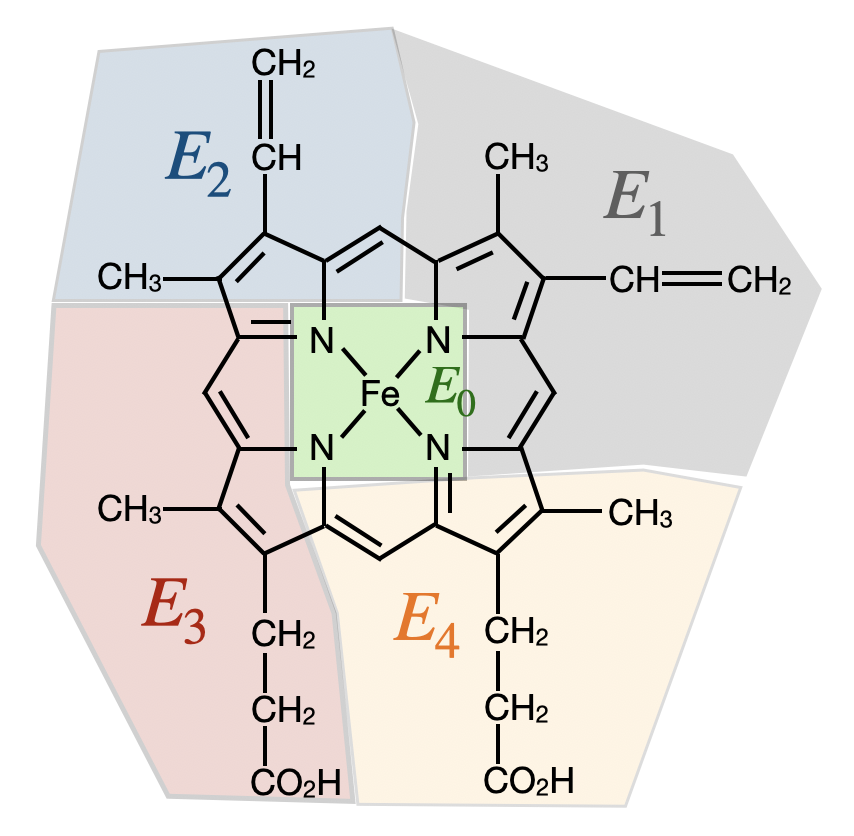}
\caption{$E_{\text{tot}} = \sum_{i=0}^4 E_i$}
\end{subfigure}
\caption{Evaluating the total energy with density matrix embedding theory.
(a) A hydrogen ring composed of $10$ atoms obeying the periodic boundary condition,
and the impurity (supercell) is two adjacent atoms. The total energy equals
the energy of the supercell times the number of supercells.
(b) A single ligand heme molecule is divided into $5$ non-overlapping
fragments, and the DMET energy of each fragment is calculated. The total energy
is the sum of energies from all fragments. 
 }\label{fig:partition_intro}
\end{figure}

This thesis provides several tools to study the finite temperature behaviors
of strongly correlated materials. First we will introduce the finite temperature density matrix embedding theory (FT-DMET) in Chapter~\ref{chp:dmet}.  
FT-DMET, as a thermal extension of ground state DMET (GS-DMET)~\citep{Knizia2012,KniziaJCTC2013,ZhengPRB2016,WoutersJCTC2016},
maps the many-body lattice thermal problem onto an impurity thermal problem.
Same as in GS-DMET, the system is divided into non-overlapping
fragments, which are defined by a set of local orbitals (LOs).
For periodic systems, the fragments are chosen as supercells and thus all 
fragments are equivalent. For systems that do not obey periodicity,
extensive observables are evaluated for each fragment and the total 
value of the observable is the summation of those from all fragments.
An illustration is shown in Fig.~\ref{fig:partition_intro} of the above two cases.
Note
that in the latter case, one should be careful when evaluating intensive
properties, which should either be defined for a specific fragment or
evaluated from global extensive properties.
This real-space partition ensures that most of the entanglement is retained 
in the fragment for systems with short correlation lengths. 
When treating one fragment, we call this fragment the \textit{impurity}
and the rest of the fragments the \textit{environment}.

To further capture
the entanglement between the impurity and the environment, we introduce 
a term called \textit{bath}. Bath in DMET is a subspace of the environment which
is directly entangled with the impurity, spanned by a set of basis called
bath orbitals. In strongly correlated systems, the correlation is highly
localized, and the entanglement entropy obeys the area law. One could imagine
that the bath orbitals mostly come from sites adjacent to the impurity.
In practice, the bath orbitals are derived from the Schmidt decomposition
of the total system wavefunction, which is initialized as the mean-field
wavefunction and optimized in a bootstrap manner. A nice property of GS-DMET
is that the number of bath orbitals generated from Schmidt decomposition
is exactly equal to the number of impurity orbitals, with the assumption
that the impurity is much smaller than the environment.

The key issue going from GS-DMET to FT-DMET is that the Schmidt decomposition
no longer works since the system cannot be described by one single 
wavefunction. In fact the finite temperature state is described by a 
density matrix of the mixed state. Remember that the Schmidt decomposition
of a wavefunction is equivalent to the singular value decomposition (SVD) of
the corresponding density matrix. In FT-DMET algorithm, we start from
the mean-field single-particle density matrix $\rho_{0}$, and apply SVD to the
impurity-environment block to generate a set of bath orbitals, as 
described in the theory part in Chapter~\ref{chp:dmet}. Note that 
since the temperature enlarges the entanglement length, one should expect
more bath orbitals to cover all impurity-environment entanglement than in
GS-DMET. To do so, we continue to apply  SVD to the impurity-environment block 
of powers of $\rho_0$ to get the rest of the bath orbitals. The algorithm is 
benchmarked with one- and two-dimensional Hubbard models, and shows
systematically improved accuracy by increasing bath or impurity size.

In Chapter~\ref{chp:hlatt}, we further extend the FT-DMET algorithm to handle
\textit{ab initio} problems. While model systems can be used to reproduce 
some of the behaviors and phases in realistic lattices, being able to perform
\textit{ab initio} simulations is key to achieve a complete understanding
of the materials. There are two technical differences between model systems
and \textit{ab initio} systems: (ii) in most of the model systems, 
site basis is used which is naturally localized, while in \textit{ab initio} 
systems the Coulomb interaction is of long range and the basis set used
is usually not localized; (ii) in model systems, the two-body interaction
form is very simple, while the two-body interaction in an \textit{ab initio} Hamiltonian
is described by a complicated rank-$4$ matrix. The above two technical difficulties
are universal for all \textit{ab initio} simulations. For \textit{ab initio} 
FT-DMET, one also needs to deal with the large embedding space due to the size
of the supercell and the basis set, which requires necessary truncation
to the bath space. Moreover, a finite temperature impurity solver that
can handle \textit{ab initio} Hamiltonian efficiently is also crucial 
for any meaningful simulations. In Chapter~\ref{chp:hlatt}, we provide
solutions to the above problems and present the \textit{ab initio} FT-DMET
algorithm. We further use this algorithm to explore properties and
 phase transitions of hydrogen lattices.

Chapter~\ref{chp:dmet} and Chapter~\ref{chp:hlatt} present an efficient
numerical tool to simulate both strongly-correlated model systems and
\textit{ab initio} systems. The next question to answer is what
order parameters we can use to capture essential thermal properties 
and phase transitions at finite temperature. In Chapter~\ref{chp:cp},
we will study one of the most common but complex phase transitions:
metal-insulator transition (MIT). We argue that compared to the band
structure theory which is widely used to distinguish metal from insulator, 
electron locality is a more universal criteria which can be used to detect
finite temperature MIT. We further introduce an order parameter named
complex polarization to measure the locality of electrons and provide
a thermofield approach to evaluate finite temperature complex polarization.
The finite temperature complex polarization formulation provides an easy 
but well-defined way to characterize MIT in any periodic materials.

In Chapter~\ref{chp:qite}, several quantum algorithms will be introduced
for both ground state and finite temperature simulations on quantum devices.
With the development of quantum computing technology, especially the hardware,
it can be foreseen that certain categories of difficult problems in classical
simulations can be solved with less effort on a quantum device. 
The bridge to connect chemical problems and successful quantum simulations 
is efficient quantum algorithms for noisy intermediate-scale quantum (NISQ) devices.
Several quantum algorithms have been developed
to carry out quantum chemical simulations in the past decades, including
quantum phase estimation (QPE)~\citep{Farhi_MIT_2000,Kitaev_arxiv_1995} and hybrid quantum-classical variational
algorithms such as quantum approximate optimization algorithm (QAOA)~\citep{Farhi_MIT_2014,Otterbach_arxiv_2017,Moll_QST_2018}
and variational quantum eigensolver (VQE)~\citep{Peruzzo_Nature_2013,McClean_NJP_2016,grimsley2018adapt}. While the above algorithms
have many advantages as advertised, they all require quantum or classical
resources that can easily exceed the capacity of current devices.
In Chapter~\ref{chp:qite}, the key quantum algorithm that will be introduced
is called quantum imaginary time evolution (QITE). As mentioned in 
Section~\ref{sec:ftalgos}, imaginary time evolution
is an efficient algorithm to find the ground state. If the initial state is
the identity density matrix at infinite temperature, then one could evaluate 
the density matrix and thus the thermal observables at any temperature. 

The conflict of implementing imaginary time evolution on a quantum device is that
the imaginary evolution operator $e^{\beta \hat{H}}$ is a non-unitary operator,
while only unitary operators are allowed on a quantum device. We present an 
approach to reproduce a non-unitary operator with a rescaled unitary 
operation on an enlarged domain. This approach could be flexibly performed
both exactly and approximately, depending on the computational resources 
available. The result is systematically improved and converges rapidly
 by increasing the size of the unitary domain. The convergence to the ground
state can be further accelerated by the quantum Lanczos algorithm (QLanczos).
QLanczos constructs a Krylov subspace with the intermediate states in 
QITE simulation, and then diagonalizes the Hamiltonian in the subspace 
representation to get a better approximation of the ground state. Unlike the
classical Lanczos algorithm mentioned in Section~\ref{sec:ftalgos} where
the Krylov subspace is spanned by $\{|\psi_0\rangle, H|\psi_0\rangle, ..., 
H^m|\psi_0\rangle\}$, the Krylov space in QLanczos is spanned by
$\{|\psi_0\rangle, e^{2\tau \hat{H}}|\psi_0\rangle, ..., e^{2m\tau \hat{H}}|\psi_0\rangle\}$. The Hamiltonian in the quantum Krylov space can be 
collected from the energy measurement at each step for free and no additional measurement is needed. 

The third algorithm introduced in Chapter~\ref{chp:qite} is the quantum
minimally entangled typical thermal states (QMETTS) algorithm. While the
first two algorithms (QITE and QLanczos) can be applied to both 
ground state and finite temperature calculations, QMETTS is designed 
in particular for finite temperature simulations. QMETTS samples a set
of minimally entangled thermal states under the thermal statistics by
a repeated imaginary time evolving and then collapsing onto the product states 
routine. The advantage
of the QMETTS algorithm is that the imaginary time evolution (fulfilled by
QITE) always starts from a product state, so that the entanglement will not
grow too large even at very low temperature. We present both
 classical and quantum simulations on a variety of problems using the above 
three quantum algorithms as examples and tests.

\chapter{Finite temperature density matrix embedding theory\label{chp:dmet}}
\section{Abstract}
We describe a formulation of the density matrix embedding theory
at finite temperature. We present a generalization of the
ground-state bath orbital construction that embeds a mean-field finite-temperature density matrix up to a given order in the Hamiltonian, or the Hamiltonian
up to a given order in the density matrix. We assess the performance
of the finite-temperature density matrix embedding on the 1D Hubbard model both at
half-filling and away from it, and the 2D Hubbard model at half-filling,
comparing to exact data where available, as well as results from finite-temperature
density matrix renormalization group, 
 dynamical mean-field theory,
 and dynamical cluster approximations. 
The accuracy of finite-temperature
density matrix embedding appears comparable to that of the ground-state theory, with at most a modest increase in bath size,
and competitive with that of cluster dynamical mean-field theory.

\section{\label{sec:intro_dmet}Introduction}
The numerical simulation of strongly correlated electrons is key
to understanding the quantum phases that derive from
electron interactions, ranging from the Mott transition~\citep{MottRMP1968,BullaPRL1999,BelitzRMP1994,QazilbashScience2007} to high 
temperature superconductivity~\citep{AndersonScience1987,LakeScience2001,LakeNature2002}. Consequently, many numerical methods have been developed for this task. In the setting of quantum lattice models, quantum
embedding methods~\citep{SunACR2016}, such as dynamical mean-field theory (DMFT)\cite{KotliarRMP2006,GeorgesRMP1996,LichtensteinPRL2001,LichtensteinPRB2000,ZgidJCP2011} and density matrix embedding theory (DMET)\cite{Knizia2012,KniziaJCTC2013,WoutersJCTC2016,ZhengPRB2016,ZhengScience2017,BulikPRB2014,BulikJCP2014},
have proven useful in obtaining insights into complicated quantum phase diagrams.
These methods are based on an approximate mapping from the full interacting quantum lattice to a simpler
self-consistent quantum impurity problem, consisting of a few sites of the original lattice
coupled to an explicit or implicit bath. In this way, they avoid
treating an interacting quantum many-body problem in the thermodynamic limit.

The current work is concerned with the extension of DMET to finite temperatures.
DMET so far has mainly been applied in its ground-state formulation (GS-DMET), where it has achieved some success, particularly in applications to quantum phases where the order is associated with large unit cells~\citep{ZhengPRB2016,ZhengScience2017,ChenPRB2014}.
The ability to treat large unit cells at relatively low cost compared to other quantum embedding methods is due to the
computational formulation of DMET, which is based on modeling 
the ground-state impurity density matrix, a time-independent quantity
accessible to a wide variety of efficient quantum many-body methods.
Our formulation of finite-temperature DMET (FT-DMET) is based on the simple structure of GS-DMET,
but includes the possibility to generalize the bath so as to better capture the finite-temperature impurity density matrix.
Bath generalizations have previously been used to extend GS-DMET to the calculation
of spectral functions and other dynamical quantities~\citep{BoothPRB2015,fertitta2019energy}. Analogously to GS-DMET, since one only needs to
compute time-independent observables,
finite-temperature DMET can be paired with the wide variety of quantum impurity solvers which can provide the finite-temperature
density matrix.

We describe the theory of FT-DMET in  Section \ref{sec:theory_dmet}. In Section \ref{sec_results_dmet} we
carry out numerical calculations on the 1D and 2D Hubbard models, using exact diagonalization (ED) and the finite-temperature density matrix
renormalization group (FT-DMRG)~\citep{FeiguinPRB2005} as quantum impurity solvers. We benchmark our results against those from the Bethe ansatz in 1D, and
DMFT and the dynamical cluster approximation (DCA) in 2D, and also explore the quantum impurity derived N\'eel transition in the 2D Hubbard model.
We finish with brief conclusions about prospects for the method in \ref{sec:conc_dmet}.

\section{\label{sec:theory_dmet}Theory}

\subsection{Ground state DMET}\label{sec:gsdmet}

In this work, we exclusively discuss DMET in lattice models (rather than for 
\emph{ab initio} simulations~\citep{WoutersJCTC2016,KniziaJCTC2013,BulikJCP2014,cui2019efficient}).
As an example of a lattice Hamiltonian, and one that we will use in numerical simulations,
we define the Hubbard model~\citep{HubbardPRS1963,GutzwillerPRL1963}, 
\begin{equation}\label{eq:theory-hubham}
\hat{H} = -t\sum_{\langle i,j\rangle,\sigma} \hat{a}^{\dagger}_{i\sigma}
\hat{a}_{j\sigma} - \mu \sum_{i,\sigma} \hat{a}^{\dagger}_{i\sigma}
\hat{a}_{i\sigma} + U\sum_{i}\hat{n}_{i\uparrow}\hat{n}_{i\downarrow}
\end{equation}
where $\hat{a}^{\dagger}_{i\sigma}$ creates an electron with spin $\sigma$
on site $i$ and $\hat{a}_{i\sigma}$ annihilates it;
$\hat{n}_{i\sigma} = \hat{a}^{\dagger}_{i\sigma}\hat{a}_{i\sigma}$;
$t$ is the nearest-neighbour (denoted $\langle i,j\rangle$) hopping amplitude, here set to $1$;
$\mu$ is a chemical potential; and $U$ is the on-site repulsion.

The general idea behind a quantum embedding method such as DMET
is to approximately solve the interacting problem in the large lattice by dividing the lattice into small
fragments or impurities~\citep{SunACR2016}. (Here we will assume that the impurities are non-overlapping).
The main question is how to treat the coupling and entanglement between the impurities.
In DMET, other fragments around a given impurity are modeled by a set of bath orbitals.
The bath orbitals are constructed to exactly reproduce the entanglement between
the impurity and environment when the full lattice is treated at a mean-field level (the so-called ``low-level'' theory).
The impurity together with its bath orbitals then constitutes a small embedded quantum impurity problem, 
which can be solved with a ``high-level'' many-body method. The low-level lattice wavefunction 
 and the high-level embedded impurity wavefunction are made approximately
consistent, by enforcing self-consistency of the single-particle density matrices
of the impurities and of the lattice. This constraint is implemented by introducing a static correlation potential on
the impurity sites into the low-level theory. 
{The correlation potential introduced in DMET is analogous to the
DMFT self-energy. A detailed discussion of the correlation potential 
including the comparison to other approaches such as density functional 
theory (DFT) can be found in ~\citep{SunACR2016,KniziaJCTC2013}.}

To set the stage for the finite-temperature theory, in the following we briefly
recapitulate some details of the above steps in the GS-DMET formulation. In particular, we discuss
how to extract the bath orbitals, how to construct the embedding Hamiltonian, and how to carry out the
self-consistency between the low-level and high-level methods. 
Additional details for the GS-DMET algorithm can be found in several articles~~\citep{Knizia2012,ZhengPRB2016,WoutersJCTC2016}, including the review in Ref.~\citep{WoutersJCTC2016}.

\subsubsection{DMET bath construction}
Given a  full lattice of $L$ sites, we define the impurity $x$ over $L_x$ sites, the Hilbert space of which 
is denoted as $\mathcal{A}^x$ and spanned by a set of orthonormal basis
$\{|A^x_i\rangle\}$. The rest of the lattice is treated as the environment of
impurity $x$, the Hilbert space of which is denoted as $\mathcal{E}^x$
spanned by an orthonormal basis $\{|E^x_i\rangle\}$. The Hilbert space of 
the entire lattice $\mathcal{H}$ is the direct product of the two 
subsystem Hilbert spaces: $\mathcal{H} = \mathcal{A}^x\otimes \mathcal{E}^x$.
Any state $|\Psi\rangle$ in $\mathcal{H}$ can be written as
\begin{equation}\label{eq:bipartite_dmet}
|\Psi\rangle = \sum_{ij} \psi_{ij} |A^x_i\rangle |E^x_j\rangle,
\end{equation}
where the coefficients $\psi_{ij}$ form a $2^{n_A}\times 2^{n_E}$ matrix. 
Absorbing $\psi_{ij}$ into the environment orbitals, one could rewrite 
Eq.~\eqref{eq:bipartite_dmet} as
\begin{equation}\label{eq:schmidt_dmet}
\begin{split}
|\Psi\rangle &= \sum_i  |A^x_i\rangle \left(\sum_j \psi_{ij} |E^x_j\rangle\right) \\
& = \sum_i |A^x_i\rangle |B^x_i\rangle
\end{split},
\end{equation}
where $|B^x_i\rangle = \sum_j \psi_{ij} |E^x_j\rangle$. 
Eq.~\eqref{eq:schmidt_dmet} tells us that the orbitals in $\mathcal{E}^x$ 
that are entangled to the impurity $x$ are of the same size as the impurity
orbitals. Note that $\{|B^x_i\rangle\}$ are not orthonormal and the 
rest of the environment enters as a separatable product state 
$|\Psi_{\text{core}}\rangle$
called "core contribution". Let $\{|\tilde{B}^x_i\rangle\}$ denote the 
orthonormal states derived from $\{|B^x_i\rangle\}$, then
Eq.~\eqref{eq:schmidt_dmet} can be rewritten as
\begin{equation}\label{eq:emb_core_dmet}
|\Psi\rangle = \left(\sum_i \lambda_i |A^x_i\rangle |\tilde{B}^x_i\rangle\right)
|\Psi_{\text{core}}\rangle.
\end{equation}
The orbitals $\{|\tilde{B}^x_i\rangle\}$ are directly entangled with the 
impurity $x$, and thus are called \textit{bath orbitals}. The space spanned
by impurity and bath is called \textit{embedding space}. One can then 
derive the embedding state as
\begin{equation}\label{eq:emb_wf_dmet}
|\Psi_{\text{emb}}\rangle = \sum_i \lambda_i |A^x_i\rangle |\tilde{B}^x_i\rangle.
\end{equation}
If $|\Psi\rangle$ is an eigenstate of the Hamiltonian $\hat{H}$ in the full lattice,
then one can prove that $|\Psi_{\text{emb}}\rangle$ is also an eigenstate
of the embedding Hamiltonian $\hat{H}_{\text{emb}}$ defined as the projection
of $\hat{H}$ onto the embedding space. The two eigenvalues are identical. 
Therefore, the full lattice problem can be reduced to a smaller embedding
problem.

In practice, the exact bath orbitals are unknown since the many-body
eigenstate $|\Psi\rangle$ 
is the final target of the calculation. Instead, we construct a set of 
approximated bath orbitals from a mean-field ("low-level") wavefunction 
$|\Phi\rangle$, which is an eigenstate of a quadratic lattice Hamiltonian $\hat{h}$.
We rewrite $|\Phi\rangle$ according to Eq.~\eqref{eq:emb_core_dmet} and
Eq.~\eqref{eq:emb_wf_dmet} in the form
\begin{align}
  \label{eq:theory-mfwf}
  |\Phi\rangle = |\Phi_\text{emb}\rangle |\Phi_\text{core}\rangle.
  \end{align}

The single-particle density matrix $D^\Phi$ obtained from $|\Phi\rangle$
contains all information on the correlations in $|\Phi\rangle$.
Thus the bath orbitals can be defined from this density matrix.

We consider the impurity-environment block $D^{\Phi}_{\text{imp-env}}$ ($D_{ij}$ for $i \in x, j \notin x$) of dimension $L^{x} \times (L-L^x)$.
Then taking the thin SVD
\begin{equation}
D^{\Phi}_{\text{imp-env}} = U\lambda B^{\dagger},
\end{equation}
the columns of $B$ specify the bath orbitals in the lattice basis.
The bath space is thus a function of the density matrix, denoted $B(D)$.

\subsubsection{Embedding Hamiltonian}

After obtaining the bath orbitals, we construct the embedded Hamiltonian of the quantum impurity problem. 
In GS-DMET, there are two ways to do so: the
interacting bath formulation and the non-interacting bath formulation. The conceptually simplest approach
is the interacting bath formulation. In this case, we project
the interacting lattice Hamiltonian $\hat{H}$ into the space of
the impurity plus bath orbitals, defined by the projector $\hat{P}$, i.e. the embedded Hamiltonian
is  $\hat{H}_\text{emb} = \hat{P}\hat{H}\hat{P}$. 
$\hat{H}_\text{emb}$  in general contains non-local interactions involving the bath orbitals, as they are non-local
orbitals in the environment. From the embedded Hamiltonian, we compute the high-level
ground-state impurity wavefunction,
\begin{align}
  \hat{H}_\text{emb} |\Psi\rangle = E |\Psi\rangle.
  \end{align}
If $\hat{H}$ were itself the quadratic lattice Hamiltonian $\hat{h}$, then
then $\Psi = \Phi$ and 
\begin{equation}\label{eq:theory-samegs}
\hat{P}\hat{h}\hat{P} |\Phi\rangle = E |\Phi\rangle.
\end{equation}
Another way to write Eq.~(\ref{eq:theory-samegs}) for a mean-field state is
\begin{align}\label{eq:theory-samegs2}
  [{P}{h}{P}, {P}{D}^\Phi{P}] = 0,
  \end{align}
where $h$ denotes the single-particle Hamiltonian matrix and $P$ is the single-particle
projector into the impurity and bath orbitals. 
These conditions imply that the  lattice Hamiltonian and the embedded Hamiltonian $\hat{H}_\text{emb}$ share
the same ground-state at the mean-field level, which is the basic approximation in GS-DMET.

In the alternative non-interacting bath formulation, interactions on the bath are
approximated by a quadratic correlation potential (discussed below). This formulation retains the same exact embedding property as the interacting
bath formulation for a quadratic Hamiltonian.
In practice,  both formulations give similar results in the Hubbard model~\citep{BulikPRB2014,WuJCP2019}, and the choice between the two
depends on the available impurity solvers; the interacting bath formulation generates non-local two-particle interactions in the bath
that not all numerical implementations can handle.
In this work, we use the interacting bath formulation in the 1D Hubbard model where an ED solver is used.
In the 2D Hubbard model, we use the non-interacting bath formulation, where both ED and FT-DMRG solvers are used.
This latter choice is because the cost of treating non-local interactions in FT-DMRG is relatively high (and we
make the same choice with ED solvers to keep the results strictly comparable).

\subsubsection{Self-consistency}

To maintain self-consistency between the ground-state of the lattice mean-field $|\Phi\rangle$,  
and that of the interacting embedded Hamiltonian $|\Psi\rangle$, we introduce
a quadratic correlation potential $\hat{u}$ into $h$, i.e.
\begin{align}
  \hat{h} \to \hat{h} + \hat{u},
  \end{align}
where $\hat{u}$ is constrained to act on sites in the impurities, i.e. $\hat{u} = \sum_x \hat{u}^x$. To study magnetic order, we choose
the form
\begin{align}
  \hat{u}^x = \sum_{ij \in x, \sigma \in \{ \uparrow,\downarrow\}} u^x_{i j\sigma} a^\dag_{i\sigma} a_{j\sigma}.
\end{align}
The coefficients $u^x_{ij\sigma}$ are adjusted to match the density
matrices on the impurity that are evaluated from the low-level wavefunction $|\Phi\rangle$
and from the high-level embedded wavefunction $|\Psi\rangle$. In this work, we
only match the single-particle density matrix elements of the impurity (impurity-only matching~\citep{WoutersJCTC2016}) by minimizing the cost function:
\begin{equation}\label{eq:cost_func_dmet}
f(u) = \sum_{i,j\in \text{imp}}(D_{ij}^{\text{low}} - D_{ij}^{\text{high}})^2,
\end{equation}
where $D^{\text{low}}$ and $D^{\text{high}}$ are single-particle density
matrices of low-level and high-level solutions, respectively. For each
minimization iteration, we assume that the high-level single-particle density
matrix is fixed, and the gradient of Eq.~\eqref{eq:cost_func_dmet} is
\begin{equation}\label{eq:gradient_cost_dmet}
\frac{\mrm{d}f}{\mrm{d}u} = \sum_{i,j\in \text{imp}}2(D_{ij}^{\text{low}} - D_{ij}^{\text{high}}) 
\frac{\mrm{d}D_{ij}^{\text{low}} }{\mrm{d}u}.
\end{equation}
For ground state, Ref.~\citep{WoutersJCTC2016} provided an analytical 
approach to evaluate $\frac{\mrm{d}D_{ij}^{\text{low}} }{\mrm{d}u}$ using 
the first order perturbation theory. At finite temperature, one could also
evaluate the gradient analytically as shown in the Appendix.

Note also that we will only be considering translationally invariant systems, and thus $\hat{u}^x$ is the same
for all impurities.

\subsection{Ground-state expectation values}

\label{sec:gsexpect} 
Ground-state expectation values are evaluated from the density matrices of each high-level impurity wavefunctions
$|\Psi^x\rangle$. Since there are multiple impurities (in a translationally invariant system, these
are constrained to be identical), an expectation value is typically assembled from
the multiple impurity wavefunctions using a democratic 
partitioning~\citep{WoutersJCTC2016}. For example, given two sites $i$, $j$, 
where $i$ is part of impurity $x$ and $j$ is part of impurity $y$, 
the expectation value of $a^\dag_i a_j$ is
\begin{align}
  \langle a^\dag_i a_j \rangle = \frac{1}{2}[\langle \Psi^x | a^\dag_i a_j | \Psi^x\rangle +\langle \Psi^y | a^\dag_i a_j | \Psi^y\rangle].
\end{align}
Note that the {\textit pure bath} components of the high-level wavefunctions, e.g. $\langle \Psi^x | a^\dag_i a_j |\Psi^x\rangle$ for $i, j \notin x$
{are not used in defining the DMET expectation values.
Instead, the democratic partitioning is arranged such that an individual impurity embedding contributes the correct amount to a global expectation value so long as the impurity wavefunction produces
correct expectation values for operators that act on the impurity alone, or the impurity and bath together.}

\subsection{Finite temperature DMET}\label{sec:ftdmet}

Our formulation of FT-DMET follows the same rubric as the ground-state theory: a low-level
(mean-field-like) finite-temperature
density matrix is defined for the lattice; this is used to
obtain a set of bath orbitals to define the impurity problem;
a high-level finite-temperature
density matrix is calculated for the embedded impurity; and
self-consistency is carried out between the two via a correlation potential.
The primary difference lies in the properties of the bath, which we focus on below,
as well as in the appearance of quantities such as the entropy, which are
formally defined from many-particle expectation values.

\subsubsection{Finite temperature bath construction}
In GS-DMET bath construction, the bath orbitals are directly defined
from Schmidt decomposition of the full lattice ground state wavefunction as 
in Eq.~\eqref{eq:emb_core_dmet}. However, at finite temperature, the state
of an open quantum system (grand canonical ensemble) is described by a 
mixed state: the density matrix is described by a linear combination of
pure state density matrices. As a consequence, the Schmidt decomposition
can no longer be used to define bath orbitals. In fact, with non-zero temperature,
the entanglement becomes more delocalized. To capture the entanglement between
the impurity and environment, a larger bath space is needed compared to 
that of ground state. In Fig.~\ref{fig:bath_occ}, we plotted the weight of 
entanglement with the impurity as a function of distance (in sites) from 
the impurity for a $100$-site tight binding model ($\hat{H} = \sum_{i}
\hat{a}_i^\dag\hat{a}_{i+1} + \text{h.c.}$). One could see that as the 
temperature rises, more and more farther sites are entangled with the
impurity, and eventually all sites are uniformly and maximumly entangled. 
At ground state
($T = 0$), the weight decreased with distance with an oscillating manner, with
wavelength = $2$ sites; at $T = 0.15$, the wavelength increased to $6$
sites due to the smearing effect of finite temperature. The increase of 
oscillating wavelength is another example
of the increase of correlation length with temperature. 
\begin{figure}
    \centering
    \includegraphics[width=1\textwidth]{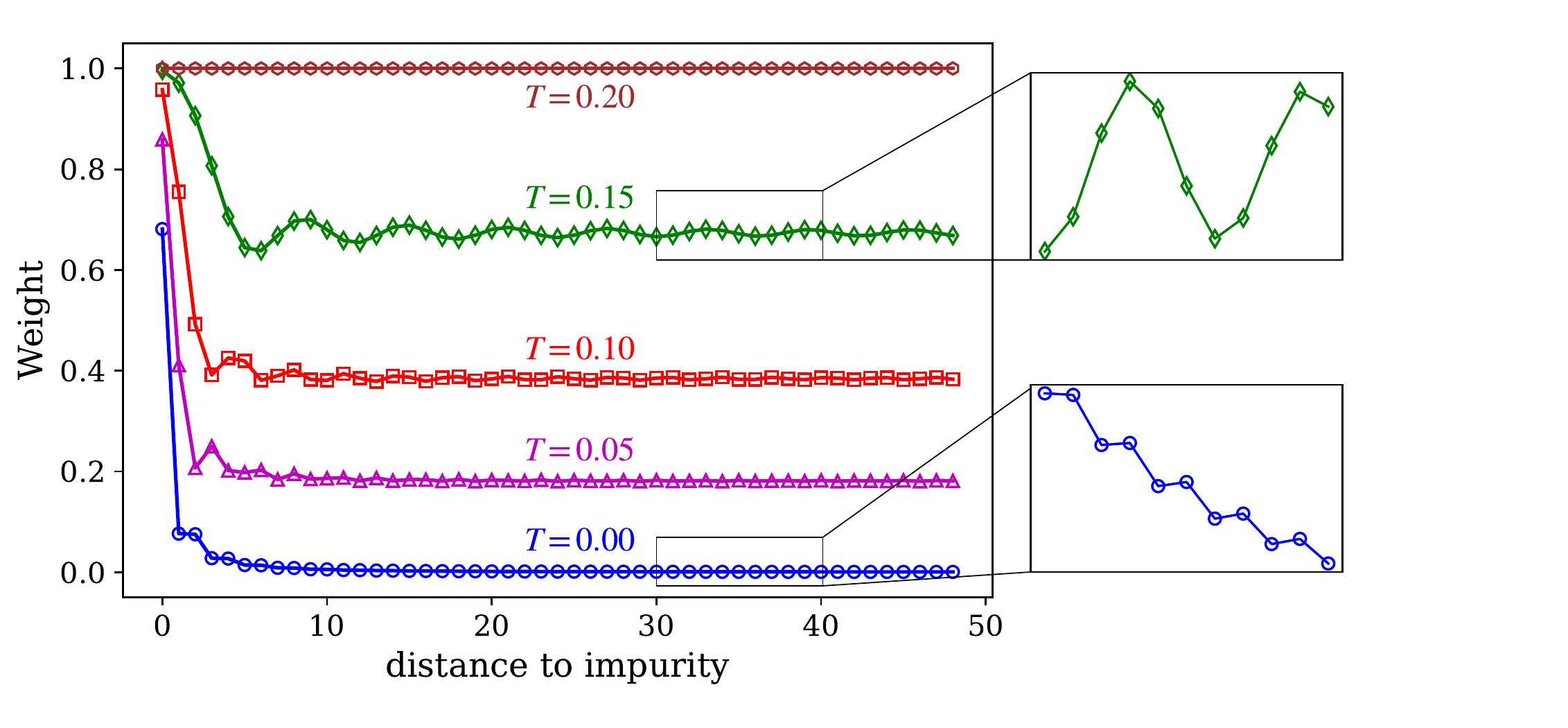}
    \caption{Weight of the entanglement with the impurity on environmental sites.
    The weights are evaluated as the square norm of projections of 
    environmental sites to the bath space. Due to periodicity of the system,
    only half of the environmental sites are shown in the figure.
    } \label{fig:bath_occ}
\end{figure}

The difficulty of finite temperature bath orbital construction can also be
demonstrated by the commutation relation between the projected 
single-particle density matrix and projected Hamiltonian.
In GS-DMET, the bath orbital construction is designed to be exact
if all interactions are treated at the mean-field level, giving rise
to the commuting condition for the projected single-particle density matrix
and projected Hamiltonian in Eq.~(\ref{eq:theory-samegs2}).
At finite temperature, the above commuting condition does not stand and  
one should expect approximated bath orbitals even at mean-field level.
In general, we can look for a finite-temperature bath construction that preserves a similar property.
As pointed out in Sec.~\ref{sec:gsexpect}, the DMET embedding is still exact
for single-particle expectation values if the embedded projected single-particle
density matrix produces the correct expectation values in the impurity and impurity-bath
sectors, due to the use of the democratic partitioning. We aim to satisfy this slightly relaxed condition.

The finite temperature single-particle density matrix of a quadratic Hamiltonian $\hat{h}$
is given by the Fermi-Dirac function 
\begin{equation}\label{eq:theory-fdfull}
{D}(\beta) = \frac{1}{1 + e^{({h}-\mu)\beta}},
\end{equation}
where $\beta = 1/k_B T$ ($k_B$ is the Boltzmann constant, $T$ is the temperature). In the following, we fix $k_B = 1$, thus
$\beta = 1/T$. If we could find an embedding directly analogous to the ground-state construction, we would obtain a projector ${P}$, such that
the embedded density matrix ${P} {D}{P}$ is the Fermi-Dirac function of the embedded quadratic Hamiltonian, i.e.
${P} h {P}$, i.e.
\begin{equation}\label{eq:theory-fdemb}
{P} {D} {P} = \frac{1}{1 + e^{({P}{h}{P}-\mu)\beta}}.
\end{equation}
However, unlike in the ground-state theory, the non-linearity of the exponential function means that
Eq.~(\ref{eq:theory-fdemb}) can only be satisfied 
exactly if ${P}$ projects back into the full lattice basis. Thus a bath orbital construction at
finite temperature is necessarily always approximate, even for quadratic Hamiltonians.

Nonetheless, one can choose the bath orbitals to reduce the error between the l.h.s. and r.h.s. in
Eq.~(\ref{eq:theory-fdemb}). First, we require that the equality is satisfied only
for the impurity-environment block of $D$, following the relaxed requirements
of the democratic partitioning.  Second, we require the equality to
be satisfied only up to a finite order $n$ in $h$, i.e.
\begin{align}
\label{eq:fdemb_nth}
[P D P ]_{ij} = \left[\frac{1}{1 + e^{({P}{h}{P}-\mu)\beta}}\right]_{ij} + O(h^n) \quad i \in x,j \notin x.
\end{align}
Then there is a simple algebraic construction of the bath space  as (see Appendix for a proof)
\begin{align}
\label{eq:hbath}
\span\{B(h) \oplus B(h^2) \oplus B(h^3) \ldots B(h^n) \},
\end{align}
where $B(h^k)$ is the bath space derived from $h^k$, $k = 1, ..., n$. Note that each order of $h$ adds $L_x$ bath orbitals to the total
impurity plus bath space.

We can alternatively choose the bath to preserve the inverse relationship between the density matrix and Hamiltonian,
\begin{align}
[P h P ]_{ij} = \mathrm{inverseFD}(P D P) + O(D^n) \quad \mathrm{not}\ i,j \notin x,
\end{align}
where $\mathrm{inverseFD}$ is the inverse Fermi-Dirac function, and the bath space is
then given as
\begin{align}
\label{eq:dbath}
\span\{B(D) \oplus B(D^2) \oplus B(D^3) \ldots B(D^n) \}.
\end{align}
The attraction of this construction is that the lowest order corresponds to the standard GS-DMET bath construction.

The above generalized bath constructions allow for the introduction of an unlimited number of bath sites (so long as the total number
of sites in the embedded problem is less than the lattice size). Increasing the size of the embedded problem by increasing the number of bath
 orbitals (hopefully) increases the accuracy of the embedding, but it also
 increases the computational cost. However, an alternative way to increase
accuracy is simply to increase the number of impurity sites. Which strategy is better is problem dependent,
and we will assess both in our numerical experiments.

\subsubsection{Thermal observables}

The thermal expectation value of an observable $\hat{O}$ is defined as
\begin{equation}\label{eq:theory-ftexpval}
\langle \hat{O}(\beta)\rangle = \Tr\left[\hat{\rho}(\beta)\hat{O}\right].
\end{equation}
Once $\hat{\rho}(\beta)$ is obtained from
 the high-level impurity calculation, for observables based on low-rank (e.g. one- and two-) particle
reduced density matrices, we evaluate Eq.~(\ref{eq:theory-ftexpval}) using the democratic partitioning formula for
expectation values in Sec.~\ref{sec:gsexpect}. 

We will also, however, be interested in the entropy per site, which is a many-particle expectation value.
Rather than computing this directly as an expectation value, 
we will obtain it by using the thermodynamic relation
$\frac{\dd S}{\dd E} = \beta$, and
\begin{equation} 
S(\beta_0) = S(0) + \int_{E(0)}^{E(\beta_0)} \beta(E) \dd E 
\end{equation}
 where $\beta_0$ is the desired inverse temperature, and $S(0) = \ln 4 \approx 1.386$.

\section{\label{sec_results_dmet}Results}

\subsection{Computational details}

We benchmarked the performance of FT-DMET in the 1D and 
2D Hubbard models as a function of $U$ and $\beta$. For the 1D Hubbard model, we compared our FT-DMET results
to exact solutions from the thermal Bethe ansatz ~\citep{TakahashiPRB2002}.
For the 2D Hubbard model, the FT-DMET results were compared to DCA and 
 DMFT
 results~\citep{MF_DMFT2d,maier2005,kunes,jarrellPRB,jarrellEPL,LeBlancPRX2015}. We used large DMET mean-field lattices with periodic boundary conditions  (240 sites in 1D, $24 \times 24$ sites in 2D).
We used exact diagonalization (ED) and finite temperature DMRG (FT-DMRG)
as impurity solvers. 
There are two common ways to carry out FT-DMRG calculations: the purification method~\citep{FeiguinPRB2005} and the
minimally entangled typical thermal states (METTS) method~\citep{StoudenmireNJP2010}.
In this work, we used the purification method implemented with the ITensor
package~\citep{ITensor} as the FT-DMRG impurity solver, as well as to provide the finite lattice reference data in Fig.\ref{fig:1ddopping}.
In the FT-DMRG solver, the sites were ordered with the impurity sites coming first, followed by the bath sites (an orthonormal
  basis for the set of bath sites of different orders was constructed via singular value decomposition, and ordered in
  decreasing size of the singular values) and the ancillae arranged in between each physical site.
In the 1D Hubbard model, we used ED exclusively and the interacting bath formulation of DMET, while in the 2D Hubbard model,
we used ED for the 4 impurity, 4 bath
calculations, and FT-DMRG for the 4 impurity, 8 bath calculations, both within the non-interacting bath formulation.
FT-DMRG was carried out using 4th order Runge-Kutta time evolution. 
To denote different calculations with different numbers of impurity and bath orbitals, we use the notation $InBm$, where $n$
denotes the number of impurity sites and $m$ the number of bath orbitals.

\subsection{1D Hubbard model}\label{sec:1dhub}


The 1D Hubbard model is an ideal test system for FT-DMET as its thermal properties
can be exactly computed via the thermal Bethe ansatz. We thus use it to assess various choices
within the FT-DMET formalism outlined above.

We first compare the relative performance of the two proposed bath constructions, generated via the Hamiltonian in Eq.~(\ref{eq:hbath}) or
via the density matrix in Eq.~(\ref{eq:dbath}). In Fig.~\ref{fig:hbath-vs-dbath}, we show the
error in the energy per site (measured from the thermal Bethe ansatz) for $U=2, 4$ and half-filling for these two choices. (The behaviour
for other $U$ is similar).
Using 4 bath sites, the absolute error in the energy is comparable to that of the ground-state calculation (which uses 2 bath sites)
over the entire temperature range.
Although the Hamiltonian
construction was motivated by the high temperature expansion of the density matrix, the density matrix construction appears to perform well
 at both low and high temperatures.
Consequently, we use the density matrix derived bath in the subsequent calculations.


\begin{figure}
\centering
\justify
\includegraphics[width=1\textwidth]{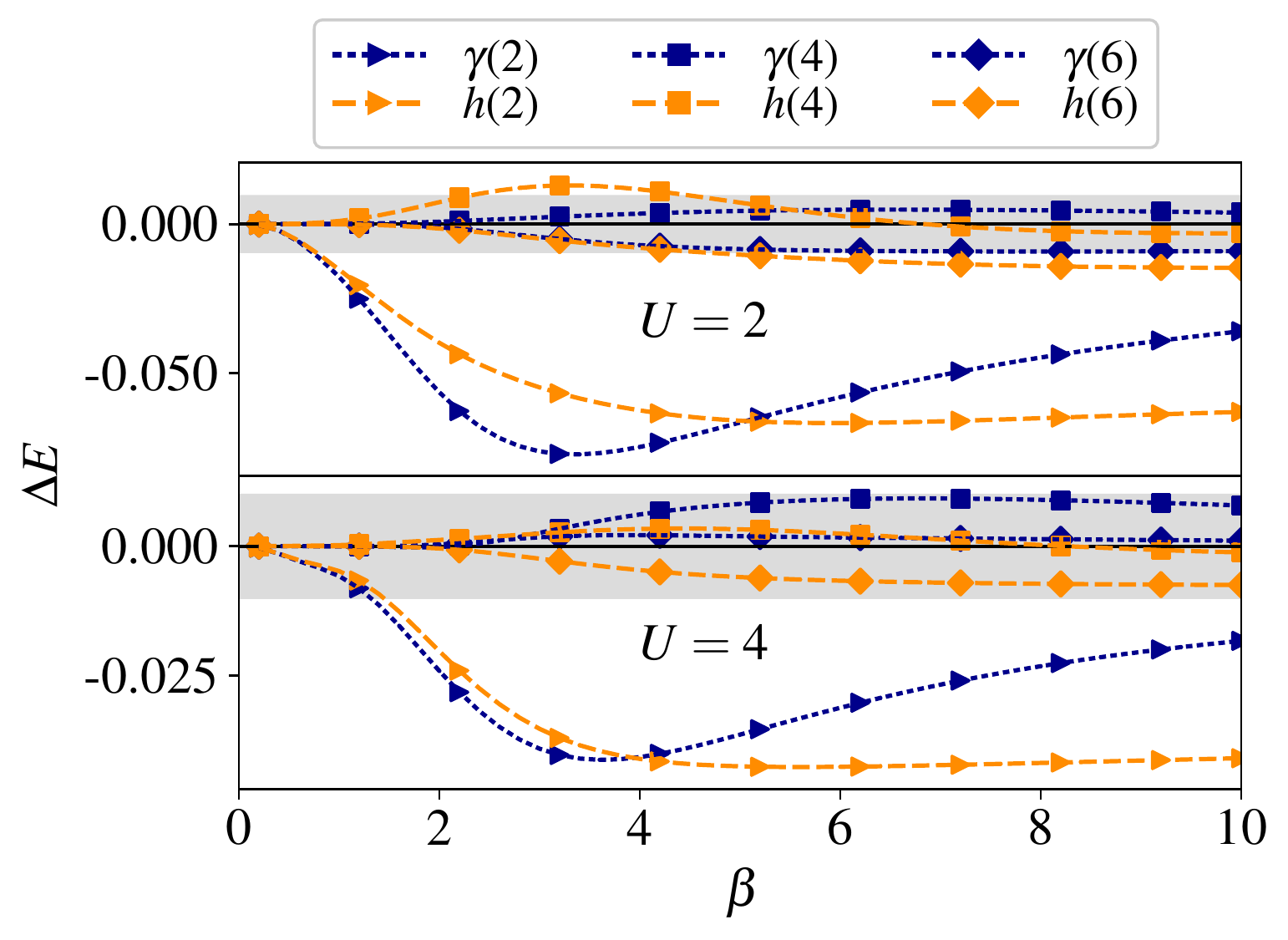}
\caption{Error in energy per site (units of $t$)  of FT-DMET 
     for the 1D Hubbard model at $U=2$ and $U=4$ (2 impurity sites and 
     half-filling) with bath orbitals generated via the 
       density matrix $\gamma$ (Eq.~(\ref{eq:dbath})) (blue lines) 
     or lattice Hamiltonian $h$ (Eq.~(\ref{eq:hbath})) (orange lines)
     as a function of inverse temperature $\beta$.  The numbers in 
    parentheses denote the number of bath orbitals. 
    The grey area denotes the ground state error with 2 impurity orbitals.
    }\label{fig:hbath-vs-dbath}
\index{figures}
\end{figure}

We next examine the effectiveness of the density matrix bath construction in removing
the finite size error of the impurity. As a first test, 
in Fig.~\ref{fig:dmet-vs-ed} we compare the energy error obtained with FT-DMET and  $I2B2$ with
a pure ED calculation with 4 impurity sites ($I4$) and periodic (PBC) or antiperiodic (APBC) boundary
conditions, at various $U$ and $\beta$. For weak ($U=2$) to moderate ($U=4$) coupling, FT-DMET shows a significant improvement
over a finite system calculation with the same number of sites, reducing the error by a factor of $\sim 2-6$ depending on the $\beta$, thus
demonstrating the effectiveness of the bath. The maximum FT-DMET energy error is 8.1, 6.6, 3.1\% for $U=2, 4, 8$.
At very strong couplings, the error of the finite system ED with PBC approaches that of FT-DMET. This is because both the finite size error
and the effectiveness of the DMET bath decrease as one approaches the atomic limit.

\begin{figure}
    \centering
    \includegraphics[width=1\textwidth]{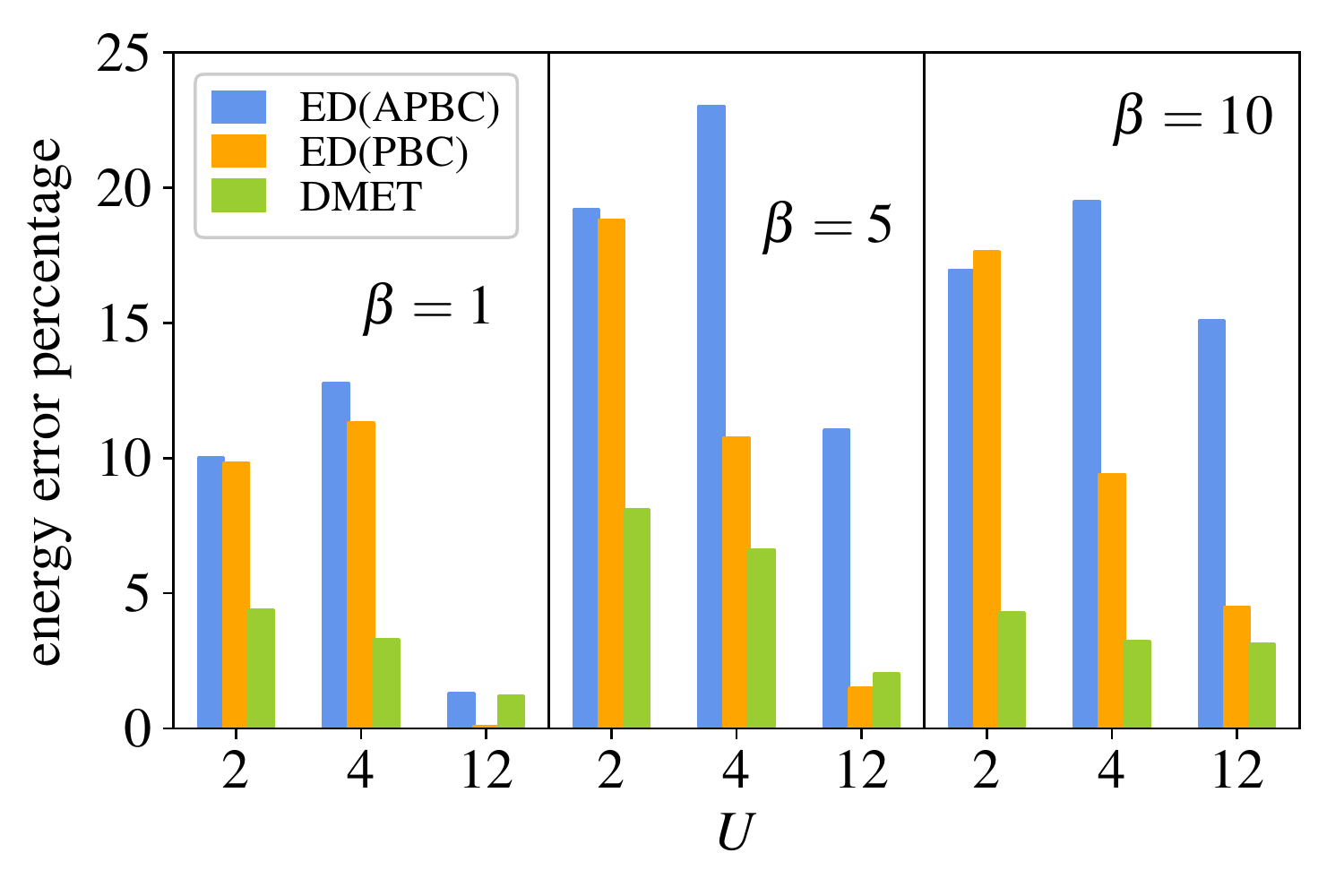}
    \caption{Percentage error of the FT-DMET (with 2 impurity sites and 2 bath orbitals) energy per site vs. ED (4 sites) on a non-embedded cluster with PBC and APBC boundary conditions
    for the 1D Hubbard model at various $U$ and $\beta$.} \label{fig:dmet-vs-ed}
\end{figure}


As a second test, in Fig.~\ref{fig:impvsbath} we compare increasing the number of impurity sites versus increasing
the number of bath orbitals generated in Eq.~(\ref{eq:dbath}) for various $U$ and $\beta$. Although complex behaviour
is seen as a function of $\beta$, we roughly see two patterns. For certain impurity sizes, (e.g. $I4$) it can be slightly
more accurate to use a larger impurity with an equal number of bath sites, than a smaller impurity with a larger number of bath sites.
(For example, at $U=8$, one can find a range of $\beta$ where $I4B4$ gives a smaller error than $I2B6$). However, there
are also some impurity sizes which perform very badly; for example $I3B3$ gives a very large error, because the (short-range) antiferromagnetic
correlations do not properly tile between adjacent impurities when the impurities are of odd size. Thus, due to these
size effects, convergence with impurity size is highly non-monotonic, but  increasing the bath size (by including more terms
in Eq.~(\ref{eq:dbath})) is less prone to strong odd-size effects.
The ability to improve the quantum impurity by simply increasing the number of bath sites,
is expected to be particularly relevant in higher-dimensional lattices such as the 2D Hubbard model, where ordinarily to
obtain a sequence of clusters with related shapes,
it is necessary to increase the impurity size by large steps. Nonetheless, convergence with bath size is also not strictly monotonic,
as also illustrated in Fig.~\ref{fig:1dES}, where we see that the error in the $I2B4$ entropy can sometimes be less
than that of $I2B6$ for certain ranges of $\beta$. For the largest embedded problem $I2B6$, the maximum
error in the entropy is $4\times 10^{-3}$ and $2\times 10^{-2}$ for $U =4$ 
and $8$, respectively.

\begin{figure}
    \centering
    \includegraphics[width=1\textwidth]{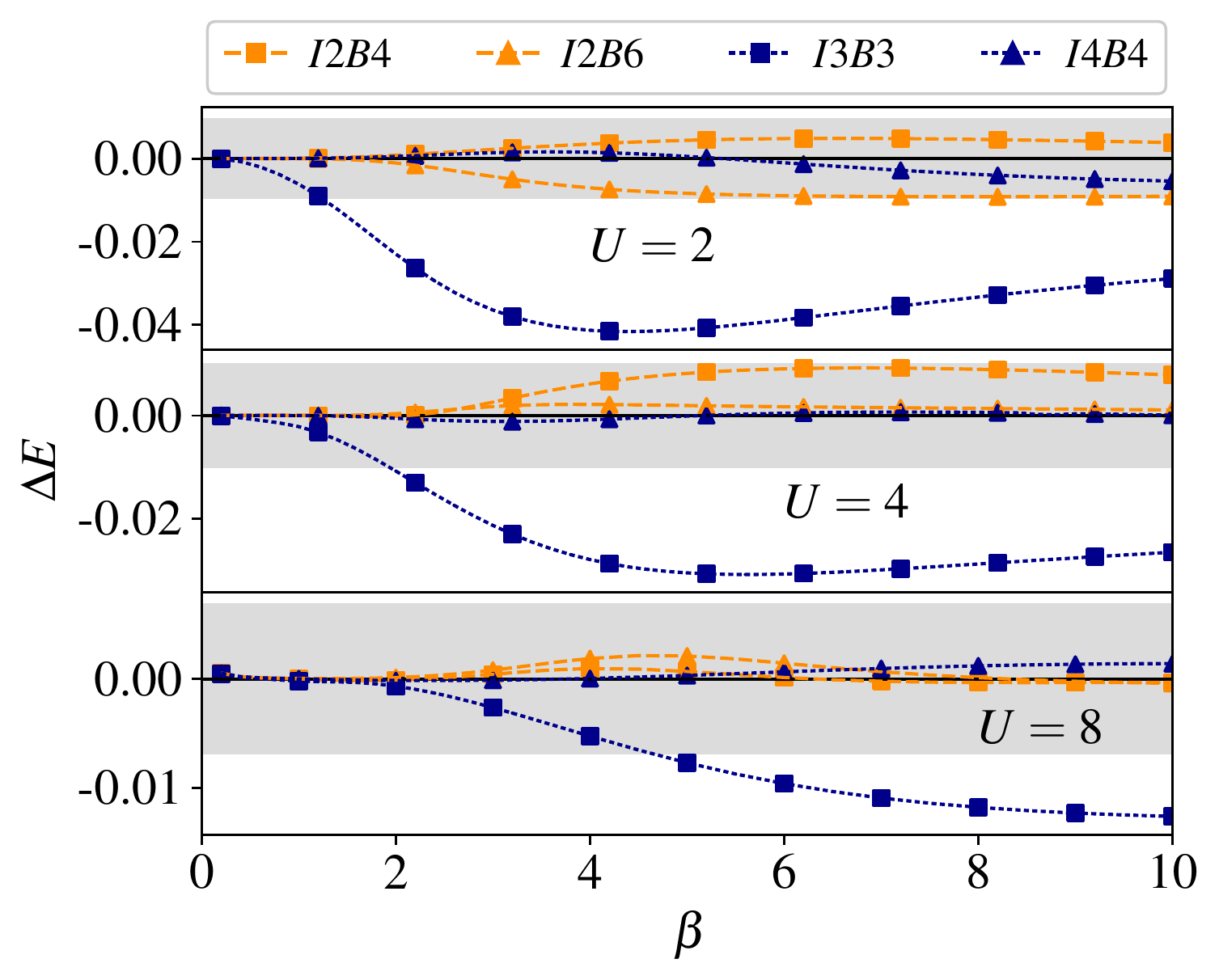}
     \caption{Absolute error of the FT-DMET
      energy per site of the 1D Hubbard model at half-filling
     as a function of impurity and bath size. $InBm$ denotes $n$ impurity sites and $m$ bath orbitals.
Increasing impurity (blue lines); increasing bath (orange lines). The grey band depicts the ground state error with $2$ impurity sites and $2$ bath orbitals.}
     \label{fig:impvsbath}
\end{figure}

\begin{figure}
     \centering
     \includegraphics[width=1\textwidth]{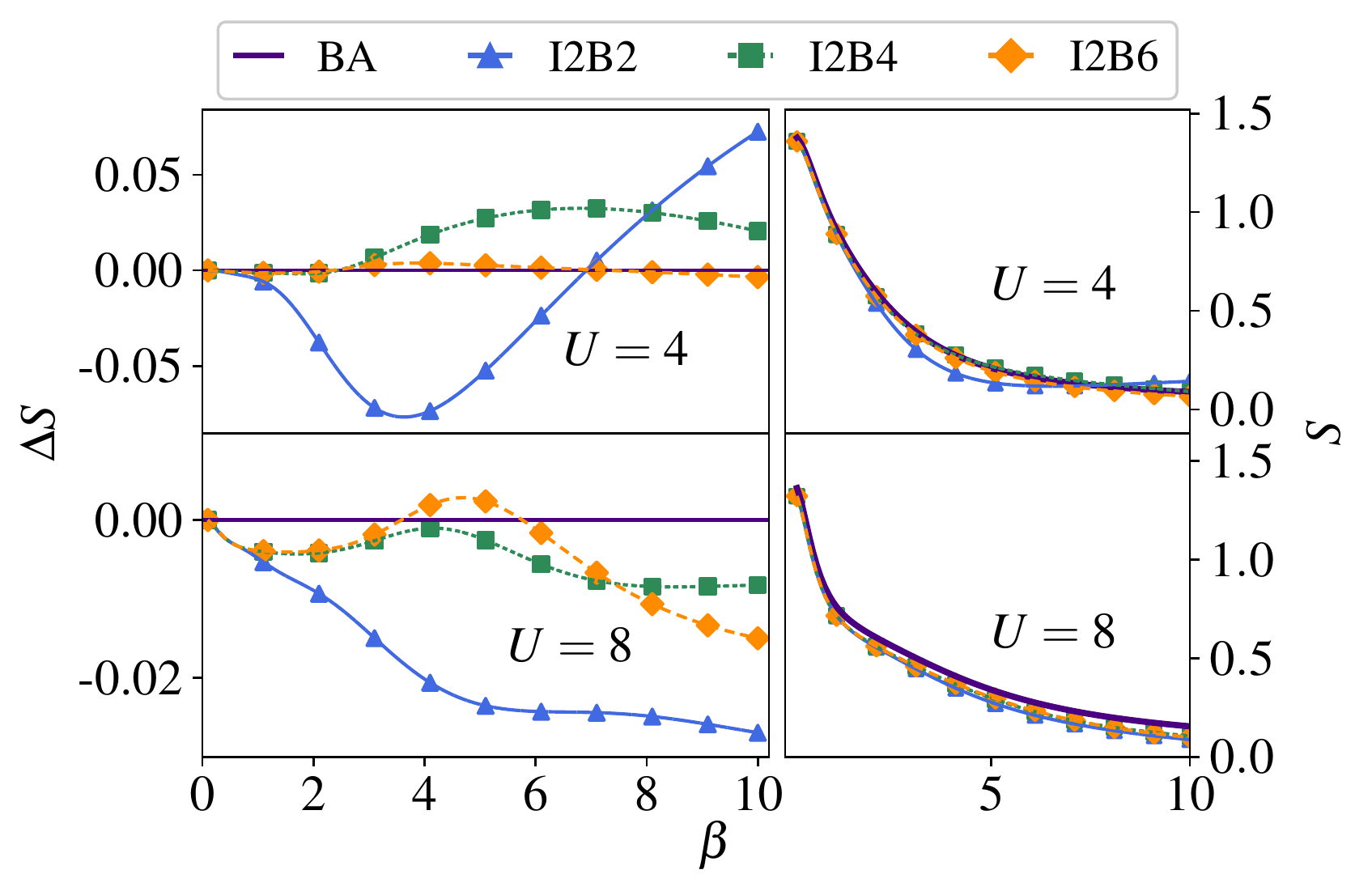}
     \caption{Absolute error of the FT-DMET entropy per 
              site of the 1D Hubbard model at half-filling
              as a function of the number of bath sites. 
              The right panels show the absolute entropy.
               }
     \label{fig:1dES}
\end{figure}

The preceding calculations were all carried out at half-filling. Thus, 
in Fig.~\ref{fig:1ddopping} we show FT-DMET calculations on the 1D Hubbard model away from half-filling at $U=4$.
We chose to simulate a finite Hubbard chain of 16-sites with PBC in order to readily generate  numerically exact reference data
using FT-DMRG (using a maximum bond dimension of $2000$ and an
imaginary time step of $\tau = 0.025$). The agreement between the FT-DMRG energy per site and
that obtained from the thermal Bethe ansatz can be seen at half-filling, corresponding to a chemical potential $\mu=2$.
We see excellent agreement between FT-DMET and FT-DMRG results across the full range
of chemical potentials, and different $\beta$, suggesting that FT-DMET works equally well for doped systems as
well as for undoped systems.

\begin{figure}
     \centering
     \includegraphics[width=1\textwidth]{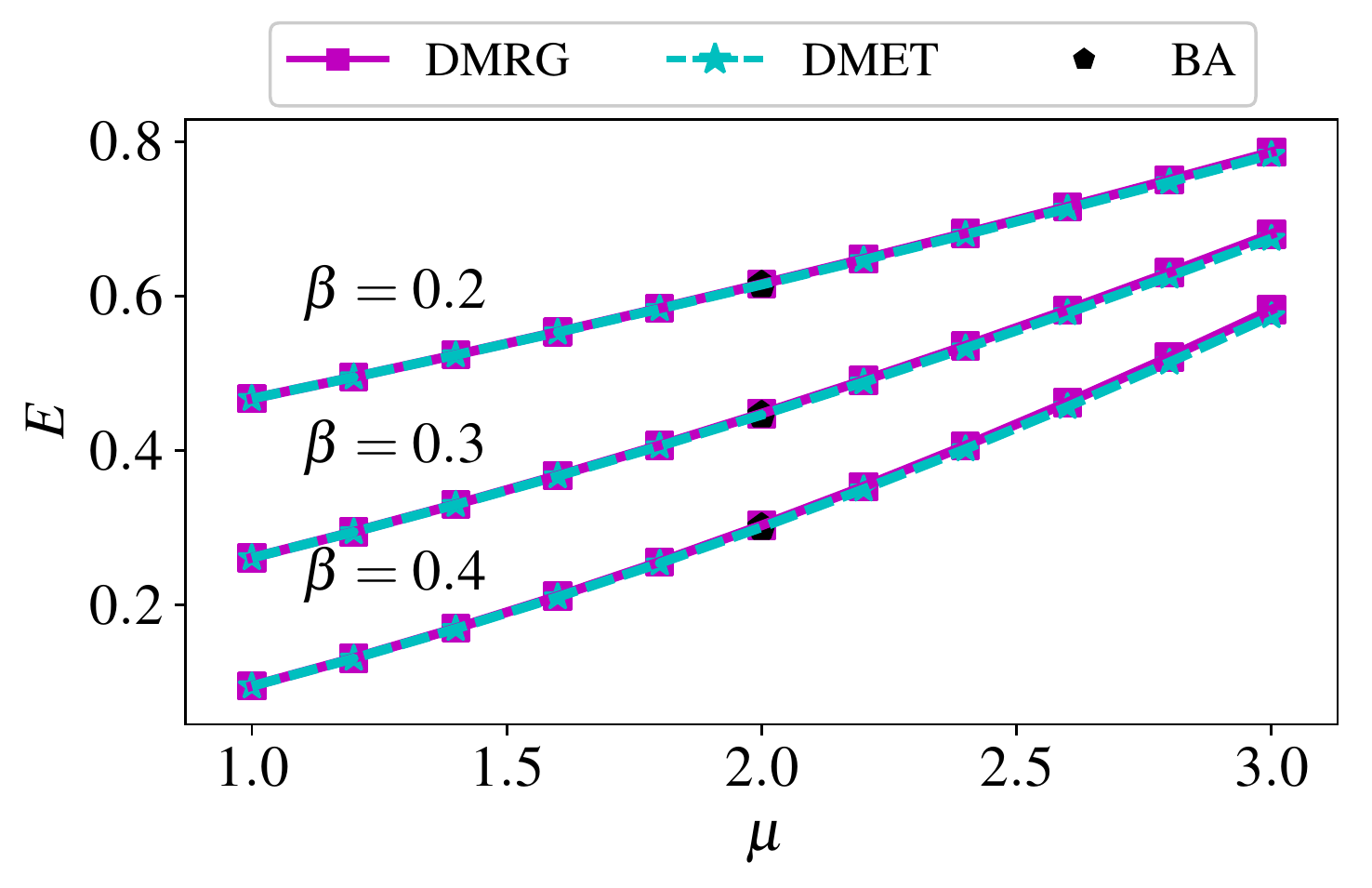}
     \caption{Energy per site (units of $t$) of a 16-site 
     Hubbard chain with periodic boundary conditions at $U=4$ as a function
     of the chemical potential $\mu$ at various $\beta$ values.     
     The difference between the DMRG and DMET ($I2B6$) energies per site is 
     $0.01-1.4\%$.
     Solid lines: DMRG energies; dashed lines: DMET energies; pentagons: Bethe ansatz.}
     \label{fig:1ddopping}
\end{figure}

\subsection{2D Hubbard model}\label{sec:2dhub}

The 2D Hubbard model is an essential model of correlation physics in materials.
We first discuss the accuracy of FT-DMET for the energy 
of the 2D Hubbard model at half-filling, shown in Fig.~\ref{fig:2dE}. The FT-DMET calculations
are performed on a $2\times 2$ impurity, with 4 bath orbitals ($I4B4$) (green diamond markers) 
and 8 bath orbitals ($I4B8$)  (red triangular markers). The results are compared
to DCA calculations with clusters of size $34$ (orange circle markers), $72$
(blue square markers)~\citep{LeBlancPRX2015}, and $2\times 2$  (light blue hexagon markers) (computed for this work). 
The DCA results with the size $72$ cluster can be considered here to represent the thermodynamic limit.
The 
DCA($2\times 2$) 
data provides an opportunity to assess the relative contribution of the FT-DMET embedding to the finite size error; in particular, one can compare the 
difference between FT-DMET and DCA(72) to the difference between 
DCA($2\times2$) and DCA(72).
Overall, we see that the FT-DMET energies with 
8 bath orbitals are in good agreement with the DCA(72) energies across the different $U$ values, and that the accuracy is slightly
better on average than that of
DCA($2\times2$). 
The maximum error in the $I4B8$ impurity compared to thermodynamic limit extrapolations of the
DCA energy~\citep{LeBlancPRX2015} is found at $U=4$ and is in the range of 1-2\%,
comparable to errors observed in ground-state DMET at this cluster size (e.g. the error in GS-DMET at $U=4$ and $U=8$ is 0.3\% and 1.8\%, respectively).
In the $\beta = 8$ case, the FT-DMET calculations with two different bath sizes give very similar results; at low temperature, the 
bath space constructed by the FT procedure is similar to that of the ground state, and the higher order bath sites do not contribute
very relevant degrees of freedom. Thus even the smaller bath achieves good accuracy in the low temperature FT-DMET calculations. 

\begin{figure}
     \centering
     \includegraphics[width=1\textwidth]{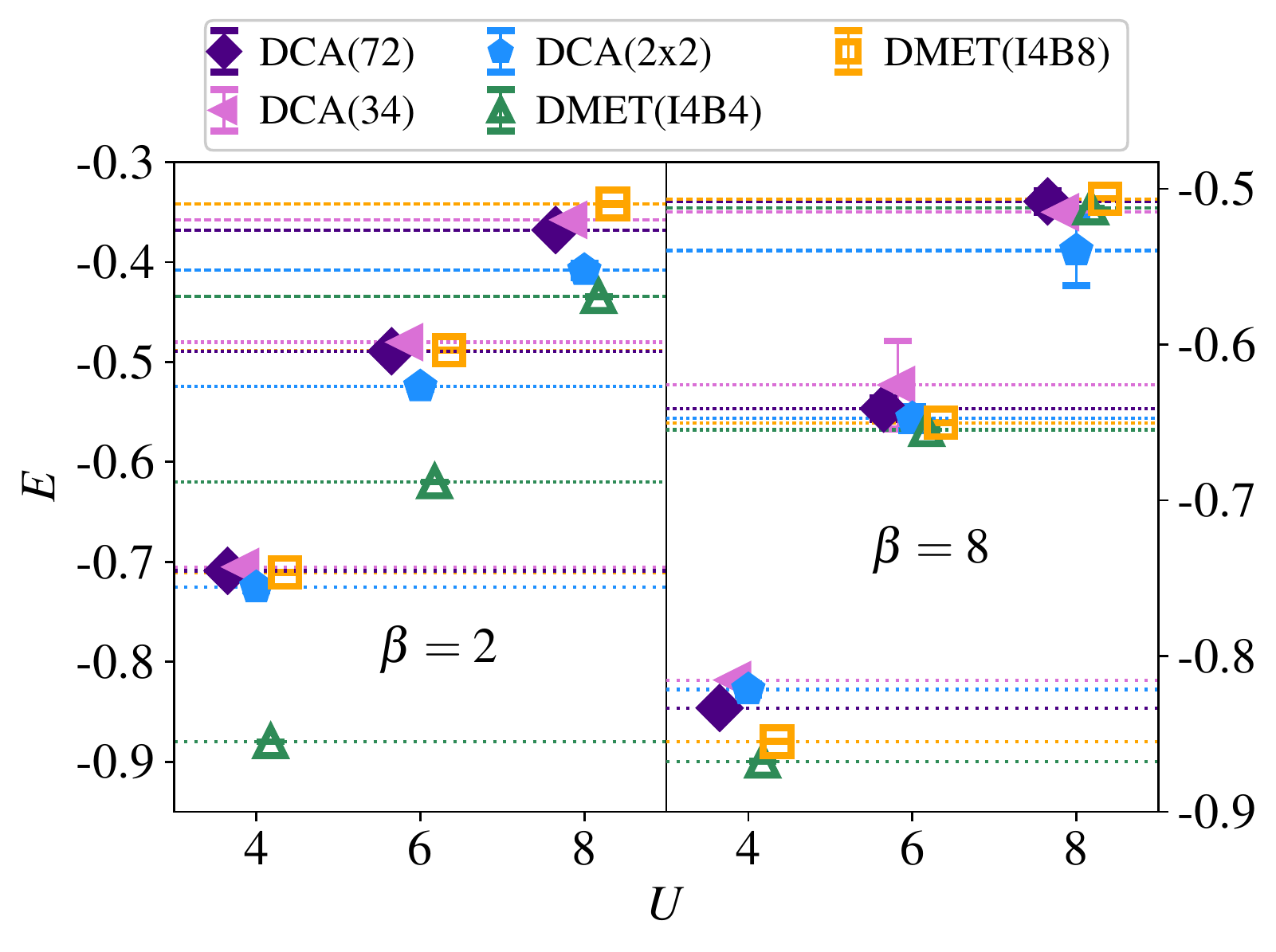}
     \caption{Energy per site versus U (units of $t$) of the 
     2D Hubbard model at half-filling with FT-DMET ($2\times2$ cluster with 
     4 and 8 bath orbitals), 
     DCA (34, 72 and $2\times 2$ site clusters).
     }\label{fig:2dE}
     \end{figure}

\begin{figure}[t!]
\centering
\begin{subfigure}[t]{0.85\textwidth}
\includegraphics[width=\textwidth]{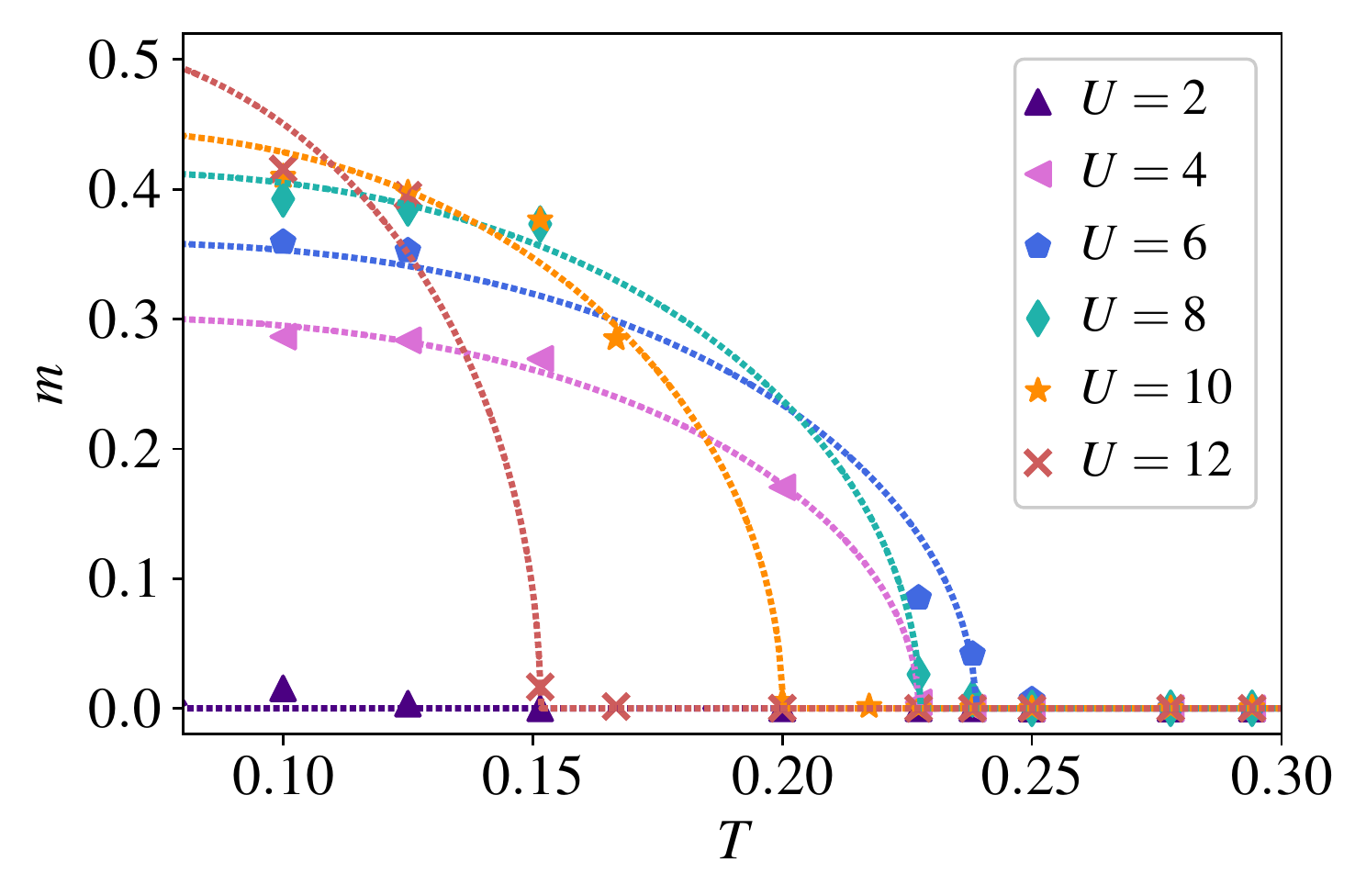}
\caption{ }
\end{subfigure}
\begin{subfigure}[t]{0.85\textwidth}
\includegraphics[width=\textwidth]{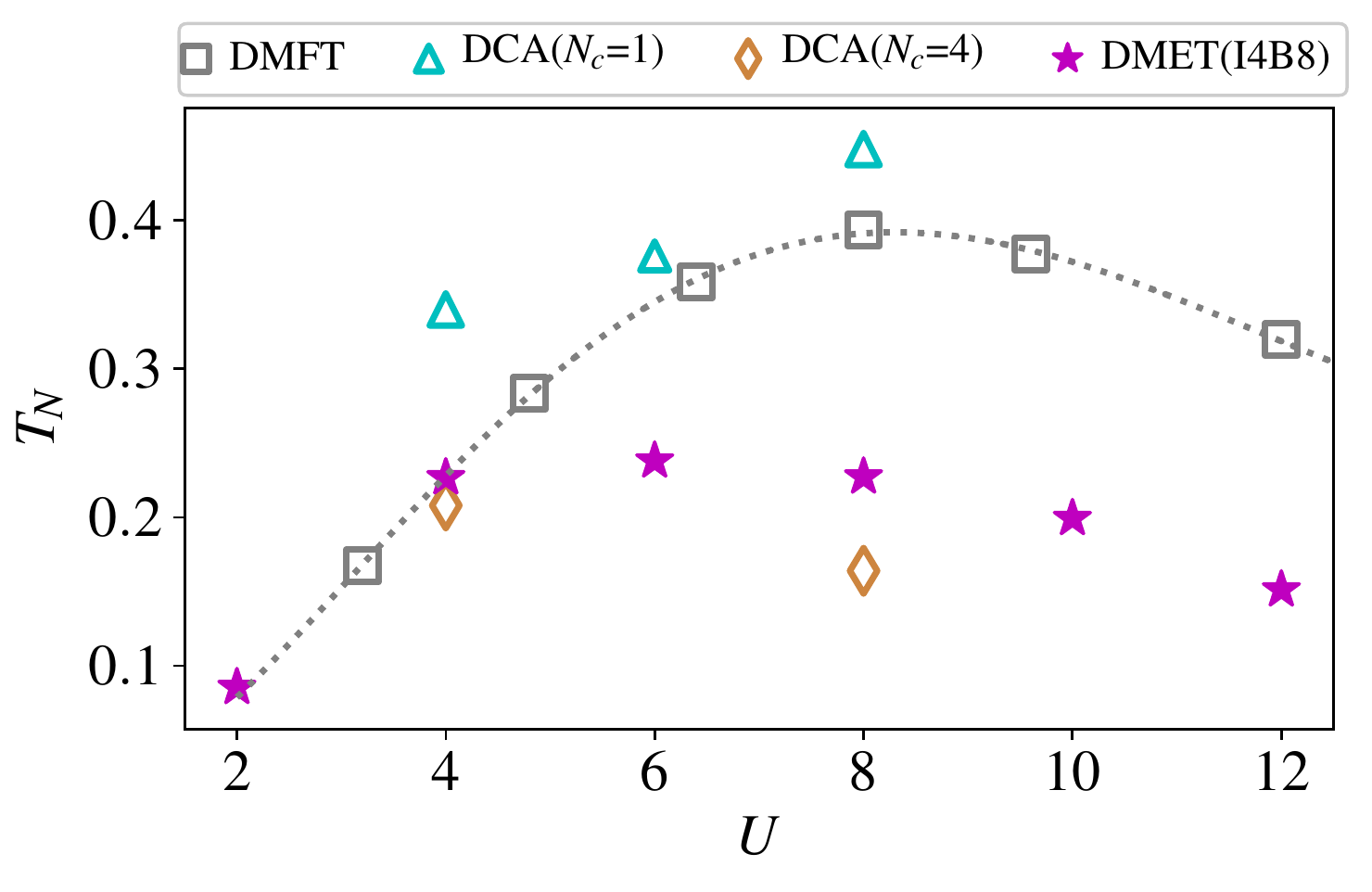}
\caption{ }
\end{subfigure}
\caption{N\'eel  transition for the 2D Hubbard model within quantum impurity simulations.
 (a) Antiferromagnetic moment $m$ as a function of $T$ with various $U$ values (units of $t$);
(b) N\'eel  temperature $T_N$ calculated with FT-DMET, single-site DMFT and DCA.  DMFT data is taken from Ref.~\citep{kunes}, DCA/NCA data for $U=4$ 
 is taken from Ref.~\citep{maier2005}, DCA/QMC data for $U=6$ is 
 taken from Ref.~\citep{jarrellPRB}, and DCA/QMC data for $U=8$ is 
 taken from Ref.~\citep{jarrellEPL}.}\label{fig:2dlinemag}
\end{figure}

A central phenomenon in magnetism is the finite-temperature N\'eel  transition.
In the thermodynamic limit, the 2D Hubbard model does not exhibit a true N\'eel transition,
but shows a crossover~\citep{mermin}. However, in finite quantum impurity calculations,
the crossover appears as a phase transition at a nonzero N\'eel temperature. 
Fig.~\ref{fig:2dlinemag}(a)
shows the  antiferromagnetic moment $m$ calculated as $m = \frac{1}{2L_x}\sum_{i}^{L_x}|n_{i\uparrow}-n_{i\downarrow}|$ as a function of temperature $T$ for various 
$U$ values. As a guide to the eye, we fit the data to a mean-field 
magnetization function $m = a\tanh\left(bm/T\right)$,
where $a$ and $b$ are parameters that depend on $U$. The FT-DMET calculations are performed with a $2\times 2$ impurity and $8$ bath orbitals,
using a finite temperature DMRG solver  with maximal bond dimension $M = 600$ and time step $\tau = 0.1$. With this, the error in $m$ from the solver
is estimated  to be less than 10$^{-3}$.
$m$ drops sharply to zero as $T$ is increased
signaling a N\'eel  transition. The N\'eel  temperature $T_N$ is taken as the point of intersection
of the mean-field fitted line with the $x$ axis; assuming this form of the curve, the uncertainty in $T_N$ is invisible
on the scale of the plot.
The plot of $T_N$ versus $U$ is shown in Fig.~\ref{fig:2dlinemag}(b), showing that
 the maximal $T_N$ occurs at $U=6$.
Similar $T_N$ calculations on the 2D Hubbard model with single site DMFT~\citep{kunes} and DCA\cite{jarrellPRB,jarrellEPL, maier2005} are
also shown in Fig.~\ref{fig:2dlinemag}(b) for reference. Note that the difference in the DMFT results~\citep{kunes} and single-site DCA (formally equivalent to DMFT)~\cite{jarrellPRB,jarrellEPL}
likely arise from the different solvers used.
The behaviour of $T_N$ in our $2\times 2$ FT-DMET calculations
is quite similar to that of the 4-site DCA cluster.
In particular, we see in DCA that the
$T_N$ values obtained from calculations with a single-site cluster ($N_c=1$)
are higher than the $T_N$ values obtained from calculations with a 
4-site cluster ($N_c = 4$). 

An alternative visualization of the N\'eel transition in FT-DMET is shown in 
Fig.~\ref{fig:2dpd4b}. The FT-DMET calculations here were performed with 
a $2\times 2$ impurity and $4$ bath orbitals using an ED solver.
Though less quantitatively accurate than the $8$ bath orbital simulations, these FT-DMET calculations still capture
the qualitative behavior of the N\'eel  transition. Focusing on the dark blue
region of the phase diagram, one can estimate the maximal $T_N$ to occur near $U\approx 9$, an increase
over the maximal N\'eel temperature using the $8$ bath orbital impurity model. This increase
in the maximal $T_N$ appears similar to that which happens when moving from a 4-site cluster to a 1-site cluster in DCA
in Fig.~\ref{fig:2dlinemag}.

\begin{figure}[t!]
\centering
\includegraphics[width=0.9\textwidth]{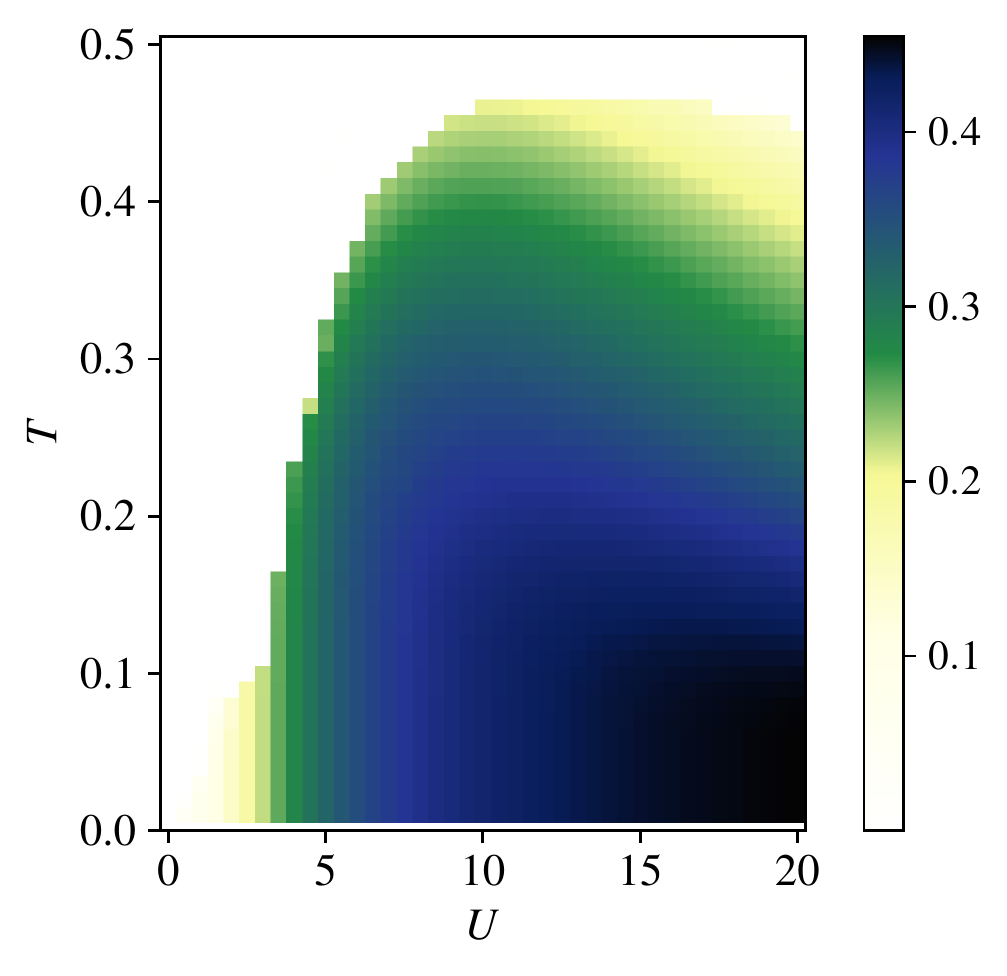}
\caption{2D Hubbard antiferromagnetic moment (color scale) as a function of $T$ and $U$ (units of $t$) in FT-DMET ($2\times 2$ impurity, 4 bath orbitals.)}\label{fig:2dpd4b}
\end{figure}


\section{\label{sec:conc_dmet}Conclusions}
To summarize, we have introduced a finite temperature formulation
of the density matrix embedding theory (FT-DMET). This temperature
formulation inherits most of the basic structure of the ground-state
density matrix embedding theory, but modifies the bath construction
so as to approximately reproduce the mean-field finite-temperature density matrix.
From numerical assessments on the 1D and 2D Hubbard model, we conclude
that the accuracy of FT-DMET is comparable to that of its ground-state counterpart, with
at most a modest increase in size of the embedded problem. From the limited comparisons, it also
appears to be competitive in accuracy with the cluster dynamical mean-field theory for the same sized cluster.
Similarly to ground-state DMET, we expect FT-DMET to be broadly applicable
to a wide range of model and \emph{ab initio} problems of correlated electrons at finite temperature~\citep{ZhengScience2017,cui2019efficient}. 

%


\chapter{\textit{Ab initio} finite temperature density matrix embedding theory\label{chp:hlatt}}
\section{\label{sec:hlatt_abs}Abstract}
This work describes the framework of the finite temperature density matrix 
embedding theory (FT-DMET) for \textit{ab initio} simulations of solids.
We introduce the implementation details 
including orbital localization, density fitting
treatment to the two electron repulsion integrals, bath truncation, lattice-to-embedding
projection, and impurity solvers. We apply this method to study the 
thermal properties and phases of hydrogen lattices. We provide the finite 
temperature dissociation curve, paramagnetic-antiferromagnetic transition,
and metal-insulator transition of the hydrogen chain.

\section{\label{sec:hlatt_intro}Introduction}
The numerical study of the many-electron problem has been playing a profound
role in understanding the electronic behaviors in molecules and materials.
One big challenge for current numerical methods is the description of strong
electron correlations, which requires non-trivial treatment of the 
electron-electron interaction beyond the mean-field level. A variety of numerical
algorithms have been invented in the past decades to treat strong electron
correlations, including post-Hartree-Fock quantum chemistry methods such
as CCSD~\citep{Monkhorst1977, Bartlett2007}, DMRG and its multi-dimensional
alternatives~\citep{White1992, Scholl2005, Chan2011,Evenbly2015,Verstraete2008}, the QMC family such as AFQMC~\citep{Ceperley1977,Acioli1994, Foulkes2001}, and embedding
methods such as DMET~\citep{Honma1995, Carlson2011}. During the past decades, noticeable progress
has been made in the study of strongly correlated models such as one-dimensional
and two-dimensional Hubbard
models\cite{Lieb1989,White1989, LeBlanc2015}, while the \textit{ab initio}
 study of strongly correlated solids is rare. Compared to model systems where
forms of the two-electron interaction are usually simple, \textit{ab initio}
Hamiltonians contain much more complicated two-electron terms with size $N^4$,
where $N$ is the number of orbitals. This complexity brings higher 
computational costs. Therefore, an efficient method that can handle 
the realistic Hamiltonian accurately is crucial for understanding the
physics behind real materials.

The hydrogen lattice is believed
to be the simplest chemical system with a straightforward analog to the
Hubbard model. A thorough comparison between the hydrogen lattice and Hubbard
model could provide insights of the roles of (i) the long range correlation
and (ii) the multi-band effect (with basis set larger than STO-6G). The
numerical study of hydrogen chain can be traced back to the 70s
with simple theoretical tools such as many body perturbation
theory (MBPT)\cite{Ross1976}. The rapid development of numerical algorithms
made it possible to achieve a better accuracy and thus plausible
conclusions\cite{Stella2011, Hachmann2006, Al-Saidi2007, Sinitskiy2010,
NguyenLan2016, Motta2017, Motta2019, Liu2020}.
Motta et al. benchmarked the equation of state\cite{Motta2017} and
explored the ground state properties\cite{Motta2019} of the hydrogen chain
with various popular numerical methods including DMRG and AFQMC. Despite the
numerous ground state simulations, the finite temperature study of hydrogen
lattices is rare, while the finite temperature study is crucial for
understanding the temperature-related phase diagrams. 
Liu et. al. studied the finite temperature behaviors of
hydrogen chain with the minimal basis set (STO-6G), and identified the 
signature of Pomeranchuk effect\cite{Liu2020}. However, the minimal basis
set hindered the exploration of more interesting phenomena caused by
the multi-band effect. A more thorough study 
beyond the minimal basis set is needed to reach a quantitative observation
of the finite temperature behaviors of the hydrogen lattices.

In this work, we apply \textit{ab initio} finite temperature density matrix
embedding theory (FT-DMET)~\citep{Sun2020} 
algorithm to study metal-insulator and magnetic crossovers in periodic 
one-dimensional and two-dimensional hydrogen lattices as a function
of temperature $T$ and H-H bond distance $R$. We also explore how
basis set size influences the shape of the phases by comparing the results
with STO-6G, 6-31G, and CC-PVDZ basis sets. The rest of the article is 
organized as follows: in Section~\ref{sec:abinit_ftdmet}, we present the
formulation and implementation details of \textit{ab initio} FT-DMET, 
including orbital localization, tricks to reduce the cost due to the 
two electron repulsion terms, bath truncation, impurity solver, and thermal
observables. In Section~\ref{sec:res_hlatt}, we demonstrate the  \textit{ab initio} FT-DMET algorithm by studies of the dissociation curves and 
phase transitions in a one-dimensional periodic hydrogen lattice. We finalize
this article with conclusions in Section~\ref{sec:conc_hlatt}.

\section{\textit{Ab initio} FT-DMET}\label{sec:abinit_ftdmet}
In our previous work\cite{Sun2020}, we introduced the basic formulation
of FT-DMET for lattice models. Going from lattice models to chemical
systems, there are several practical difficulties\cite{Cui2020}:
(i) the definition of impurity relies on the localization of the orbitals;
(ii) the number of orbitals in the impurity can be easily very large depending on
the infrastructure of the supercell and the basis set; (iii) manipulating
two-electron repulsion integrals in a realistic Hamiltonian is usually very expensive; (iv) an impurity solver which can handle \textit{ab initio} 
Hamiltonians efficiently at finite temperature is required.
In the rest of this section, we discuss solutions to the above challenges and
provide implementation details of \textit{ab initio} FT-DMET.

\subsection{Orbital localization}
Since we are dealing with periodic lattices, the whole lattice problem is 
described with Bloch (crystal) orbitals in the momentum space.
 Thus crystal atomic
orbitals (AOs) $\{\phi^{\textbf{k}}_{\mu}(\textbf{r})\}$ are a natural choice.
The definition of impurity, 
however, is based on real space localized orbitals~\citep{Edmiston1963}. Therefore we 
define a two-step transformation from Bloch orbitals to localized orbitals (LOs) 
$\{w_i(\mathbf{r})\}$. 
\begin{equation}\label{eq:lo_trans}
\begin{split}
w_i^{\mathbf{R}}(\mathbf{r}) =& \frac{1}{\sqrt{N_\mathbf{k}}}\sum_{\mathbf{k}}
e^{-i\mathbf{k}\cdot \mathbf{R}}w_i^{\mathbf{k}}(\mathbf{r}),\\
w_i^{\mathbf{k}}(\mathbf{r}) =& \sum_{\mu} \phi^{\textbf{k}}_{\mu}(\textbf{r})
C_{\mu i}^{\textbf{k}},
\end{split}
\end{equation}
where $C_{\mu i}^{\textbf{k}}$ transforms AOs in momentum space 
 $\{\phi^{\textbf{k}}_{\mu}(\textbf{r})\}$
into LOs  in momentum space $w_i^{\mathbf{k}}(\mathbf{r})$, and
LOs in real space $w_i^{\mathbf{R}}(\mathbf{r})$ are derived by
a Wannier summation over the local crystal orbitals $w_i^{\mathbf{k}}(\mathbf{r})$. 

With the \textit{ab initio} periodic system expressed in LOs, one could
choose the impurity to be spanned by the LOs in a single unit cell or 
supercell. In the rest of this paper, we choose the impurity to be
the supercell at the lattice origin.

To define the localization coefficients $C_{\mu i}^{\textbf{k}}$ in 
Eq.~\eqref{eq:lo_trans}, we use a bottom-up strategy: transform from 
AO computational basis to LOs. This strategy uses linear algebra to 
produce LOs, and thus avoids dealing with complicated optimizations.
There are several choices of LOs from the bottom-up strategy: 
L\"owdin and meta-L\"owdin orbitals~\citep{Lowdin1950,}, natural AOs (NAO)~\citep{Lowdin1956,Reed1985},
and intrinsic AOs (IAO)~\citep{KniziaJCTC2013}. In this work,
we used the $\mathbf{k}$-space unrestricted Hartree-Fock (KUHF)
function with density fitting in the quantum chemistry package PySCF\cite{PYSCF2017,PYSCF2020} to generate a set of crystal MOs.
Then we applied an adapted
IAO routine to generate a set of crystal IAOs from the crystal MOs with
 $\mathbf{k}$-point sampling. The crystal IAOs generated from this routine
are valence orbitals that exactly span the occupied space of the mean-field
 calculation. Note that the number of crystal IAOs is equal to the size of minimal basis only. To carry out calculations beyond the minimal basis, 
we construct the rest nonvalence orbitals to be projected AOs (PAOs)~\citep{Saebo1993},
orthogonalized with L\"owdin orthogonalization~\citep{Aiken1980}. This IAO+PAO
strategy has been used in previous ground state DMET calculations~\citep{WoutersJCTC2016,Cui2020}. 

\subsection{Bath truncation and finite temperature bath}
In the standard DMET routine, the bath orbitals used to construct the 
embedding space are obtained from the SVD of the mean-field off-diagonal
density matrix between the impurity and the remaining lattice (called environment) $\gamma^{\mathbf{R}\neq\mathbf{0},\mathbf{0}}$
\begin{equation}\label{eq:svd_bath}
\gamma^{\mathbf{R}\neq\mathbf{0},\mathbf{0}} = B^{\mathbf{R}\neq\mathbf{0}}\Lambda V^{\mathbf{0}\dag},
\end{equation}
where $B^{\mathbf{R}}$ defines the coefficients of bath orbitals in the LO 
basis. Thus we can construct the projection matrix in real space
\begin{equation}\label{eq:proj_mat_R}
P^{\mathbf{R}} = \begin{pmatrix}
\mathbf{I} & \mathbf{0} \\
\mathbf{0} & \mathbf{B^{R\neq 0}}
\end{pmatrix}.
\end{equation}
The projection in momentum space can be derived from Eq.~\eqref{eq:proj_mat_R}
with Wannier transformation
\begin{equation}
P^{\mathbf{k}} = \sum_{\mbf{R}}e^{-i \mbf{k\cdot R}} P^{\mathbf{R}}.
\end{equation}
Note that the projection matrices $P^{\mathbf{R}}$ and $P^{\mathbf{k}}$ 
are in the basis of LOs, and to obtain the $\mbf{k}$-space transformation 
matrix in AOs, simply multiply $C^{\mbf{k}}$ from Eq.~\eqref{eq:lo_trans}
to the left of $P^{\mathbf{k}}$.

From Eq.~\ref{eq:svd_bath}, one generates a set of bath orbitals with the same size
as the impurity. This setting is valid and efficient for model systems, and
the embedding space is purely constructed with valence bands.
However, for \textit{ab inito} calculations, low-lying core and high-energy
virtual impurity orbitals will not entangle with the environment, and thus
with the bath orbitals. In practice, this results in singular values in
the SVD of Eq.~\eqref{eq:svd_bath}, leading to difficulties in the 
convergence of the DMET self-consistency procedure. To overcome this difficulty,
we use the following strategy~\citep{WoutersJCTC2016}: we identify the 
impurity orbitals as core, valence, and virtual orbitals, and then only take
valence columns of the off-diagonal density matrices of the off-diagonal 
density matrix to construct the bath orbitals, as illustrated in 
Fig.~\ref{fig:svd_valence}. With this strategy, the 
size of bath orbitals is equal to the size of valence impurity orbitals,
and thus the number of embedding orbitals is reduced from $2n_{\text{imp}}$ to 
$n_{\text{imp}} + n_{\text{val}}$. Note that if pseudopotential is included
in the calculation, there are no core orbitals.

\begin{figure}[h!]
\centering
\begin{subfigure}[t]{0.8\textwidth}
\includegraphics[width=\textwidth]{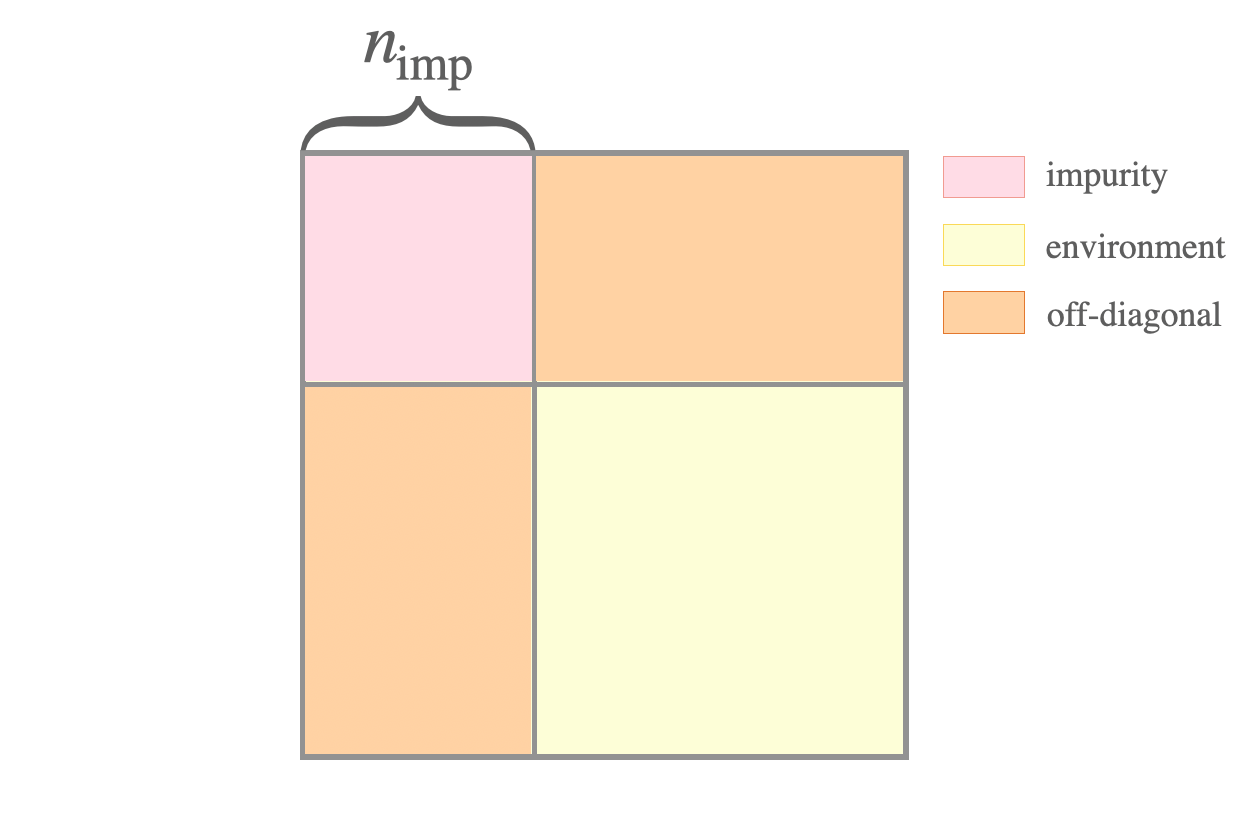}
\caption{ }
\end{subfigure}
\begin{subfigure}[t]{0.8\textwidth}
\includegraphics[width=\textwidth]{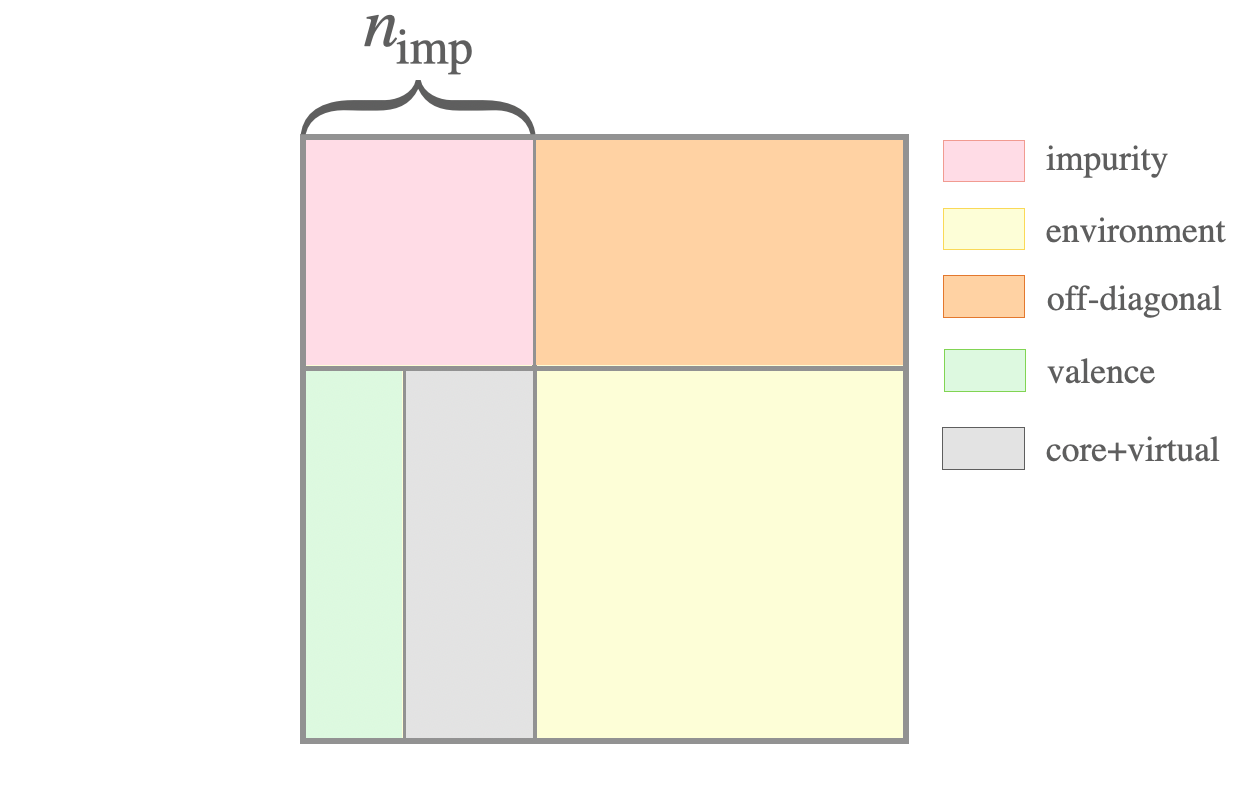}
\caption{ }
\end{subfigure}
\caption{Bath orbitals from singular value decomposition (SVD) of the off-diagonal block of the
mean-field density matrix. The whole square represents the mean-field
density matrix of size $N\times N$, with $N$ being the total number of orbitals, and the first $n_{\text{imp}}$ columns/rows of the matrix are orbitals
in the impurity. (a) Standard DMET routine computes the bath orbitals by the 
SVD of the off-diagonal block (orange blocks on the left and top, the block
on the left corresponds to Eq.~\eqref{eq:svd_bath}). (b) \textit{Ab initio}
DMET computes the bath orbitals by the SVD of only the valence columns in 
the off-diagonal blocks (green block).}\label{fig:svd_valence}
\end{figure}

At finite temperature, electronic occupation numbers of the energy levels 
are ruled by the Fermi-Dirac distribution
\begin{equation}\label{eq:fd_hlatt}
f(\varepsilon_i) = \frac{1}{1+e^{\beta(\varepsilon_i-\mu)}},
\end{equation}
where $\varepsilon_i$ is the energy of the $i$th molecular orbital, 
$\beta = \frac{1}{T}$ is the inverse temperature (we set the Boltzmann's constant $k_B = 1$) and $\mu$ is the chemical potential or the energy 
of the Fermi level at ground state. When $\beta = \infty$, 
Eq.~\eqref{eq:fd_hlatt} reproduces the ground state electron number
distribution: when $\varepsilon_i < \mu$, the occupation number is $1$ (occupied
orbitals), and when $\varepsilon_i > \mu$, the occupation number is $0$ 
(virtual orbitals). However, when $\beta$ is finite, the 
electronic occupation number on virtual orbitals is no longer $0$. The 
extreme case is when $\beta = 0$ where all energy levels are uniformly occupied
with occupation number $f = 0.5$. Therefore, the ground state bath
construction described previously is no longer suitable to provide an
accurate embedding Hamiltonian. There are generally two strategies:
(1) include part of the core and virtual orbitals into the off-block for SVD;
(2) obtain additional bath orbitals from higher powers of the mean-field
density matrix~\citep{Sun2020}: take the valence columns of the off-diagonal
blocks of $\gamma$, $\gamma^2$, ..., $\gamma^l$ and apply SVD to them,
respectively, to get $l$ sets of bath orbitals, then put the bath orbitals 
together and perform orthogonalization to produce the final bath orbitals. 
The disadvantage of the first strategy is obvious: as temperature gets 
higher, the Fermi-Dirac curve in Fig.~\ref{fig:fd_curve} gets flatter,
and thus more non-valence orbitals are needed. Compared to the first strategy,
the latter strategy generally requires less number of bath orbitals.
For most systems, truncating
$l$ to $2$ or $3$ is already enough for the whole temperature spectrum,
therefore the number of embedding orbitals is $n_{\text{imp}} + ln_{\text{val}}$.
Since the number of valence orbitals is much smaller than that of the
non-valence orbitals, DMET with bath orbitals derived from the second
strategy is more economic and stable. In this paper, we adopt the second
strategy for our FT-DMET calculations.

\begin{figure}[h!]
     \centering
     \includegraphics[width=1.\textwidth]{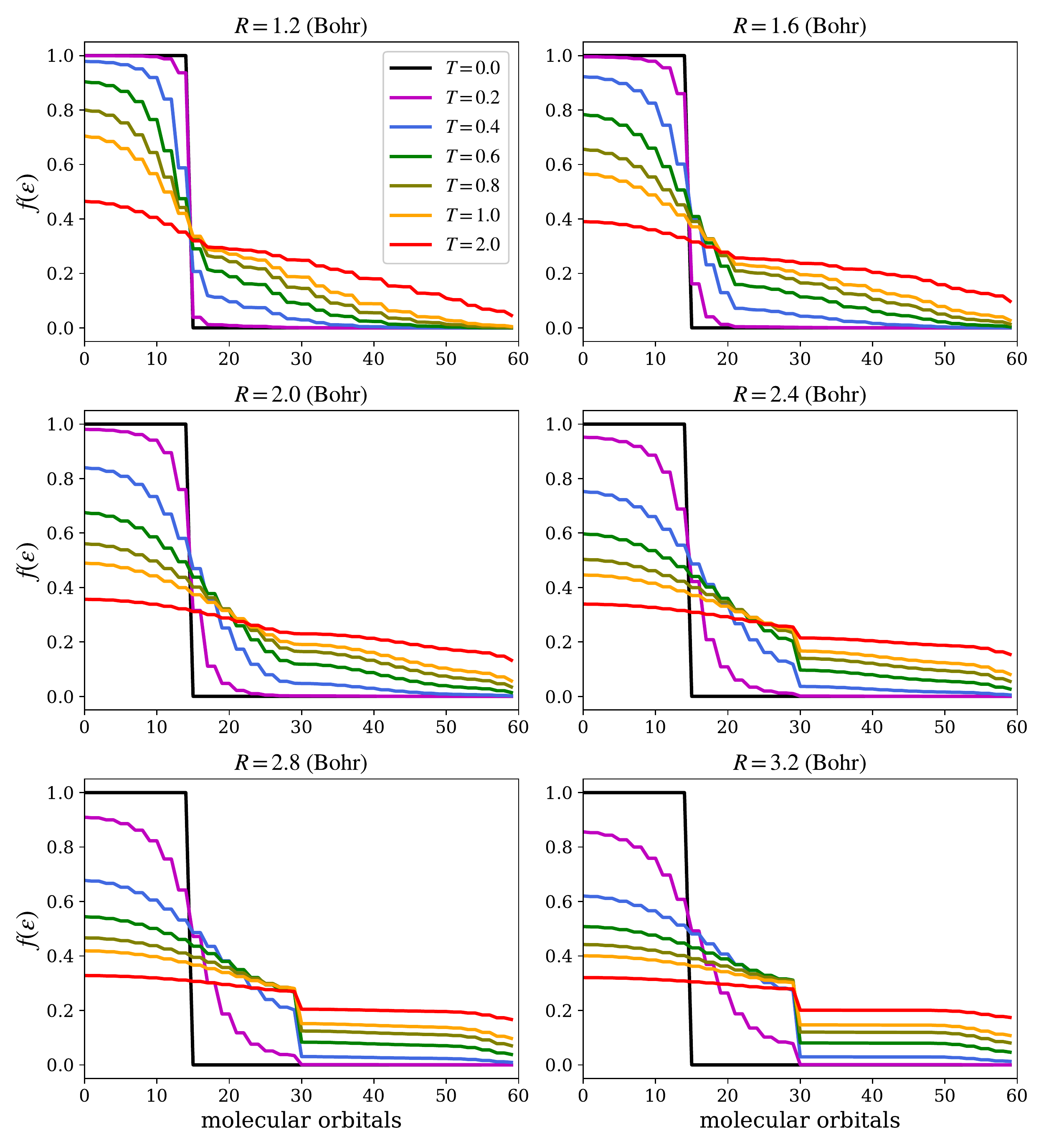}
     \caption{Fermi-Dirac distribution of electrons on Hartree-Fock molecular 
    orbitals for H$_{30}$ chain with STO-6G basis. $T$ is in unit Hartree.}\label{fig:fd_curve}
\end{figure}

\subsection{Embedding Hamiltonian}
There are two choices of constructing the embedding Hamiltonian: (i)
interacting bath formalism and (ii) non-interacting bath formalism
~\citep{WoutersJCTC2016}. We pick the interacting bath formalism to
restore most of the two-body interactions. The embedding Hamiltonian
constructed from interacting bath formalism has the form
\begin{equation}
H_{\text{emb}} = \sum{pq}F^{\text{emb}}_{pq}c^{\dag}_pc_q - \mu\sum_{p\in \text{imp}} c^{\dag}_pc_p + \frac{1}{2}\sum_{pqrs}\left(pq|rs\right)c^{\dag}_pc^{\dag}_r
c_sc_q.
\end{equation}
Note that we use $p,q,r,s$ to index embedding orbitals and $i,j,k,l$
to index lattice orbitals. A chemical potential $\mu$ is added to only 
apply on the impurity, making sure that the number of electrons on 
the impurity is correct during the DMET self-consistency.

The embedding Fock matrix $F^{\text{emb}}$ is obtained by projecting
the lattice Fock in AOs to the embedding orbitals. Using
$\tilde{P}^{\mbf{k}} = C^{\mbf{k}}P^{\mbf{k}}$ to denote the projection
operator, one computes the embedding Fock matrix by
\begin{equation}\label{eq:fock_rotate}
\tilde{F} = \frac{1}{N_{\mbf{k}}}\sum_{\mbf{k}}\tilde{P}^{\mbf{k}\dagger}
F^{\mbf{k}} \tilde{P}^{\mbf{k}}
\end{equation}
where $F^{\mbf{k}}$ is the lattice Fock matrix in $\mbf{k}$-space AO basis. 
To avoid double counting, we subtract the contribution of the embedding
 electron repulsion integrals (ERIs) from $\tilde{F}$
\begin{equation}
F^{\text{emb}}_{pq} = \tilde{F}_{pq} - 
\left[\sum_{rs} \left(pq|rs\right)\gamma^{\text{emb}}_{sr} - \frac{1}{2}  \left(pr|sq\right)\gamma^{\text{emb}}_{rs}\right]
\end{equation}
where $\gamma^{\text{emb}}$ is the lattice density matrix rotated to the
embedding space.

The time-consuming part is the construction and projection of the two-electron
ERIs to the embedding space. To reduce the cost, we use density 
fitting~\citep{Whitten1973,Sun2017} to convert the four-center ERIs to the three-center ERIs,
\begin{equation}
\left(\mu\mbf{k}_{\mu} \nu\mbf{k}_{\nu}| \kappa\mbf{k}_{\kappa}\lambda
\mbf{k}_{\lambda}\right) \approx \sum_{L\mbf{k}_L} \left(\mu\mbf{k}_{\mu} \nu\mbf{k}_{\nu}| L\mbf{k}_L\right) \left(L\mbf{k}_L|\kappa\mbf{k}_{\kappa}\lambda
\mbf{k}_{\lambda}\right)
\end{equation}
where $L\mbf{k}_L$ is the auxiliary basis and only three $\mbf{k}$ indices are
independent due to the conservation of momentum: $\mbf{k}_{L} = \mbf{k}_{\mu}
- \mbf{k}_{\nu} +n\mbf{q}$ ($n\mbf{q}$ is the integer multiple of reciprocal lattice vectors). The auxiliary basis
used in this work is a set of chargeless Gaussian crystal orbitals
with the divergent part of the Coulomb term treated in Fourier space~\citep{Sun2017}. Density fitting with the above auxiliary basis is called 
 Gaussian density fitting (GDF). In practice, we first transform 
three-center ERIs from the lattice orbitals to the embedding orbitals
with cost $\mathcal{O}\left(n_{\mbf{k}}^2n_L n_{\text{lat}}n_{\text{emb}^2}
+ n_{\mbf{k}}^2n_L n_{\text{lat}}^2n_{\text{emb}}\right)$; then we convert
the three-center ERIs back to the four-center ERIs in the embedding
space with cost $n_{\mbf{k}}n_Ln_{\text{emb}}^2$. The computational
cost is significantly reduced compared to direct transformation 
with cost $\mathcal{O}\left(n_{\mbf{k}^3n_{\text{lat}^5}}\right)$.

\begin{figure}[t!]
\centering
\begin{subfigure}[t]{0.7\textwidth}
\includegraphics[width=\textwidth]{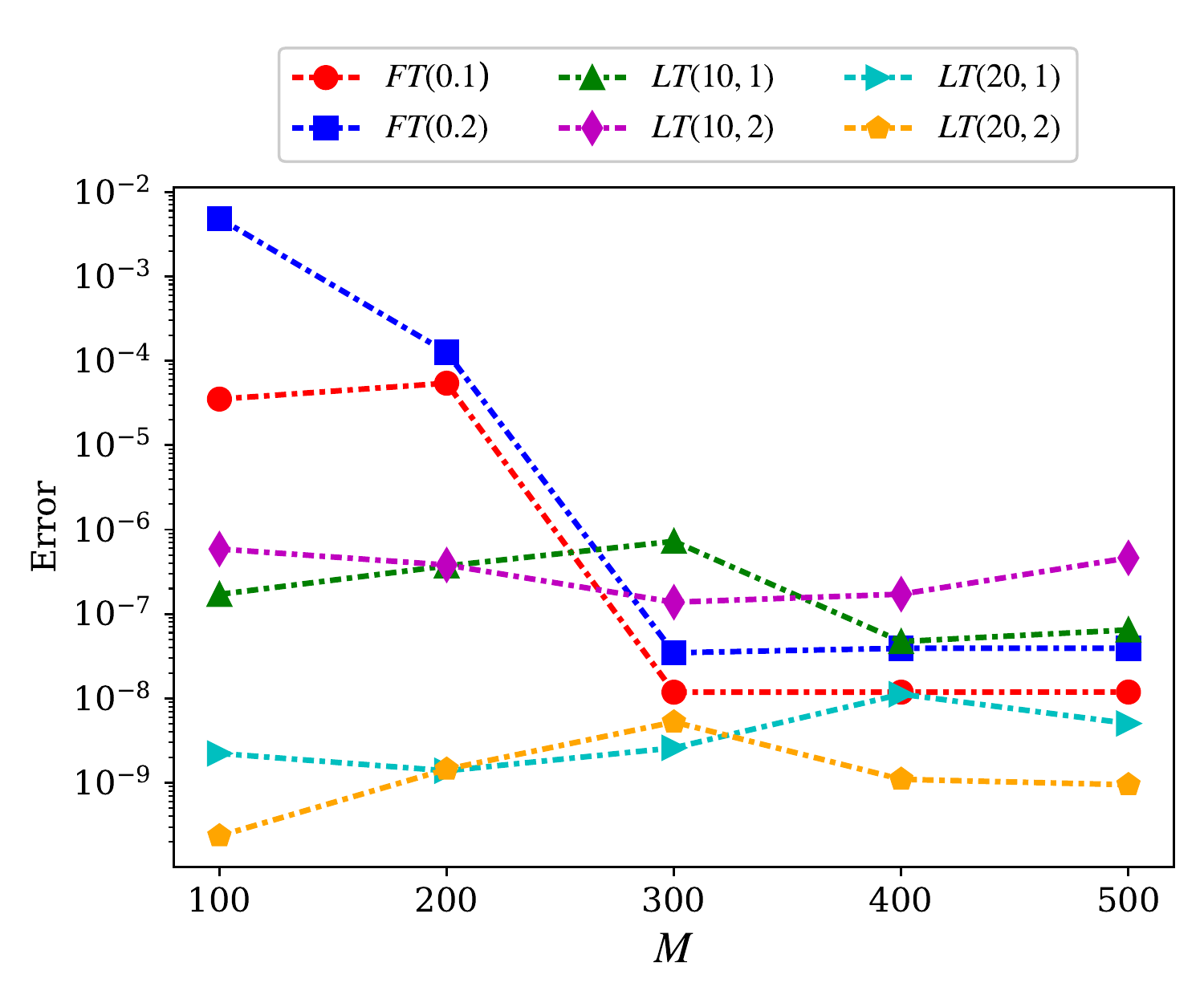}
\caption{$R = 1.5$ Bohr}
\end{subfigure}
\begin{subfigure}[t]{0.7\textwidth}
\includegraphics[width=\textwidth]{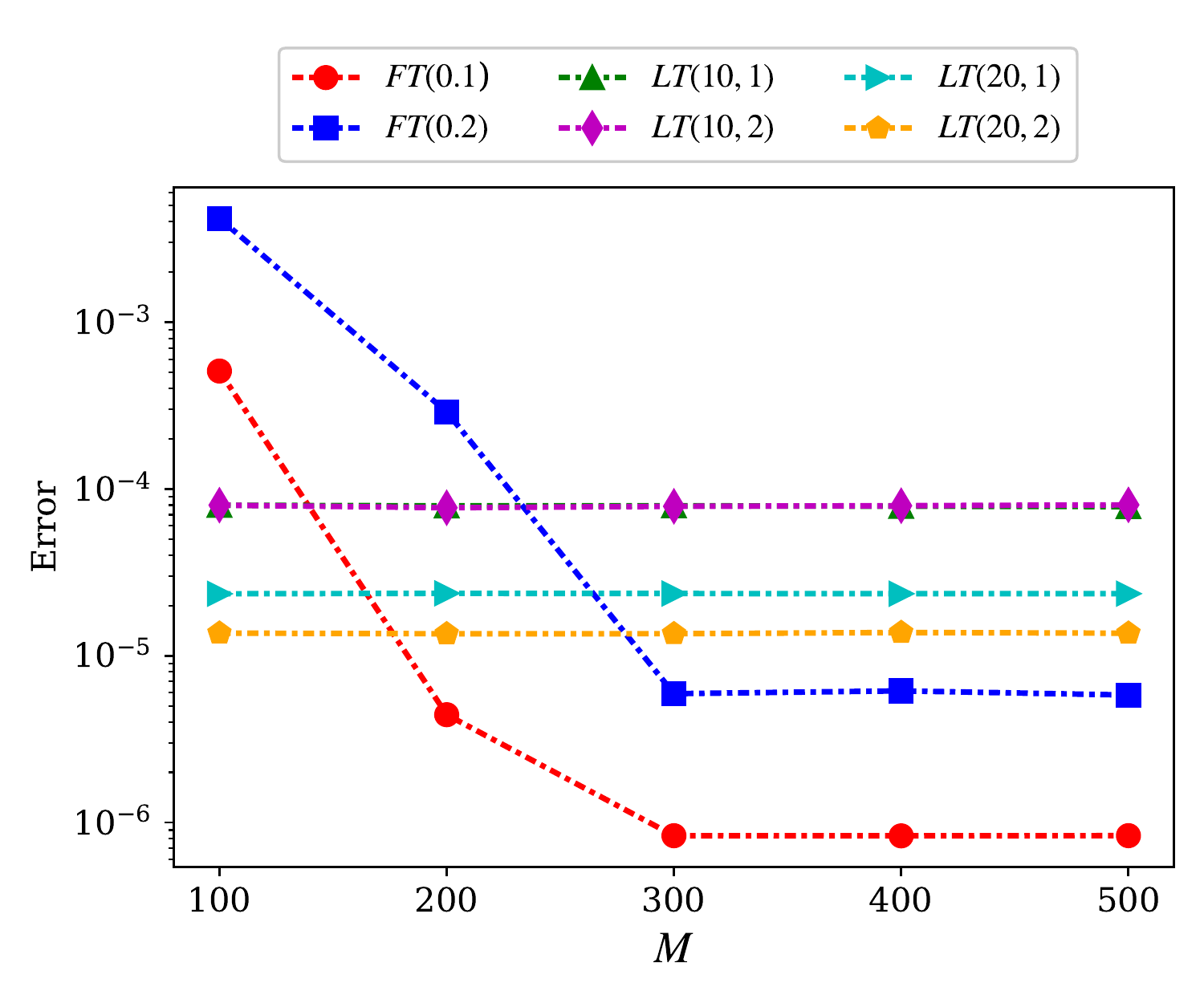}
\caption{$R = 3.0$ Bohr}
\end{subfigure}
\caption{Accuracy test on FT-DMRG and LT-DMRG solvers against exact diagonalization. The label "FT($x$)" stands for FT-DMRG solver with $\tau = x$, and 
the label "LT($x, y$)" stands for LT-DMRG solver with $x$ Davidson roots and 
$y$ electron deviations from half-filling for both spins.}\label{fig:dmrg_solvers}
\end{figure}

\subsection{Impurity solver}
An accurate finite temperature algorithm is required as the impurity
solver. In this work, we use homemade finite temperature
exact diagonalization (FT-ED) and finite temperature density matrix
renormalization group (FT-DMRG) for small and large impurity 
problems, respectively. In particular, there are two ways to implement 
FT-DMRG: (1) imaginary time evolution from an enlarged Hilbert space,
also known as the purification approach~\citep{Feiguin2005} (referred as
FT-DMRG) ; and (2)
using Davidson diagonalization to generate a set of low-energy levels
to be used in the grand canonical statistics (referred as low temperature
DMRG, LT-DMRG). While FT-DMRG can be used for the whole temperature spectrum,
LT-DMRG is especially for low temperature calculations. Because most
of the phase transitions happen at the low temperature regime, LT-DMRG can provide
accurate enough calculations with lower cost compared to FT-DMRG. 

Since FT-DMRG is based on imaginary time evolution from inverse
temperature $\beta = 0$, the entanglement grows rapidly as $\beta$ increases,
and at low temperature a bond dimension that is much larger than the 
ground state bond dimension is required. Another error source of FT-DMRG
is the imaginary time step $\tau = \beta / N$, where $N$ is the number
of time steps. For symmetrized Trotter-Suzuki approximation, the local
trucation error is on the order of $\mathcal{O}(\tau^3)$, while the total
accumulated error is on the order of $\mathcal{O}(\tau^2)$. If the 4th order
Runge-Kutta (RK4) method is used, the local truncation error is on the 
order of $\mathcal{O}(\tau^5)$ and the total
accumulated error is on the order of $\mathcal{O}(\tau^4)$. The FT-DMRG
used in the calculations of this work used the RK4 method.
Note that since
the matrix product state (MPS) truncation is applied at every time step,
having a too small $\tau$ will lead to a large accumulation of MPS truncation 
errors. Therefore, one needs to choose the $\tau$ value to be not too small
to introduce large MPS truncation errors and not too big to introduce large
time evolution truncation errors. The error source of LT-DMRG is from the 
truncation of the grand canonical summation and the number of roots in 
the Davidson diagonalization. Generally the ground state bond dimension is 
enough for the low temperature calculations with LT-DMRG.

An assessment of the accuracy of FT-DMRG and LT-DMRG solvers is shown in 
Fig.~\ref{fig:dmrg_solvers}. The embedding system is composed of two impurity
orbitals and two bath orbitals, generated from a $6$-site hydrogen chain 
with the STO-6G basis at $R = 1.5$ and $3.0$ Bohr at $\beta = 20$. Exact 
diagonalization (ED) is chosen as the exact reference. 
In Fig.~\ref{fig:dmrg_solvers}, we try to understand the role of imaginary
time step $\tau$ in FT-DMRG solver and the roles of the number of 
Davidson roots and the size of the truncated grand canonical space in the
LT-DMRG solver. At $\beta = 20$, the smaller $\tau$ (red lines) gave a
smaller error compared to $\tau = 0.2$ (blue lines), and the FT-DMRG
results converged at $M \sim 300$. The errors of LT-DMRG solver do not change
too much with the bond dimension $M$, so $M=100$ is already enough for 
a $4$-site system. The accuracy at $R = 1.5$ Bohr is generally better than 
$R = 3.0$ Bohr, since larger $R$ corresponds to stronger correlation. 
At larger $R$, one needs to include a larger number of Davidson roots to 
achieve high enough accuracy with LT-DMRG.  
Generally with large enough bond dimension $M$, FT-DMRG could provide
more accurate results. However, when the embedding system is too large
to use a large $M$, one could consider the LT-DMRG method. In DMET calculations,
we use $ED$ solver for $L_{\text{emb}} < 8$ embedding problems and use
FT-DMRG solver for larger problems.

\subsection{Thermal observables}
In order to identify the metal-insulator transition and the N\'eel transition 
and explore the mechanism behind the crossings, we compute the following
order parameters: staggered magnetic moment $m$, double occupancy $D$, 
complex polarization $Z$,
spin-spin correlation functions $\mathcal{C}_{ss}$, and 
charge-charge correlation functions $\mathcal{C}_{cc}$. 

\noindent\emph{Staggered magnetic moment.} The staggered magnetic momentum
is calculated as
\begin{equation}\label{eq:hlatt_mag}
m = \frac{1}{N^{\text{imp}}}\sum_{i\in \text{imp}}\frac{|n_{i, \uparrow} - n_{i, \downarrow}|}{2},
\end{equation}
where $N^{\text{imp}}$ is the total number of H atoms in
the impurity (supercell), and $n_{i, \uparrow}$
and $n_{i, \downarrow}$ are electron numbers on $i$th atom with up spin 
and down spin, respectively. Note that if one uses Bohr magneton 
$\mu_B = \frac{e\hbar}{2m_e} = \frac{1}{2}$ as the unit, then one would
drop $2$ in the denominator in Eq.~\eqref{eq:hlatt_mag}. 
To evaluate $n_{i,\uparrow}$ on the $i$th atom, we first 
compute the one-particle impurity density matrix for up-spin 
in the IAO basis, and then 
sum up the diagonal terms that belong to the $i$th atom. For example,
when STO-6G basis is used, the occupation numbers on $1s$ orbital and 
$2s$ orbital of atom-$i$ sum up to the electron density on atom-$i$.
$n_{i,\downarrow}$ is evaluated in the same way from the down-spin
one-particle impurity density matrix.

\noindent\emph{Double occupancy.} The double occupancy measures the probability
of two electrons with opposite spins occupying the same hydrogen atom,
calculated by
\begin{equation}\label{eq:docc_hlatt}
D = \frac{1}{N^{\text{imp}}}\sum_{i\in \text{imp}} \langle \hat{n}_{i\uparrow}\hat{n}_{i, \downarrow}\rangle.
\end{equation}
Note that the hat on $\hat{n}_{i\uparrow}$ denotes that it is an operator, not
a number, with $n_{i\uparrow} = \langle \hat{n}_{i\uparrow}\rangle$. 
Since there are multiple bands on each atom, we expand $\hat{n}_{i\uparrow}$
as 
\begin{equation}
\hat{n}_{i\uparrow} = \sum_w \hat{n}^w_{i\uparrow},
\end{equation}
where $w$ is the index of the bands on the $i$-th atom ( e.g., $1s$, $2s$,
$2p_x, 2p_y, 2p_z, ...$). Therefore, the precise expression of double occupancy
becomes
\begin{equation}
D = \frac{1}{N^{\text{imp}}}\sum_{i\in \text{imp}}\sum_{w, w'\in i} 
\langle \hat{n}^{w}_{i\uparrow}\hat{n}^{w'}_{i, \downarrow}\rangle.
\end{equation}

\noindent\emph{Complex polarization.}
Complex polarization measures the mobility of electrons, and thus can be
used as an indicator of metal-insulator transition. The definition of
complex polarization on $z$ direction is
\begin{equation}
Z = \langle e^{i\frac{2\pi}{L}\hat{z}} \rangle,
\end{equation}
where $L$ is the chain length and $\hat{z}$ is the location operator 
in the $z$-direction. When $Z = 0$, electrons are delocalized and the system is metallic; when 
$Z = 1$, electrons are localized and the system is insulating. 
At mean-field level, the ground state is a Slater determinant $|\phi\rangle$
 of occupied
orbitals, so the complex polarization is evaluated by
\begin{equation}
Z = \langle \phi| e^{i\frac{2\pi}{L}\hat{z}}|\phi \rangle,
\end{equation}
which is equivalent to 
\begin{equation}
Z = \text{Det}\left[C_{\text{occ}}^{\dag} e^{i\frac{2\pi}{L}z} C_{\text{occ}}\right],
\end{equation}
where $C_{\text{occ}}$ represents the coefficients of occupied orbitals.

At finite temperature, we use a thermofield approach~\citep{Harsha2019} 
from our recent work (see Chapter~\ref{chp:cp}). We construct the infinite temperature
determinant with an enlarged Hilbert space $\tilde{\phi}$, and thermofield
operators of the Hamiltonian $\tilde{H}$ and position operator $\tilde{z}$.
Then the finite temperature complex polarization is evaluated by
\begin{equation}
Z(\beta) = \frac{1}{\mathcal{Z}} \langle \tilde{\phi} | e^{-\beta(\tilde{H})}e^{i\frac{2\pi}{L}\tilde{z}} | \tilde{\phi}\rangle,
\end{equation}
where $\mathcal{Z} \langle \tilde{\phi} | e^{-\beta(\tilde{H})} | \tilde{\phi}\rangle$ is the partition function.

\noindent\emph{Spin-spin correlation and charge-charge correlation functions.}
We define the two correlation functions as follows:
\begin{equation}
\begin{split}
\mathcal{C}^{ss}_i =& \langle (\hat{n}_0^{\uparrow} - \hat{n}_0^{\downarrow})
( \hat{n}_i^{\uparrow} - \hat{n}_i^{\downarrow})\rangle
                    - \langle \hat{n}_0^{\uparrow} - \hat{n}_0^{\downarrow}\rangle
\langle \hat{n}_i^{\uparrow} - \hat{n}_i^{\downarrow}\rangle,\\
\mathcal{C}^{cc}_i =&  \langle (\hat{n}_0^{\uparrow} + \hat{n}_0^{\downarrow})
( n_i^{\uparrow} + n_i^{\downarrow})\rangle
                    - \langle \hat{n}_0^{\uparrow} + \hat{n}_0^{\downarrow}\rangle
\langle n_i^{\uparrow} + n_i^{\downarrow}\rangle,
\end{split}
\end{equation}
where $\hat{n}_i^{\sigma}$ is the electron density operator of spin $\sigma$
on site $i$.

\section{Results}\label{sec:res_hlatt}
In this section, we show some preliminary results of FT-DMET calculations
on the hydrogen chain system with periodic boundary condition. 
First, we examine the basis set effect on a $22$-atom chain with $2$ atoms
in the impurity. Fig.~\ref{fig:mag_22k} shows the magnetic moment at
both ground state and $T=0.02$ Hartree calculated by DMET with STO-6G, 
6-31G, and 
CC-PVDZ basis sets. Paramagnetic-antiferromagnetic
(PM-AFM) transition is observed at both ground state and $T=0.02$ Hartree. A very interesting behavior of the magnetic moment at
ground state is observed: with STO-6G, the magnetic moment drops at
$R > 3.0$ Bohr, while with larger basis sets, this drop did not happen.
The reason for the above behaviors could be due to the loss of 
entanglement between different sites at large $R$. Imagine at
$R = \infty$, the system should behave as $22$ individual atoms, and one 
should expect the ground state to be paramagnetic. With more diffused orbitals
(e.g., $2s$ and $2p$ orbitals), however, the entanglement between different 
sites decays slower as $R$ increases. Note that since the impurity size
is only $2$ atoms, one only needs to consider the entanglement 
between adjacent sites. Once the impurity gets larger, a drop of magnetic
moment with larger basis sets should also be expected. At $T = 0.02$ (left
panel), the magnetic moment computed with all three basis sets dropped
as $R$ increases, as a consequence of thermal dissipation.

\begin{figure}[t!]
\centering
\begin{subfigure}[t]{0.7\textwidth}
\includegraphics[width=\textwidth]{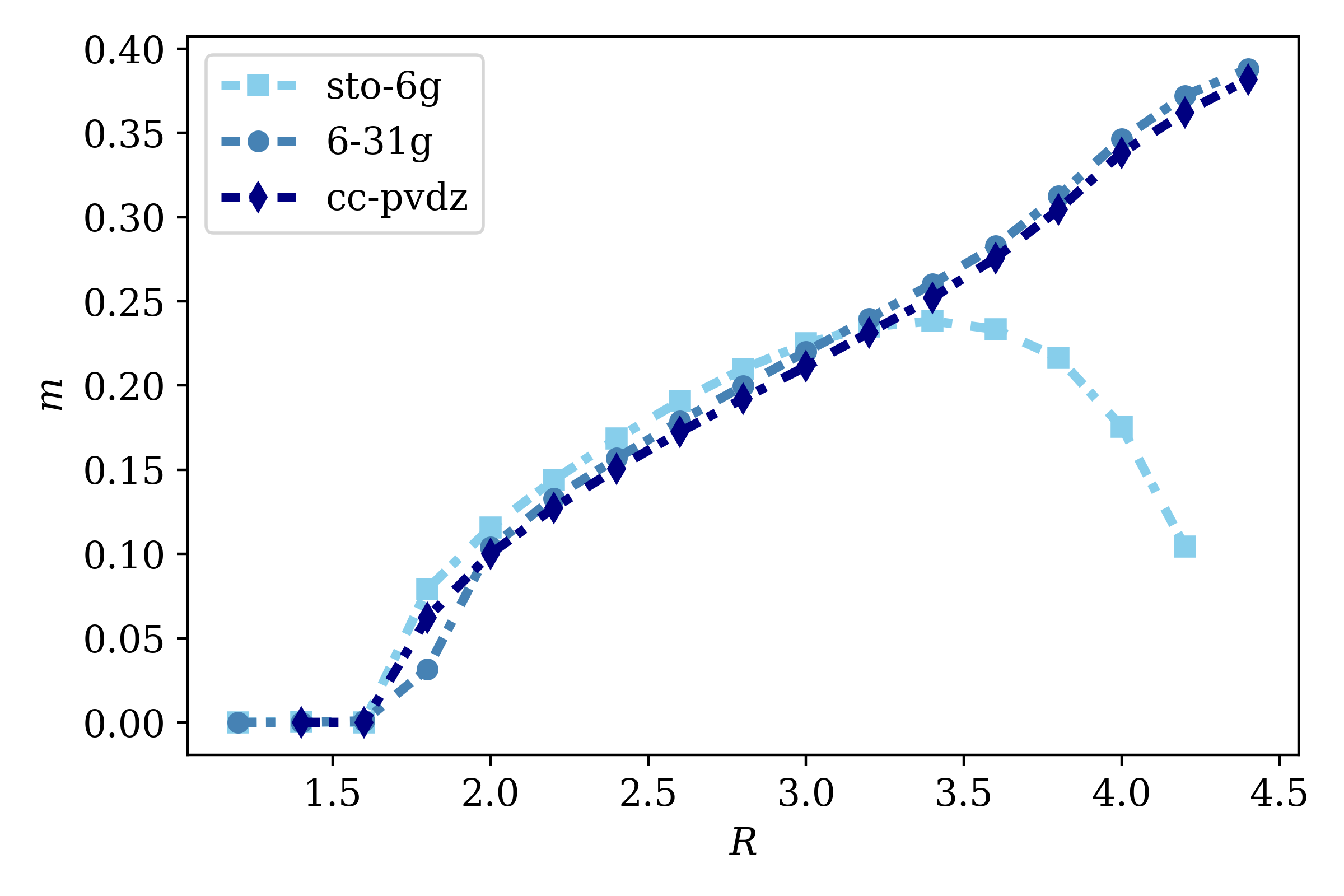}
\caption{ }
\end{subfigure}
\begin{subfigure}[t]{0.7\textwidth}
\includegraphics[width=\textwidth]{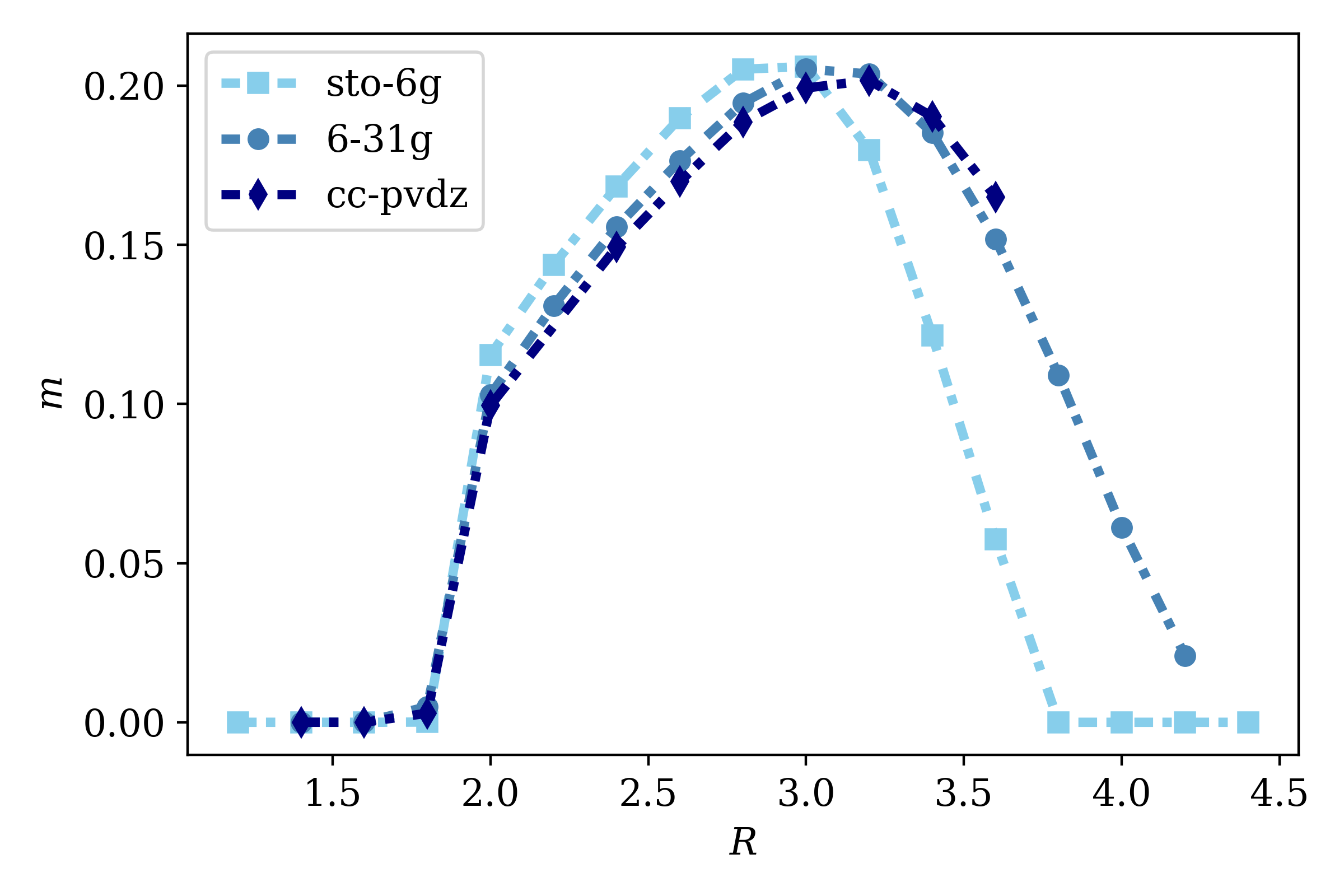}
\caption{ }
\end{subfigure}
\caption{Magnetic moment of a $22$-atom chain at ground state (left panel)
 and $T=0.02$ Hartree (right panel) with STO-6G, 6-31G, and CC-PVDZ basis sets. 
 }\label{fig:mag_22k}
\end{figure}

We further show the double occupancy from the above simulation settings in 
Fig.~\ref{fig:docc_11k}. A clear change of the gradient of $D$ as a function
of $R$ is observed for both ground state and $T=0.02$ Hartree, indicating 
a metal to insulator transition. The transition $R$ is around $1.6\sim 1.8$
Bohr, which agrees with the transition $R$ of PM-AFM transition, resulting
in a PM metal phase and AFM insulating phase.

\begin{figure}[t!]
\centering
\begin{subfigure}[t]{0.7\textwidth}
\includegraphics[width=\textwidth]{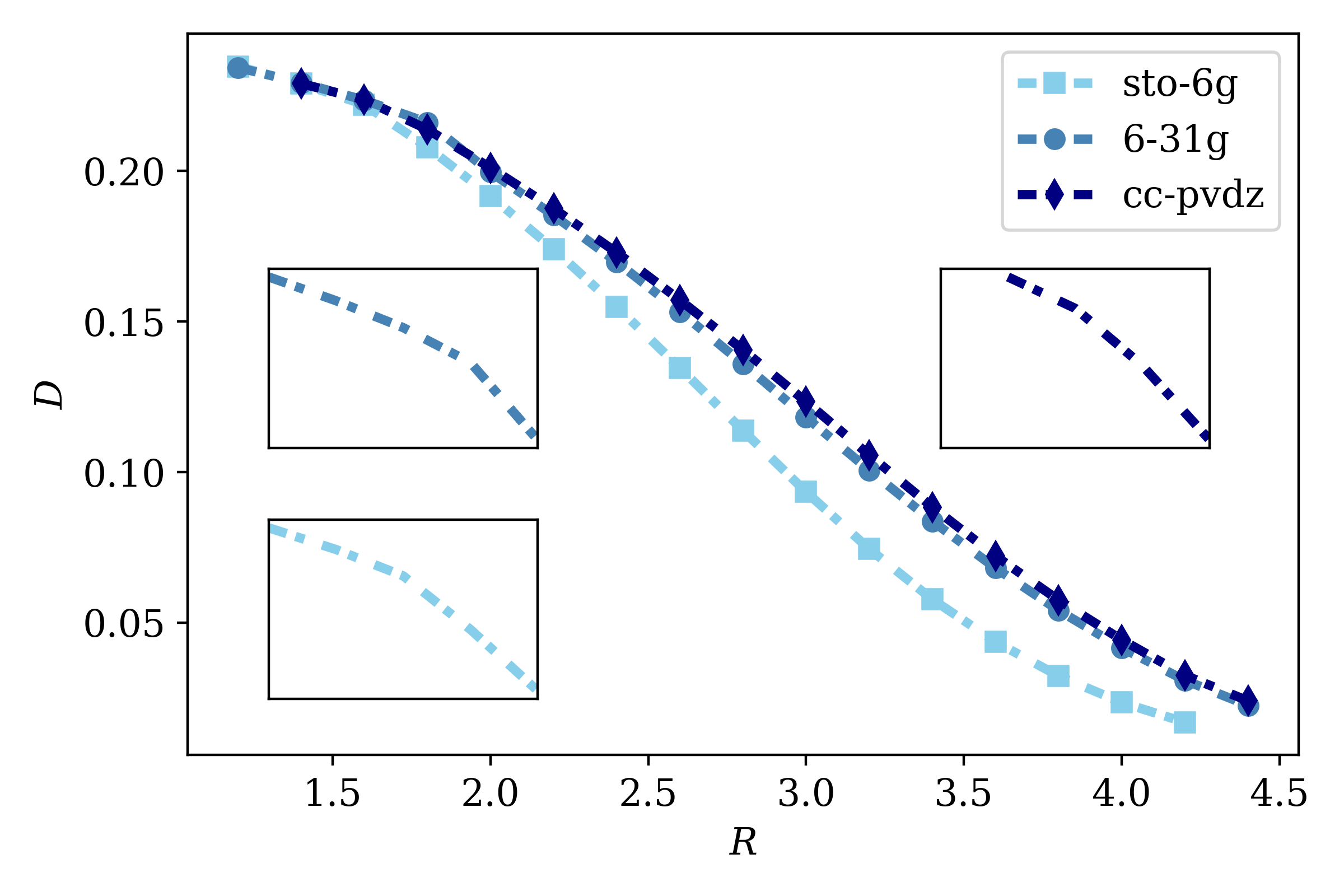}
\caption{ }
\end{subfigure}
\begin{subfigure}[t]{0.7\textwidth}
\includegraphics[width=\textwidth]{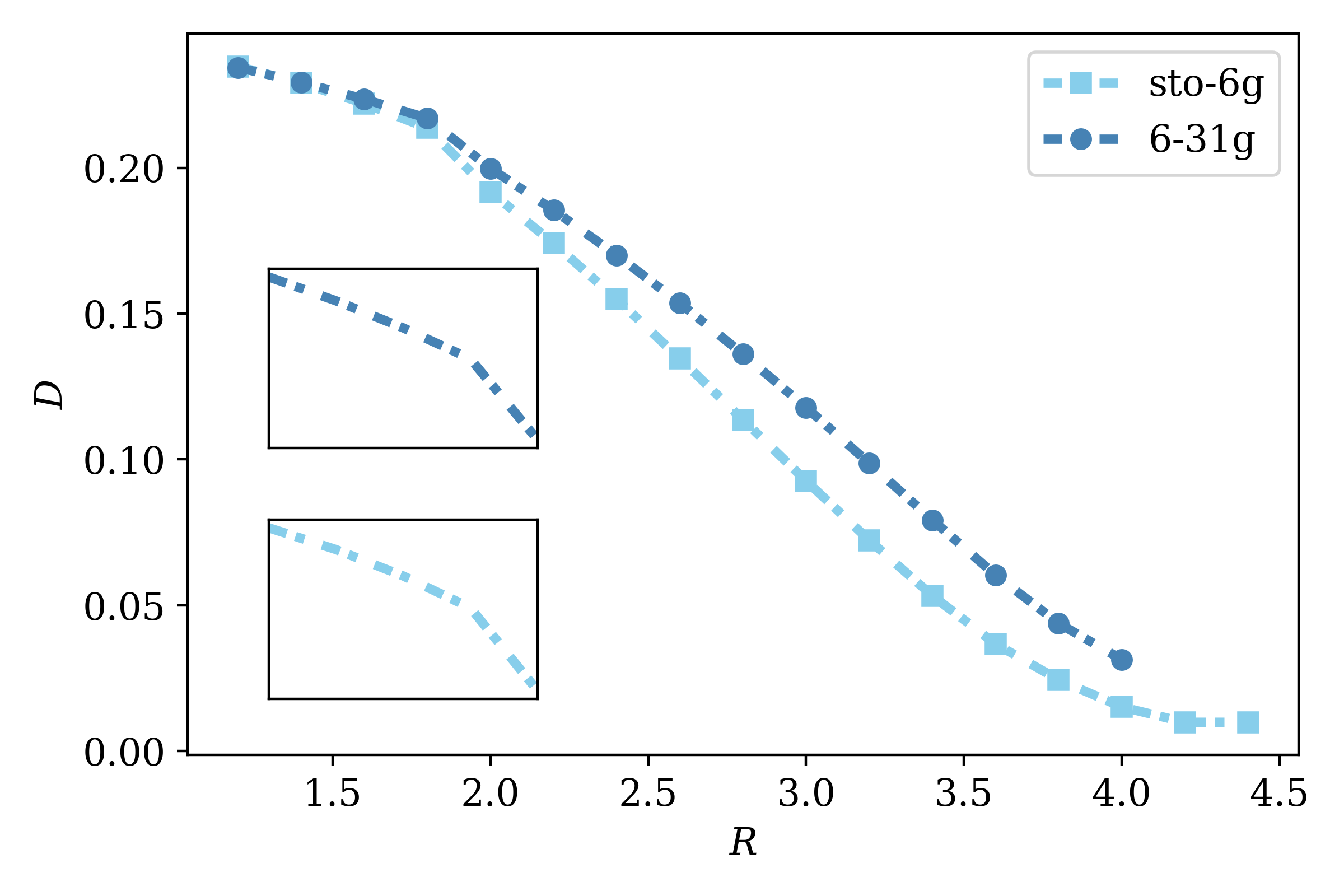}
\caption{ }
\end{subfigure}
\caption{Double occupancy of a $22$-atom chain at ground state (left panel)
 and $T=0.02$ Hartree (right panel) with STO-6G, 6-31G, and CC-PVDZ basis sets. 
The insets show a sudden change of the gradient of $D$ as a function of
$R$, indicating metal to insulator transition.
 }\label{fig:docc_11k}
\end{figure}

%

Next we increase the total number of atoms in the hydrogen chain to  $50$
atoms to eliminate the finite size effect of the total system size.
STO-6G is used as the basis set.
The impurity is
composed of two hydrogen atoms, and solved by ED. 
We first present the 
energy calculations and dissociation curve of the hydrogen chain,  shown in Fig.~\ref{fig:energy_hlatt}, where the 
energy per electron at $T=0.05$ is compared to FT-AFQMC~\citep{Liu2020}
results. The AFQMC calculation used the STO-6G basis set and  a supercell with $10$ atoms and $5$
$k$ points. The two energy curves predicted the same dissociating 
trend and equilibrium point ($\sim 1.8$ Bohr). However, the DMET energies
are all lower than the AFQMC energies, which could be due to the finite
impurity size effect.
\begin{figure}
    \centering
    \includegraphics[width=0.85\textwidth]{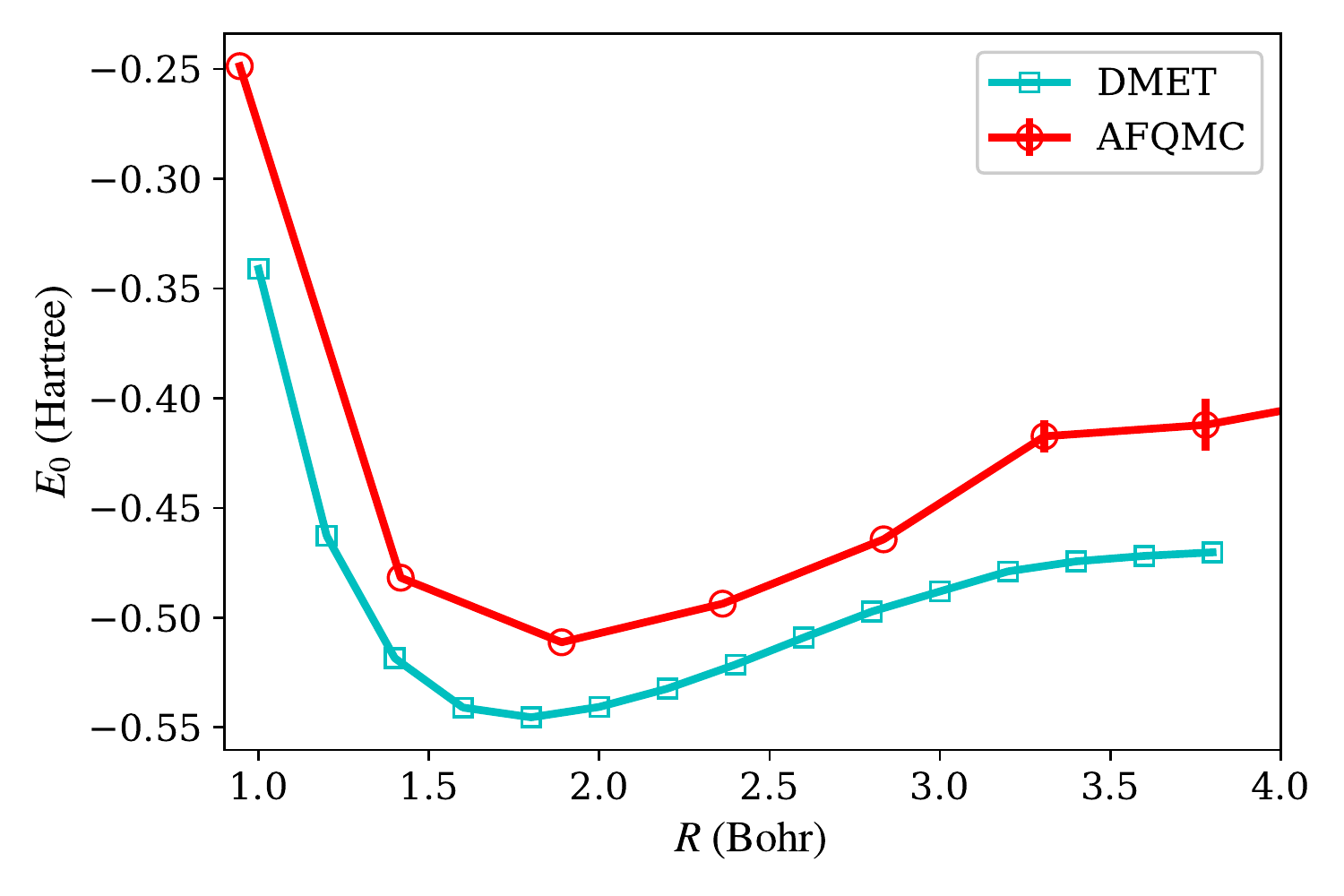}
    \caption{Dissociation curve of hydrogen chain at $T=0.05$ Hartree compared to AFQMC. The AFQMC data is extracted from Ref.~\citep{Liu2020}.
    } \label{fig:energy_hlatt}
\end{figure}

We then examine the staggered magnetic moment $m$ as a function of inter-atomic
distance $R$ at ground state, $T=0.02$, $T=0.05$, and $T=0.1$ in 
Fig.~\ref{fig:mag_hlatt_nk25}. Compared to  Fig.~\ref{fig:mag_22k} where
the paramagnetic-antiferromagnetic (PM-AFM) transition happened around
$R = 1.6$ Bohr, 
 we observed the
PM-AFM transition at $R = 1.0\sim 1.5$ region for $T < 0.1$, which could be
an effect of the finite total system size.
Although $T = 0.02$ is
considered as very close to the ground state, the magnetic moment at 
$T = 0.02$ drops earlier than the ground state curve as $R$ increases. This
behavior is due to the thermal dissipation of the magnetic order. As $R$
grows larger, the atoms are far apart from each other, and thus the 
electron-electron correlation between different sites is weaker, and eventually
not enough to preserve the long-range AFM order, which lead to the 
drop of magnetic order at large $R$ as shown in the figure. Even adding a
small temperature, the flip of the spin can happen to destroy the long-range
AFM order.

\begin{figure}
    \centering
    \includegraphics[width=0.85\textwidth]{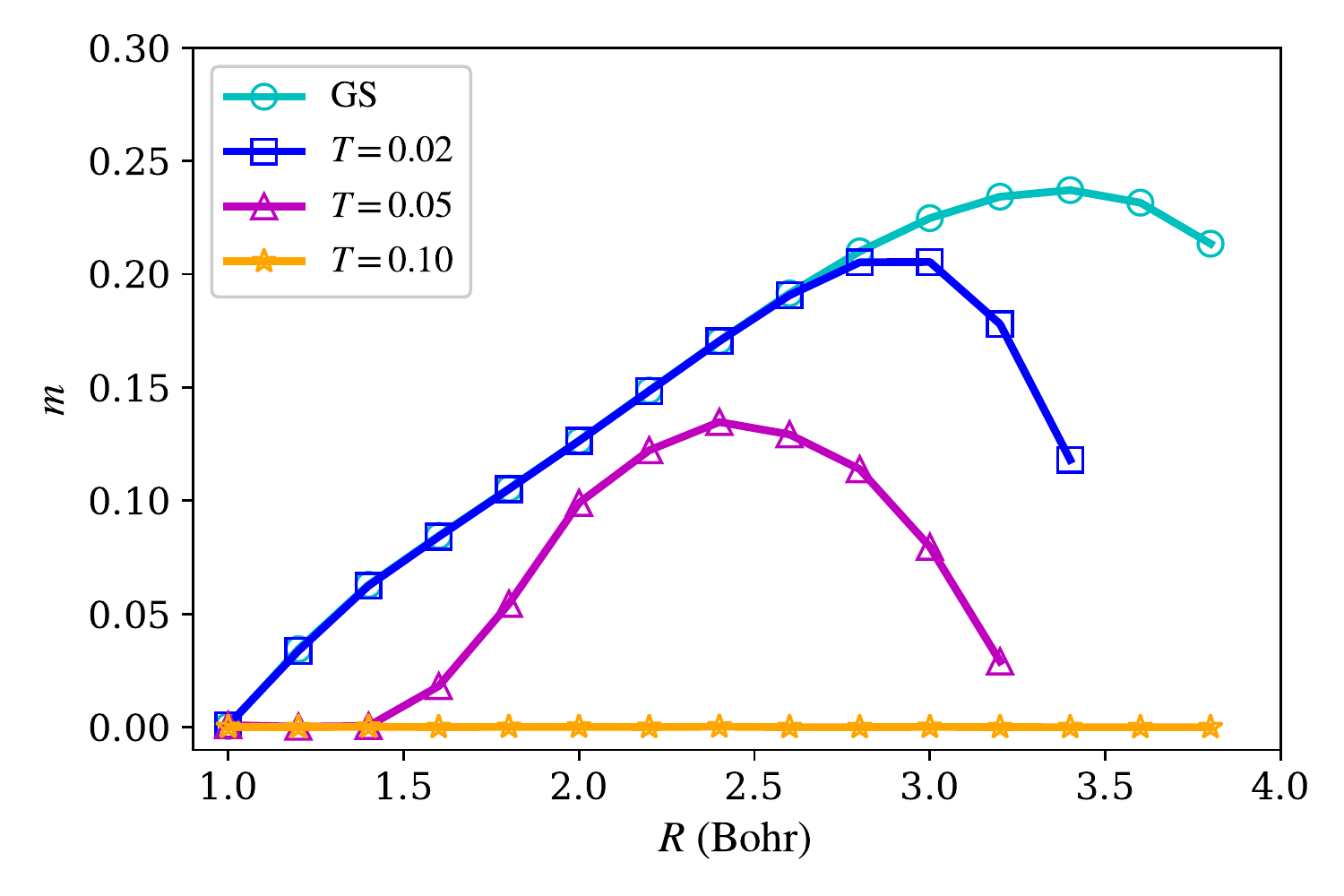}
    \caption{Staggered magnetic moment of hydrogen chain with periodic 
boundary condition at ground state, $T=0.02$, $T=0.05$, and $T=0.1$. The unit of $T$ is Hartree. 
    } \label{fig:mag_hlatt_nk25}
\end{figure}

\section{Conclusion}\label{sec:conc_hlatt}
In this work, we generalized the previously described finite temperature 
density matrix embedding theory to study \textit{ab initio} solid, and
employed the method to study the hydrogen chain problem. 
Despite the simplicity of the hydrogen chain lattice compared to other
periodic systems, it exhibits a variety of intriguing behaviors including
paramagnetic-antiferromagnetic transition and metal-insulator transition
at both ground state and finite temperature. At finite temperature,
we observed thermal dissipation for the magnetic order at large inter-atomic 
distance. We further confirmed the stabilizing effect from multi-band basis
sets. Since this work is not completely finished, in the future we will apply
this finite temperature algorithm to a larger set of \textit{ab initio}
solids including the two-dimensional and three-dimensional hydrogen
lattices, transition metal oxides, and challenging systems such as cuprate-based
 high
temperature superconductors.
%
%


\chapter{Finite temperature complex polarization and metal-insulator transition\label{chp:cp}}
\section{Abstract}
Metal-insulator transition is a fundamental yet complicated
topic in condensed matter physics and material science. 
Band structure theory has been widely used in explaining why insulators are insulating
and why metals are conducting, but can fail to describe insulation phases caused
by strong correlations or disorder. Electron locality/mobility serves as
a more general indicator of metallic and insulating phases: when the system
is metallic, the electrons are delocalized and can flow freely;
when the system is insulating, the electrons are localized. The standard 
deviation (or second cumulant moment) of the electron position 
is used as the order parameter of the electron localization, which 
is directly related to the complex polarization of the system. 
The complex polarization is widely accepted as a new indicator of the
metal-insulator transition at ground state: when the complex polarization 
equals zero,
the second cumulant moment of the position is diverged and the system
sits in the metallic phase; when the complex polarization is nonzero,
the electrons are localized and the insulating phase wins. 
In this work, we present the finite temperature formulation of the complex 
polarization. We also introduce a thermofield approach to compute the
complex polarization with thermal Slater determinant. We demonstrate how
finite temperature complex polarization works as an indicator of 
metal-insulator transition at low temperature with a modified tight binding
model and hydrogen chain system. In the hydrogen chain case, we also 
compare the metal-insulator transition with the paramagnetic-ferromagnetic 
transition, electron population, and energy gap to study the origin of 
the insulating and metallic phases, respectively.

\section{Introduction}
Phase transition happens when a system undergoes a macroscopic change
 from one phase to another phase due to the variation of control parameters
such as temperature, magnetic field, chemical substitution, and pressure. 
Near the critical 
point between the two phases, the physical properties of the bulk changes 
dramatically with respect to even a minor perturbation in control parameters.
Metal-insulator transition (MIT) is among the most common phase transitions,
yet the microscopic cause and the physics behind the phenomena is nontrivial.
From the elementary physics
textbooks, we learned that metals are conducting when an electrical field
is applied, while insulators do not allow electrons to flow freely.
However, this is a rather vague and bipartite definition, which is not 
able to answer questions such as (1) what are the microscopic driving forces
for conductivity? (2) what are the causes for MIT? and (3) how does one
characterize MIT?

In the past, the microscopic featurization of insulators and metals are
generally described by the band structure theory~\citep{Ashcroft76,Kittel2004}.
The band structure theory describes the movement of a single electron 
in a periodic solid, with the mean-field effect from the other electrons.
According to band structure theory, if the Fermi level sits in a band gap, then
the system is insulating; if the Fermi level crosses a band, the 
system is conducting. However, the band structure theory is based on
independent electron approximation and is only limited to crystalline 
systems. The insulating behavior caused by disorder or electron-electron
correlation cannot be captured accurately by the band structure theory~\citep{Kohn1964,Anderson1961,Alexandrov1994,Imada1998}. A more general
description is to use the electron localization to distinguish metal and insulator:
when the electrons are \emph{localized}, the system is insulating, and
when the electrons are \emph{delocalized}, the system is conducting. 
A widely accepted approach to evaluate electron localization is based 
on the theory of polarization~\citep{Resta1992,KingSmith1993,Vanderbilt1993,Resta1993,
Ortiz1994,Resta1994,Resta1998,Resta1999}, where the macroscopic polarization
is connected to Berry phase~\citep{Berry1984}. A more straightforward indicator
of electron localization is the second cumulant moment of the electron 
position operator describing 
the spread of electrons~\citep{Resta1999,Resta1999PRL,KingSmith1993,
Vanderbilt1993,Resta1993,}. A value that connects to both the macroscopic 
polarization and the second cumulant moment is complex polarization:
 the phase of the complex polarization is the 
Berry phase, while the second cumulant moment can be evaluated from the
modulo of the complex polarization~\citep{Resta2002,Souza2000,Aebischer2001}.
Moreover, the DC conductivity according to the ground state 
Kubo formula~\citep{Kubo1957} is also related to the modulo of the complex 
polarization. 

Ground state complex polarization and its connection to macroscopic
polarization, electron localization, and DC conductivity have been 
thoroughly studied and discussed in the past~\citep{Resta1999, Resta1999PRL, Souza2000}. However, discussion about finite temperature complex polarization 
is rare, regardless of the significance of this parameter. In this work, we 
introduce the formulation of finite temperature complex polarization
and discuss its relationship with electron localization. We also 
present a mean-field level implementation of finite temperature 
complex polarization under thermofield theory 
~\citep{Matsumoto1983,Semenoff1983,Evans1992,Harsha2019}.
In Section~\ref{sec:cplx}, we introduce the ground state formulation
of complex polarization and electron localization, where we first use 
the single particle case as a simplified example and then generalize
the single particle case to many-body mean-field formulation. 
In Section~\ref{sec:ftcplx}, we extend the ground state formulation
to the finite temperature version, and present a thermofield
approach to evaluate the complex polarization. In Section~\ref{sec:cptb},
we apply the finite temperature formulation to a modified tight binding model
both analytically and numerically, presenting a preliminary analysis of
how complex polarization provides information of metal-insulator transition.
In Section~\ref{sec:cphydrogen}, we choose hydrogen chain as an example 
of computing finite
temperature complex polarization for
\textit{ab initio} systems and explore the temperature-induced metal-insulator
transition in the hydrogen chain. We finalize this article
with a summary and outlook in Section~\ref{sec:cpconc}.

\section{Ground state complex polarization and electron localization \label{sec:cplx}}
The many-body complex polarization $Z_N^{(\alpha)}$ was first introduced as the 
ground state expectation value of certain unitary many-body operators
~\citep{Resta1998}. We start by defining a general form of the unitary
many-body operator
\begin{equation}
\hat{U}(\mbf{k}) = e^{i\mbf{k}\cdot \hat{\mbf{R}}},
\end{equation} 
where $\mbf{k}$ is an arbitrary three-dimension vector and $\hat{\mbf{R}}$
is the three-dimensional position operator, with 
$\hat{R}^{\alpha}\psi(r_1, r_2, r_3) = r_{\alpha} \psi(r_1, r_2, r_3), 
\alpha = 1,2,3$.

We introduce three $\mbf{k}$ vectors with notation $\mbf{\kappa}^{(\alpha)} 
(\alpha = 1,2,3)$, defined as
\begin{equation}
\kappa^{(\alpha)}_{\beta} = \frac{2\pi}{L}\delta_{\alpha\beta},
\end{equation}
which can be explicitly written as
\begin{equation}
\begin{split}
\kappa^{(1)} &= \left(\frac{2\pi}{L}, 0, 0 \right),\\
\kappa^{(2)} &= \left(0, \frac{2\pi}{L}, 0 \right),\\
\kappa^{(3)} &= \left(0, 0, \frac{2\pi}{L}\right).
\end{split}
\end{equation}
The ground state complex polarization is then defined as the expectation
values of the unitary many-body operators with the above three vectors:
\begin{equation}\label{eq:cplx_gs}
Z_N^{(\alpha)} = \langle \Psi_0 | \hat{U}(\mbf{\kappa}^{\alpha}) | \Psi_0\rangle,
\end{equation}
where $N$ is the number of electrons and
 $|\Psi_0\rangle$ represents the ground state of the system of interest.  
The complex polarization $Z_N^{(\alpha)}$ in Eq.~\eqref{eq:cplx_gs}
can be explicitly written as 
\begin{equation}\label{eq:phase_cplx}
Z_N^{(\alpha)} = |Z_N^{(\alpha)}|e^{i\gamma^{(\alpha)}_N},
\end{equation}
where $|Z_N^{(\alpha)}|\in [0,1]$ is the modulo of $Z_N^{(\alpha)}$
and $\gamma^{(\alpha)}_N$ is the phase of $Z_N^{(\alpha)}$,
referred as the Berry phase. In this article, we will not discuss
the macroscopic polarization, therefore the phase of Eq.~\eqref{eq:phase_cplx}
will not be mentioned. In fact, with a centrosymmetric choice of the 
origin, the complex polarization will always remain real.

\subsection{Electron localization}
We start by considering a problem with one particle in a one-dimensional 
potential wall. The locality of the particle can be measured by the
quadratic spread, or the second cumulant moment of the position $x$,
defined as
\begin{equation}\label{eq:scm_cp}
\langle \delta x^2 \rangle = \langle \psi | x^2 |\psi\rangle - 
\langle \psi | x| \psi\rangle^2,
\end{equation}
where $|\psi\rangle$ is the ground state of the particle in a box.
$\langle \delta x^2 \rangle$ is finite when the state $|\psi\rangle$ is
bounded and diverges for an unbounded state. Let $n(x) = |\psi(x)|^2$ 
be the electron density, then we can rewrite Eq.~\eqref{eq:scm_cp} as
\begin{equation}
\langle \delta x^2 \rangle = \int_{-\infty}^{\infty} \mathrm{d}x \ x^2 n(x) 
 - \left(\int_{-\infty}^{\infty} \mathrm{d}x \ x n(x)\right)^2.
\end{equation}
We now assume that $\psi(x)$ is periodic with wavelength $L$
\begin{equation}
\psi(x + mL) = \psi(x),
\end{equation}
where $m$ is an integer. The Fourier transformation of $n(x)$ gives
\begin{equation}
\tilde{n}(k) = \int_{-\infty}^{\infty} e^{-i kx} n(x).
\end{equation}

We chose the origin to be $x_0$ so that $\langle x\rangle = 0$, then 
\begin{equation}\label{eq:first_cp}
\int_{-\infty}^{\infty} \mathrm{d}x \ x n(x) = 
-i \frac{\mrm{d}\tilde{n}(k)}{\mrm{d} k} \Biggr\vert_{k=0} = 0,
\end{equation}
and $\langle \delta x^2 \rangle$ is only evaluated from the average value
of $x^2$
\begin{equation}\label{eq:second_cp}
\langle \delta x^2 \rangle = \int_{-\infty}^{\infty} \mathrm{d}x \ x^2 n(x) 
 = -\frac{\mrm{d}^2 \tilde{n}(k)}{\mrm{d} k^2 }\Biggr\vert_{k=0}.
\end{equation}
Combining Eq.~\eqref{eq:first_cp} and Eq.~\eqref{eq:second_cp}, we can
approximate $\tilde{n}(k)$ with the Taylor expansion up to the second 
order
\begin{equation}\label{eq:taylor_cp}
\tilde{n}(k) \approx 1 - \frac{1}{2} \langle \delta x^2 \rangle k^2.
\end{equation}
Now we can write down the single particle complex polarization modified
from Eq.~\eqref{eq:cplx_gs} as
\begin{equation}\label{eq:single_cplx}
z = e^{i\frac{2\pi}{L}x_0} \tilde{n}\left(-\frac{2\pi}{L}\right).
\end{equation}
Combining Eq.~\eqref{eq:taylor_cp} and Eq.~\eqref{eq:single_cplx}, we get 
the relationship between the complex polarization $z$ and 
second accumulant moment $\langle \delta x^2 \rangle$
\begin{equation}\label{eq:2nd_cp}
\langle \delta x^2 \rangle \approx 2 \left(\frac{2\pi}{L}\right)^2.
(1 - |z|)
\end{equation}

One could also rewrite Eq.~\eqref{eq:taylor_cp} as the exponential form
\begin{equation}\label{eq:exp_cp}
\tilde{n}(k) \approx e^{- \frac{1}{2} \langle \delta x^2 \rangle k^2},
\end{equation}
where we took $- \frac{1}{2} \langle \delta x^2 \rangle k^2$ in 
Eq.~\eqref{eq:taylor_cp} as the \textit{first order} of the Taylor
expansion of an exponential instead of the second order. From Eq.~\eqref{eq:exp_cp}, one gets
\begin{equation}\label{eq:log_cp}
\langle \delta x^2 \rangle \approx -2 \left(\frac{L}{2\pi}\right)^2 \log|z|.
\end{equation}
For a localized state, both Eq.~\eqref{eq:2nd_cp} and Eq.~\eqref{eq:log_cp}
go to the same finite limit for large $L$; for a delocalized state,
Eq.~\eqref{eq:log_cp} is preferred since it diverges at $|z| = 0$.

Eq.~\eqref{eq:log_cp} gives us a straightforward relationship between 
the complex polarization $z$ and second cumulant moment 
$\langle \delta x^2 \rangle$: when the system is insulating with
 $0 < |z| \leq 1$, $\langle \delta x^2 \rangle$ is finite and the ground
state is localized; when the 
system is metallic with $|z| = 0$, $\langle \delta x^2 \rangle$  diverges
and the ground state is delocalized. Therefore, one could use $z$ as
a direct indicator of the locality of the electrons.

Similarly, the many-body electron localization is defined as 
\begin{equation}
\langle \delta x^2 \rangle \approx -\frac{2}{N} \left(\frac{L}{2\pi}\right)^2 
\log(|Z_N|),
\end{equation}
where $Z_N$ is  the many-body complex polarization.

\subsection{Complex polarization for independent electrons}
When there is no interaction among electrons, the ground state
can be expressed as a Slater determinant $|\Psi_0\rangle$, and
the complex polarization can be written as
\begin{equation}
Z_N^{(\alpha)} = \langle \Psi_0 | \hat{U}(\mbf{\kappa}^{(\alpha)})| \Psi_0\rangle =  \langle \Psi_0 | \Phi_0\rangle,
\end{equation}
where $\hat{U}(\mbf{\kappa}^{(\alpha)}) = e^{i \mbf{\kappa}^{(\alpha)} \cdot \mbf{r}}$ and $| \Phi_0\rangle = \hat{U}(\mbf{\kappa}^{(\alpha)})| \Psi_0\rangle$.

According to the Thouless theorem~\citep{Thouless1960,Rosensteel1981}, 
$|\Phi_0\rangle$ is also a determinant composed of orbitals rotated from 
the orbitals in $\Psi_0$ as
\begin{equation}
\phi_\mu(\mbf{r}) =  e^{i \mbf{\kappa}^{(\alpha)} \cdot \mbf{r}} \psi_{\mu}(\mbf{r}).
\end{equation}
Therefore, $Z_N^{(\alpha)}$ is equal to the overlap between $|\Psi_0\rangle$
and $\Phi_0\rangle$. The overlap of two Slater determinants are evaluated 
by the determinant of the $N\times N$ overlap matrix $\mathcal{S}^{(\alpha)}$
evaluated by
\begin{equation}\label{eq:ovlpmat_cp}
\mathcal{S}_{\mu\nu}^{(\alpha)} = \int \mrm{d} \mbf{r} \psi_\mu^*(\mbf{r})
e^{i\mbf{\kappa}^{(\alpha)}\cdot \mbf{r}} \psi_\nu(\mbf{r}),
\end{equation}
where $\psi_\mu(\mbf{r})$ are occupied orbitals.

The many-body complex polarization is then evaluated as 
\begin{equation}\label{eq:cplx_slater}
Z_N^{(\alpha)} = \left(\text{det} \mathcal{S}^{(\alpha)}_{\uparrow}\right)
\left(\text{det} \mathcal{S}^{(\alpha)}_{\downarrow}\right),
\end{equation}
where the indices $\uparrow$ and $\downarrow$ correspond to up and down
spins. Eq.~\eqref{eq:cplx_slater} can be applied to numerical calculations
where Slater determinants can be obtained to represent the state of the 
system. The above numerical algorithms include the Hartree-Fock method,
the density functional theory (DFT), Slater determinant based quantum Monte
 Carlo (QMC) methods, etc.

\section{Finite temperature complex polarization}\label{sec:ftcplx}
At finite temperature, the expectation value (thermal average) of an 
operator $\hat{A}$ is evaluated under the grand canonical ensemble
\begin{equation}\label{eq:grand_can_cp}
\begin{split}
\langle \hat{A}\rangle(\beta) &= \frac{1}{\mathcal{Q}}\text{Tr}\langle \hat{A}
\hat{\rho}\rangle \\
&= \frac{1}{\mathcal{Q}}\sum_n \langle n | \hat{A}e^{-\beta\hat{H}} |n\rangle,
\end{split}
\end{equation}
where $\beta$ is the inverse temperature, $\hat{H}$ is the Hamiltonian with
the chemical potential,
$\{|n\rangle\}$ forms a set of orthonormal basis,
$\hat{\rho} = e^{-\beta \hat{H}}$ is the density matrix, and 
$\mathcal{Q}$ is the partition function defined as
\begin{equation}
\mathcal{Q} = \sum_n \langle n |e^{-\beta\hat{H}} |n\rangle.
\end{equation}
According to the thermofield theory, the ensemble average in 
Eq.~\eqref{eq:grand_can_cp} can be expressed as an expectation value 
over one state $|\Psi(\beta)\rangle$, known as the \emph{thermofield double
state} or simply \emph{thermal state}
\begin{equation}\label{eq:tfd_av}
\langle \hat{A}\rangle(\beta) = \frac{\langle \Psi(\beta)| \hat{A} |\Psi(\beta)\rangle }{\langle \Psi(\beta)|\Psi(\beta)\rangle}.
\end{equation}
In thermofield theory, a copy of the original Hilbert space $\mathcal{H}$
is introduced as $\tilde{\mathcal{H}}$, known as the auxiliary space. At
infinite temperature ($\beta = 0$), the thermal state is given by
a uniform summation over the orthonormal basis
\begin{equation}
|\Psi(0)\rangle = \sum_n |n\rangle \otimes |\tilde{n}\rangle,
\end{equation}
where $\{|\tilde{n}\rangle\}$ are copies of $\{|n\rangle\}$ in the auxiliary
space.

The thermal state at $\beta$ is then derived by imaginary time evolution 
from $|\Psi(0)\rangle$
\begin{equation}
|\Psi(\beta)\rangle = e^{-\beta \hat{H}/2} |\Psi(0)\rangle.
\end{equation}
Note that the Hamiltonian $\hat{H}$ only acts on the original Hilbert space
$\mathcal{H}$. Eq.~\eqref{eq:tfd_av} can be rewritten as
\begin{equation}
\begin{split}
\langle \hat{A}\rangle(\beta) &= \frac{\langle \Psi(0)| e^{-\beta \hat{H}/2}
\hat{A}e^{-\beta \hat{H}/2}|  \Psi(0)\rangle }{ \langle \Psi(0)|e^{-\beta \hat{H}}|\Psi(0)\rangle}\\
& = \frac{\langle \Psi(0)|e^{-\beta \hat{H}} \hat{A} |  \Psi(0)\rangle }{ \langle \Psi(0)|e^{-\beta \hat{H}}|\Psi(0)\rangle}.
\end{split}
\end{equation}
The complex polarization at temperature $T = 1/\beta$ is thus
\begin{equation}\label{eq:ftcplx_full}
Z_N(\beta) = \frac{\langle \Psi(0)|e^{-\beta \hat{H}} \hat{Z} |  \Psi(0)\rangle }{ \langle \Psi(0)|e^{-\beta \hat{H}}|\Psi(0)\rangle},
\end{equation}
where $\hat{Z} = e^{-i\frac{2\pi}{L}\hat{x}}$. Note that for simplicity,
we dropped the superscript $(\alpha)$ and chose only the $x$ component of 
the three-dimensional position operator $\hat{\mbf{r}}$. This simplification
is valid for a one-dimensional system, and for multi-dimensional systems,
$Z_N$ of other directions can be evaluated in the same manner.

At the mean-field level, thermal state $|\Psi_0\rangle$ can be written as a 
Slater determinant formed by the following $2L \times L$ coefficients
\begin{equation}
C_0 = 
\begin{bmatrix}
1 & 0 & 0 & \cdots & 0 \\
0 & 1 & 0 & \cdots & 0 \\
0 & 0 & 1 & \cdots & 0 \\
0 & 0 & 0 & \ddots & 0 \\
0 & 0 & 0 & \cdots & 1 \\
1 & 0 & 0 & \cdots & 0 \\
0 & 1 & 0 & \cdots & 0 \\
0 & 0 & 1 & \cdots & 0 \\
0 & 0 & 0 & \ddots & 0 \\
0 & 0 & 0 & \cdots & 1 \\
\end{bmatrix}
= \begin{bmatrix}
\mathbb{I} \\ \mathbb{I}
\end{bmatrix},
\end{equation}
where the first $L$ rows correspond to the physical sites, and the last
$L$ rows correspond to the auxiliary sites. A one-body operator
$\hat{w}$ in $\mathcal{H}$ is rewritten as 
\begin{equation}
\bar{\hat{w}} = \hat{w} \oplus 0.
\end{equation}
Under Hartree-Fock approximation, we use the Fock operator $\hat{f}$ 
as the one-body Hamiltonian, and the matrix form of the 
thermal Fock operator $\bar{\hat{f}}$ is
\begin{equation}
\left[\bar{\hat{f}}\right] = \begin{bmatrix}
[\hat{f}] & 0\\
0 & 0
\end{bmatrix}.
\end{equation}
The position operator $\hat{x}$ is also a one-body operator, with
the matrix form as
\begin{equation}
\left[\bar{\hat{x}}\right]  = \begin{bmatrix}
[\hat{x}] & 0\\
0 & 0
\end{bmatrix}.
\end{equation}
The thermal density operator $\bar{\hat{\rho}} = e^{-\beta \bar{\hat{f}}}$,
with the matrix form
\begin{equation}
\left[\bar{\hat{\rho}} \right] = 
\begin{bmatrix}
\left[e^{-\beta \hat{f}}\right] & 0\\
0 &\mathbb{I}\\
\end{bmatrix}.
\end{equation}
The thermofield expression of the complex polarization operator  
$e^{i\frac{2\pi}{L}\hat{x}}$  therefore has the matrix form
\begin{equation}
\left[ \bar{\hat{Z}}\right] = 
\begin{bmatrix}
\left[ e^{i\frac{2\pi}{L}\hat{x}}\right] & 0\\
0 & \mathbb{I}
\end{bmatrix}.
\end{equation}
The complex polarization at $\beta$ can be evaluated by a similar 
formulation as ground state
\begin{equation}\label{eq:ftcp_cp}
Z_N(\beta) = \frac{\text{det}\left( C_0^{\dag} \left[ \bar{\hat{Z}}\right] 
\left[\bar{\hat{\rho}} \right] C_0\right)}
{\text{det}\left( C_0^{\dag} \left[\bar{\hat{\rho}} \right] C_0\right)}.
\end{equation}

At infinite temperature, $\beta = 0$, and $\left[e^{-\beta \hat{f}}\right] 
= \mathbb{I}$, leading to $\left[\bar{\hat{\rho}} \right] = 
\mathbb{I}\otimes \mathbb{I}$. One could rotate the basis to the eigenstates of $\hat{x}$,
and thus $\hat{Z}$ is diagonal in this basis. Note that $C_0$ and 
$\left[e^{-\beta \hat{f}}\right]$ do not change under the rotation.
After rotation, $\left[\bar{\hat{Z}}\right]$ becomes a diagonal matrix having
the form
\begin{equation}
\left[ \bar{\hat{Z}}\right] = 
\begin{bmatrix}
z_1 & & & & \\
 & z_2 & &  &  \\
 &  & \ddots &  & \\
 & & & z_L&  \\
 & & & & [\mathbb{I}]
\end{bmatrix},
\end{equation}
where $z_{\mu} = e^{-i2\pi x_{\mu} /L}, \mu = 1,...,L$.
The denominator in Eq.~\eqref{eq:ftcp_cp} is then
\begin{equation}
\text{det}\left( C_0^{\dag}  C_0\right) = 2^L.
\end{equation}
The numerator in Eq.~\eqref{eq:ftcp_cp} is 
\begin{equation}\label{eq:infT_cp}
\text{det}\left( C_0^{\dag} \left[ \bar{\hat{Z}}\right] \left[\bar{\hat{\rho}} \right]  C_0\right) 
= \text{det}\left([\hat{Z}] + [\mathbb{I}]\right) = \prod_{\mu} \left(z_\mu + 1\right).
\end{equation}
Suppose the basis is chosen to be the site basis, i.e., $x_\mu = \mu$.
When $L$ is even, $z_{L/2} = -1$ is included in the product in 
Eq.~\eqref{eq:infT_cp} and the numerator is zero, leading to $Z_N = 0$.
When $L$ is odd, $z_{(L+1)/2}$ and $z_{(L-1)/2}$ differ from $-1$
with infinitesimal displacement when $L$ is large enough and the numerator
$\ll 2^L$, leading to $Z_N \rightarrow 0$. Therefore, at thermal dynamic limit,
$Z_N = 0$ at infinite temperature ($\beta = 0$). This observation is 
consistent with the common sense that the electron can move freely 
at infinite temperature and the second cumulant moment diverges. 

\section{Tight binding model}\label{sec:cptb}
The generalized form of a non-interacting Hamiltonian can be written
as 
\begin{equation}
\hat{h} = -\sum_{\mu \neq \nu}\left(t_{\mu\nu}\ha^{\dag}_\mu\ha_\nu + \text{h.c.}\right)
+\sum_\mu u_\mu \ha^\dag_\mu \ha_\mu,
\end{equation}
where $\ha^\dag_\mu\ha_\nu$ describes electron hopping from site $j$ to site $i$.
The one-band tight binding model Hamiltonian takes the form
\begin{equation}
\hat{h}_{\text{tb}} = -t\sum_{\langle \mu, \nu\rangle, \sigma} \ha^{\dag}_{\mu,\sigma}\ha_{\nu,\sigma} 
+ \text{h.c.},
\end{equation}
where $\langle i, j\rangle$ indicates nearest-neighbor hopping and $\sigma$
stands for spin freedom. In the following, we focus on the one-dimensional
tight binding model with periodic boundary condition (PBC) and 
SU(2) symmetry. The Hamiltonian becomes
\begin{equation}\label{eq:1dtb_cp}
\hat{h}_{\text{tb}} = -t\sum_{\mu} \ha^{\dag}_{\mu} \ha_{\mu+1} + \text{h.c.}
\end{equation}
The eigenstates of Eq.~\eqref{eq:1dtb_cp} can be analytically solved
with the help of Fourier transformation from real space to $k$ space 
(momentum space)
\begin{equation}
\begin{split}
\ha_\mu &= \frac{1}{\sqrt{2\pi}}\sum_{k\in\text{BZ}}  e^{ik\mu}\hc_k,\\
\hc_k   & = \frac{1}{\sqrt{2\pi}} \sum_{\mu} e^{-ik\mu}\ha_\mu.
\end{split}
\end{equation}
It is easy to prove that $\{\hc_k, \hc^\dag_{k'}\} = \delta_{kk'}$, so 
$\hc^\dag_{k}$ and $\hc_k$ are creation and annihilation operators in
$k$ space. Eq.~\eqref{eq:1dtb_cp} can be rewritten as
\begin{equation}
\begin{split}
\hat{h}_{\text{tb}} &= -\frac{t}{2\pi} \sum_{\mu} \sum_{k,k'}
e^{-ik\mu} e^{ik'(\mu+1)} \hc^{\dag}_k \hc_{k'} + \text{h.c.} \\
    &= -\frac{t}{2\pi} \sum_{k,k'} \sum_{\mu} \left(e^{-ik\mu} e^{ik'(\mu+1)}\right) \hc^{\dag}_k \hc_{k'} + \text{h.c.}\\
    &= -t  \sum_{k,k'} \delta_{kk'} e^{ik}\hc^{\dag}_k \hc_{k'} + \text{h.c.}\\
    &= -2t \sum_k \cos k \hc^{\dag}_k\hc_k. 
\end{split}
\end{equation}
Therefore, $\hat{h}_{\text{tb}}$ is diagonal in the basis created by $\hc^{\dag}_k$. For the crystalline case, $\hc^{\dag}_k$ creates an electron in
a Bloch wave
\begin{equation}
\psi_{k}(\mu) = e^{ik\mu} u_{k}(\mu),
\end{equation}
where $\mu = 0, 1, ..., L-1$ stands for the site basis and $k$ represents momentum numbers. 
$u_{k}(\mu)$ is identical on each site: $u_{k}(\mu + 1) = u_{k}(\mu)$,
and we will use a constant $1/\sqrt{L}$ to replace $u_{k}(\mu)$ to ensure
that $\psi_{k}(\mu)$ is normalized.

In a one-dimensional chain with $L$ sites, there are $L$ allowed 
$k$ values:
\begin{equation}
k_s = \frac{2\pi s}{L}, s = 0, 1, ..., L-1.
\end{equation}
We evaluate the overlap matrix in Eq.~\eqref{eq:ovlpmat_cp} under the 
basis $\{\psi_{k_s}(\mu)\}$,
\begin{equation}
\begin{split}
\mathcal{S}_{k_s, k_{s'}} =&  \sum_{\mu} \psi^*_{k_s}(\mu)
e^{\frac{i2\pi \mu}{L}} \psi_{k_{s'}}(\mu)\\
    &= \frac{1}{L} \sum_{\mu} e^{-i(s-s'-1)\mu}\\
    &= \delta_{s, s'+1}.
\end{split}
\end{equation}
Therefore, $\mathcal{S}_{k_s, k_{s'}}$ is nonzero only when $s = s'+1$.
When the lattice is fully occupied, both $s$ and $s'$ run over all the 
$L$ values. This means that for any $s$, there exists an occupied orbital
$\psi_{k_{s-1}}$. Therefore, any row or any column of the $\mathcal{S}$ 
matrix has one and only one nonzero element (equal to $1$). The determinant
of $\mathcal{S}$ is thus nonzero, and $Z_N = 1$, indicating an insulating 
state. 

When the lattice is not fully occupied, the overlap matrix $\mathcal{S}$ 
only consists of occupied orbitals, and if one can find an $s$ where 
$\psi_{k_{s-1}}$ is unoccupied, then the row corresponding to $\psi_{k_{s}}$
contains only zero elements, leading to $Z_N = 
\text{det}\left(\mathcal{S}\right) = 0$. For the half-filling case, whether
$Z_N = 0$ or not depends on the value of $L$. When $L$ is even, there 
are two cases: $L = 4m$ and $L = 4m+2$, where $m$  is an integer. The 
spectrum of the two cases are shown in Fig.~\ref{fig:dispersion_tb} with
$L=8$ and $L=10$. The cosine line plot reflects the dispersion relation 
between $\varepsilon_k$ and $k$: $\varepsilon_k = -t*\cos k$, and the 
circles on top of the line correspond to allowed $k$ values: $2\pi n/L,
n = 0, ..., L-1$. Fig.~\ref{fig:dispersion_tb} (a) shows the half-filling
case of $L=8$, with $4$ electrons in the lattice. The solid black
dots are occupied orbitals, while the two circles with stripes are two 
degenerate states with the total occupation number equal to $1$. 
If we consider the two striped circles as one occupied site, then 
for any $s$th occupied dot, the $(s+1)$th orbital is also 
occupied or partially occupied. Therefore, when $L = 4m$, $|Z_N| > 0$.
Fig.~\ref{fig:dispersion_tb} (b) tells a different story. With $L=10$,
there are five occupied orbitals shown as solid black dots, and there
are no partially occupied orbitals in this case. Therefore, when $L=4m+2$,
$|Z_N| = 0$.
\begin{figure}[t!]
\centering
\begin{subfigure}[t]{0.85\textwidth}
\includegraphics[width=\textwidth]{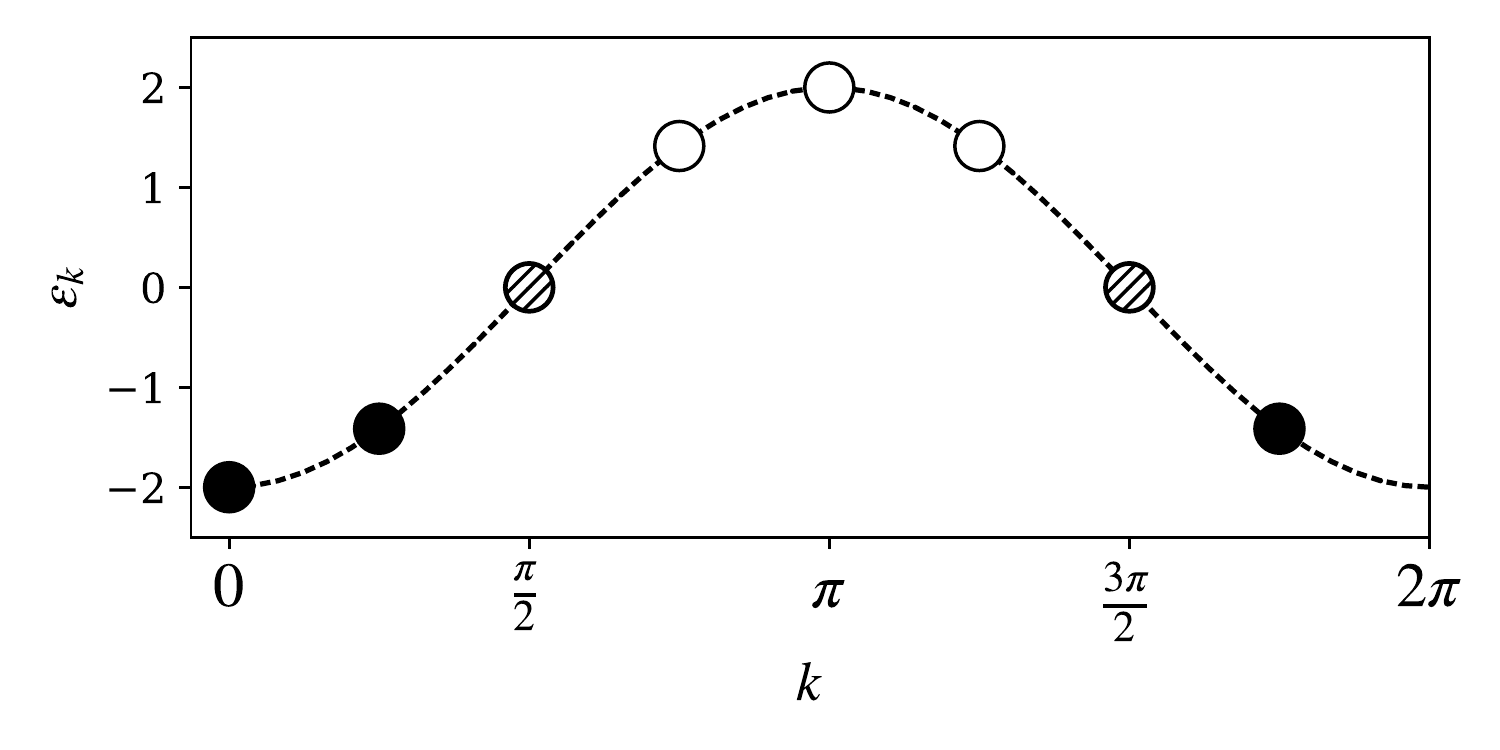}
\caption{$L = 8$}
\end{subfigure}
\begin{subfigure}[t]{0.85\textwidth}
\includegraphics[width=\textwidth]{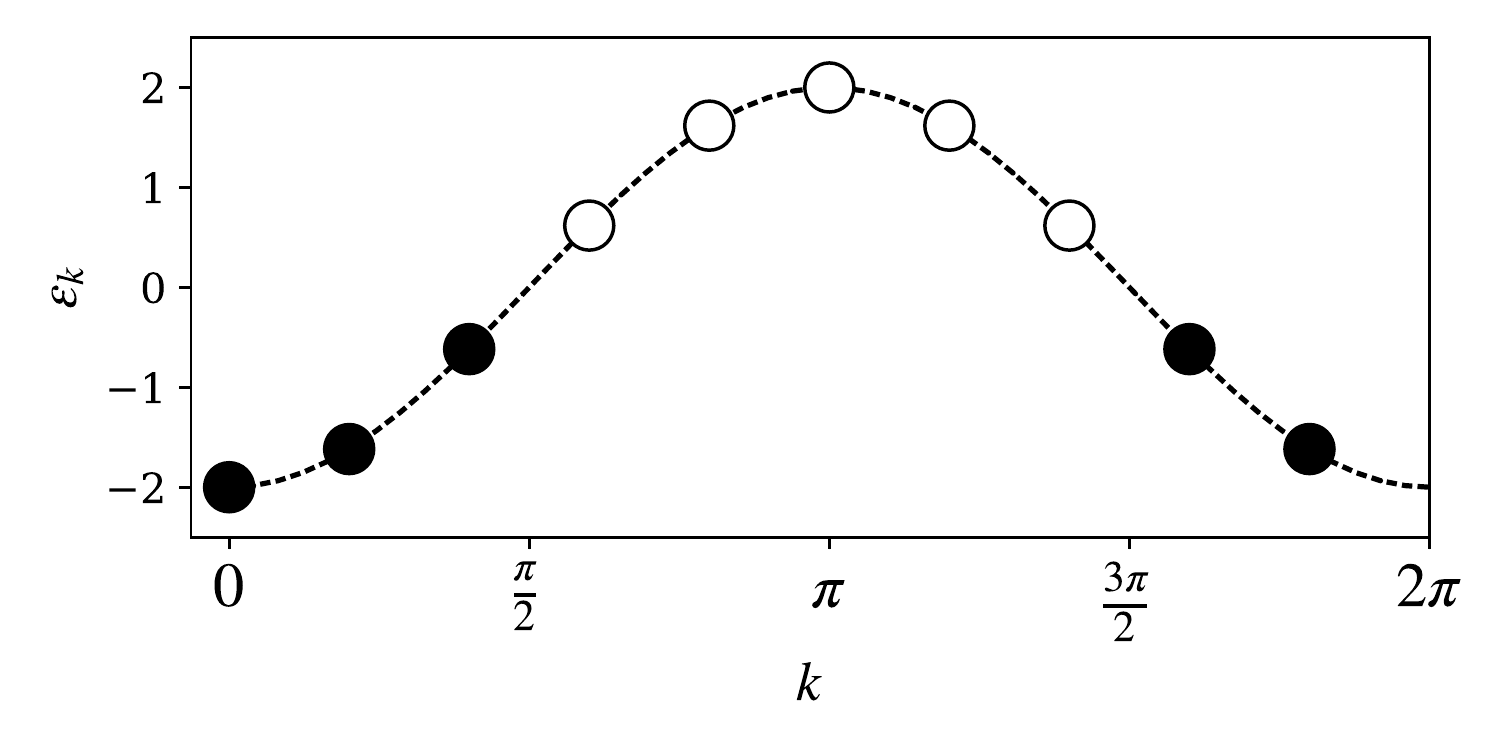}
\caption{$L = 10$}
\end{subfigure}
\caption{Dispersion relation and energy levels of the half-filled tight binding model for (a) $L = 8$ and (b) $L=10$. Solid black dots are occupied orbitals,
blank circles are unoccupied orbitals, and circles with stripes are partially 
occupied orbitals due to degeneracy.}\label{fig:dispersion_tb}
\end{figure}

At finite temperature, we again evaluate the thermal average with
thermal states.
The matrix form of the phase operator $[\hat{Z}]$ based on the above discussion is:
\begin{equation}
\left[\hat{Z}\right] = \begin{bmatrix}
0 & 0 & 0  & \cdots & 0 & 1 \\
1 & 0 & 0  & \cdots & 0 & 0 \\
0 & 1 & 0  & \cdots & 0 & 0 \\
 &  &   & \ddots &  &  \\
0 & 0 & 0  & \cdots & 1 & 0 \\
\end{bmatrix},
\end{equation}
and the thermal operator of complex polarization has the form
\begin{equation}
\left[\bar{\hat{Z}}\right] = \begin{bmatrix}
[\hat{Z}] & 0 \\
0 & \mathcal{I}
\end{bmatrix}.
\end{equation}
Since the Hamiltonian is diagonal with the basis $\{\psi_{k_{s}}\}$, 
the thermal density matrix has the form
\begin{equation}
\left[\bar{\hat{\rho}}\right] = \begin{bmatrix}
\xi_1 & & & & \\
 & \xi_2 & & & \\
& & \ddots & & \\
& & & \xi_L &  \\
& & & & [\mathcal{I}]
\end{bmatrix}.
\end{equation}
The complex polarization can be evaluated according to 
Eq.~\eqref{eq:ftcp_cp}, 
\begin{equation}\label{eq:tb_ftcp}
Z_N(\beta) =  \frac{1-(-1)^L\prod_{\mu}\xi_\mu}{\prod_{\mu}(1+\xi_\mu)}
\end{equation}.
Note that $\xi_\mu =e^{-\beta\varepsilon_\mu} > 0$, so the denominator
of Eq.~\eqref{eq:tb_ftcp} is always greater than zero.

Now let us examine two extreme cases: $\beta \rightarrow \infty$ (zero
temperature) and $\beta\rightarrow 0$ (infinite temperature). At 
$\beta \rightarrow \infty$, if the Fermi level is above all bands,
then $\xi_\mu \gg 1$ for all $\mu$, and Eq.~\eqref{eq:tb_ftcp}
is well approximated by
\begin{equation}
|Z_N| \approx \frac{\prod_{\mu}\xi_\mu}{\prod_{\mu}\xi_\mu} = 1.
\end{equation}
Therefore the lattice is an insulator. However, when the Fermi level
is below some bands (unoccupied orbitals), then the $\xi$ values of these
bands $\rightarrow 0$, and the numerator of Eq.~\eqref{eq:tb_ftcp}
$\rightarrow 1$, while the denominator $\prod_{\mu}(1+\xi_\mu)\rightarrow \infty$, so $Z_N\rightarrow 0$, giving a conducting solution. The above low 
temperature limit agrees with the previous analysis of ground state 
metal-insulator transition of the tight binding model.

At $\beta\rightarrow\infty$, all $\xi_\mu\rightarrow 1$, resulting in an
numerator $0$ ($L$ is even) or $2$ ($L$ is odd), while the denominator
is $2^L$. Therefore, $Z_N\rightarrow 0$ as $L\rightarrow \infty$, and
the electrons in the tight binding model are delocalized.

\begin{figure}
\centering
\justify
\includegraphics[width=1\textwidth]{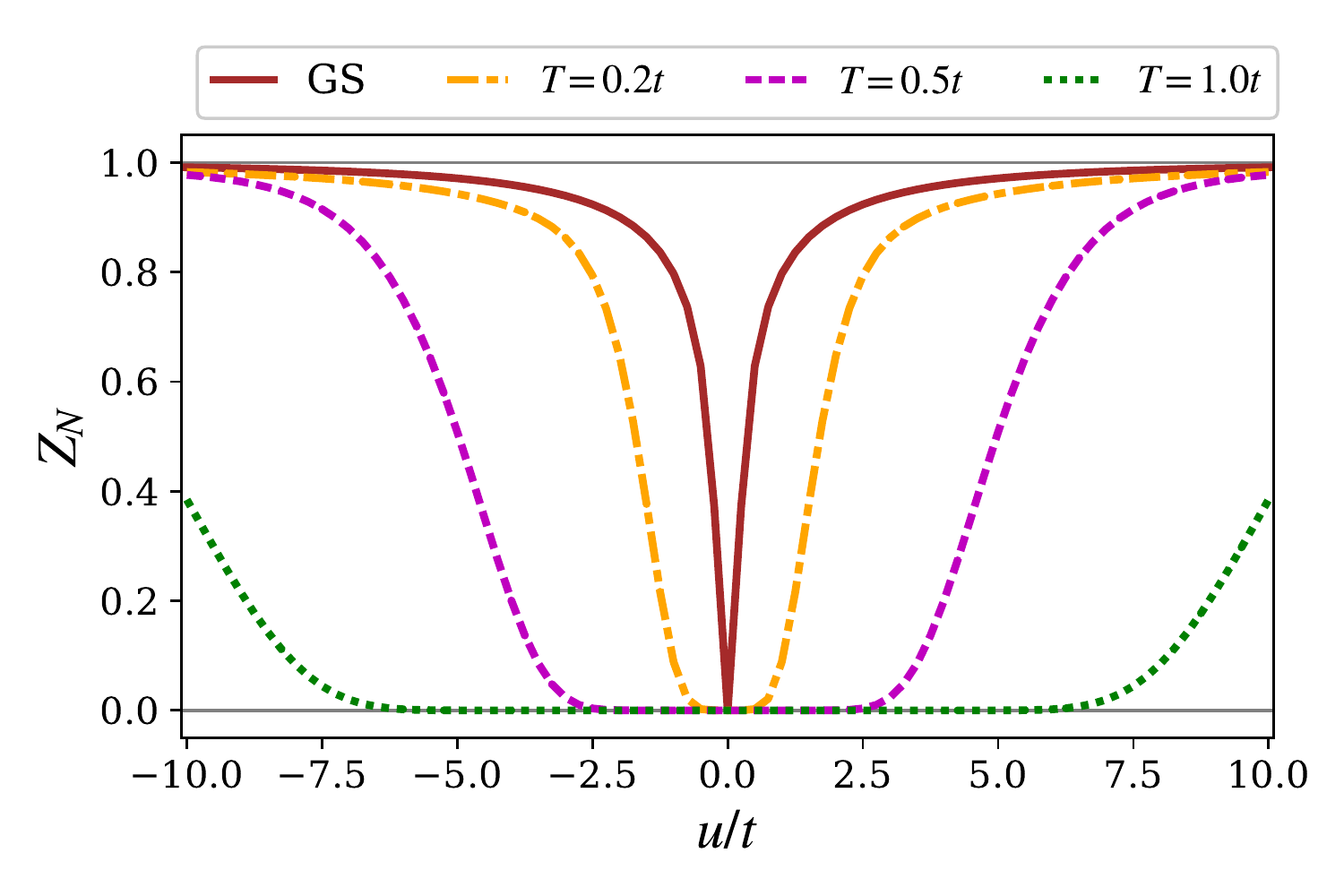}
\caption{Complex polarization of the tight binding model ($L=42$) with staggered potential $u$ at ground state (GS), $T = 0.2$, $T=0.5$, and $T=1.0$, respectively.
 }\label{fig:zn_tb}
\index{figures}
\end{figure}

Next we add the staggered potential $u$ onto the original tight binding
model:
\begin{equation}
\hat{h} = -t\sum_{\mu} \ha^{\dag}_{\mu} \ha_{\mu+1} + \text{h.c.}
+ u\sum_{\mu \in\text{odd}} \ha^{\dag}_{\mu}\ha_{\mu},
\end{equation}
where $u > 0$ is only applied to the odd sites. For simplicity, we assume 
that $L$ is even. The effect of $u$ is to provide a potential wall/well
for every other site, and this effect prohibits the free flow of electrons.
For the rest of the tight binding calculations, we choose the chain length 
$L = 42$ and Boltzmann constant $k_B = 1$, and use $t$ as the energy unit.
In Fig.~\ref{fig:zn_tb} we show the complex polarization $Z_N$ of the 
half-filled tight binding model against
the staggered potential $u$ at ground state, $T=0.2t$, $T=0.5t$ and $T=1.0t$.
As predicted above, the half-filled ground state of the original tight 
binding model ($u=0$) with $L=4m+2$ is metallic with $Z_N=0$. 
As the staggered potential 
turned on, ground state $Z_N$ grows rapidly and the system becomes more and
more insulating. With raising the temperature, the metallic regime expands
within the small $|u|$ region, and the growth curve of $Z_N$ with respect to $|u|$
becomes more flat. The temperature effect smears the sharp transition 
at ground state. Note that the curves are symmetric to $u = 0$ since only
the potential differences between adjacent sites affect the state of the
system.

\begin{figure}
\centering
\justify
\includegraphics[width=1\textwidth]{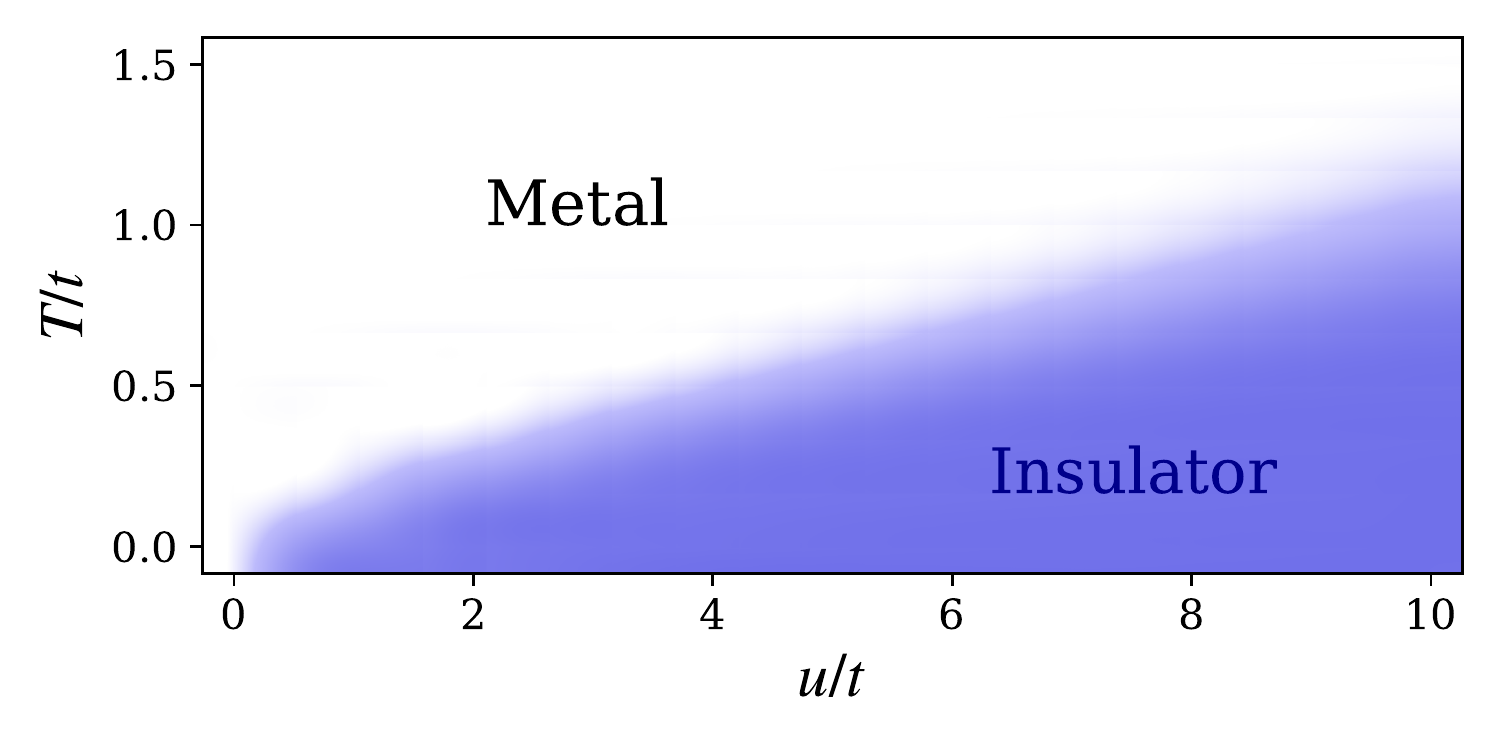}
\caption{Phase diagram of the tight binding model ($L=42$) with the staggered potential $u$. The blue area corresponds to $Z_N > 0$ (insulator) and the white area
corresponds to $Z_N = 0$ (metal). The 2D plot is smoothed by Bessel
interpolation. Grid: $20$ points in the $x$-axis and $10$ points in the 
$y$-axis.
 }\label{fig:one_band_tb_phase}
\index{figures}
\end{figure}

We show the phase diagram of $Z_N$ for the tight binding model with respect
to the staggered potential $u$ and temperature $T$ in Fig.~\ref{fig:one_band_tb_phase}. We observed a sharp barrier
between the metallic phase and insulating phase at $u\rightarrow 0_+$ and low 
temperature, and then
the barrier becomes rather vague at larger $u$ with a higher transition 
temperature. This observation is consistent with the flatter curves at a 
higher temperature in Fig.~\ref{fig:zn_tb}. We further observe a linear
growth of transition temperature $T_c$ with respect to $u$ at larger $u$.
Since the transition temperature $T_c$ is directly
related to the gap of the system, we also plotted the gap against $u$ in 
Fig.~\ref{fig:tb_gap}. The linear dependence of $\Delta_{\text{gap}}$
to the staggered potential $u$ at large $u$ region is consistant to the
 linear $T_c - u$ relationship in Fig.~\ref{fig:one_band_tb_phase}. 
At $u$ smaller than $0.1t$, we observe a rather slow growth of 
$\Delta_{\text{gap}}$ with $u$, which agrees with the metallic phase at 
$u\approx 0$ and then a sudden appearance of the insulator phase with a
nearly verticle wall.

\begin{figure}
\centering
\justify
\includegraphics[width=1\textwidth]{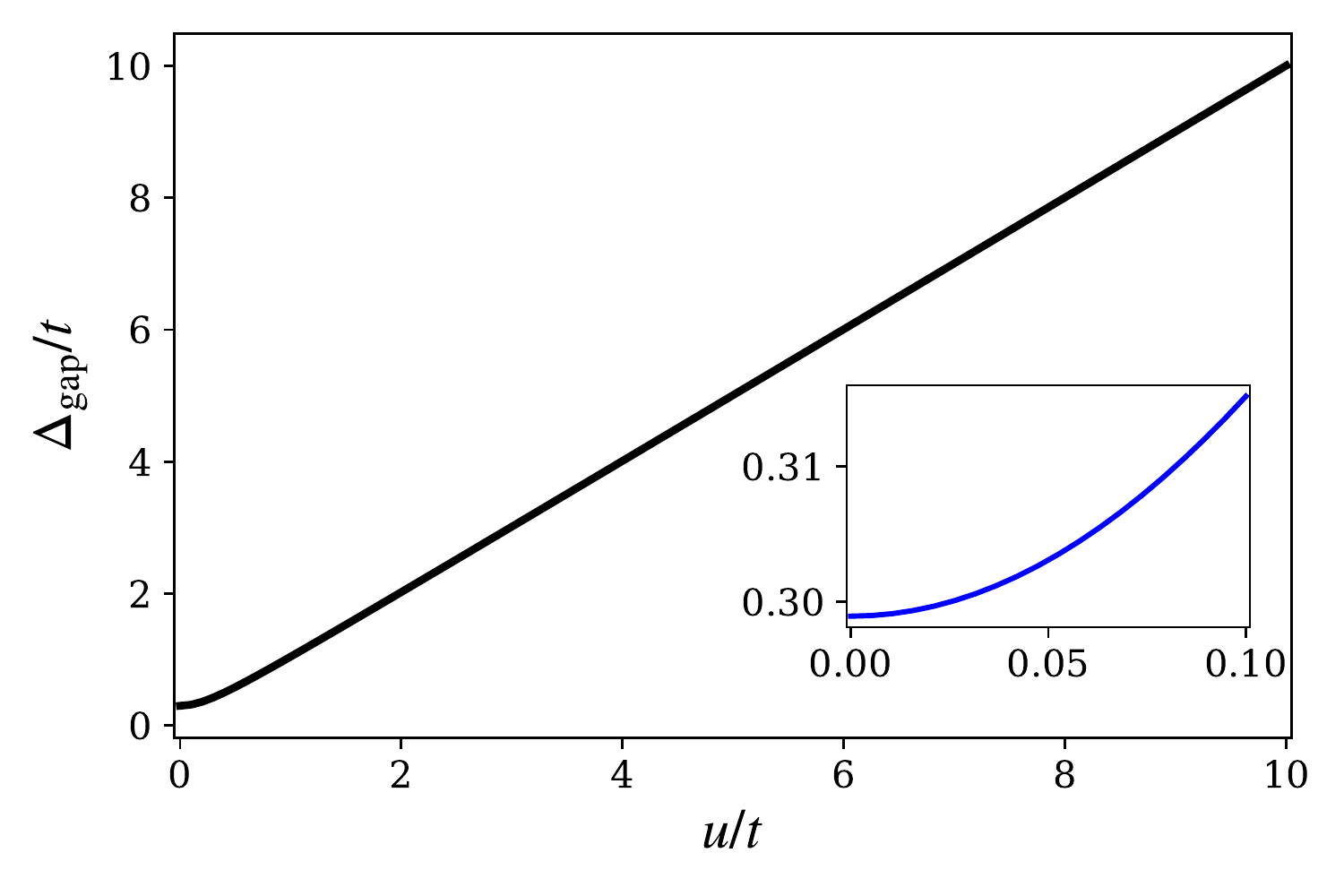}
\caption{Energy gap of the half-filled staggered tight binding model
against the staggered potential $u$. Inset: energy gap for $u \in [0t,0.1t]$.}\label{fig:tb_gap}
\index{figures}
\end{figure}

\section{Hydrogen chain}\label{sec:cphydrogen}
A linear chain of hydrogen atoms equispaced~\citep{Hachmann2006,AlSaidi2007,
Sinitskiy2010,Stella2011,Motta2017,Motta2020} is the simplest \textit{ab initio} 
periodic system that one can find. Unlike the simplicity of the structure,
the phase diagram of the hydrogen chain involves complex components: 
metal-insulator transition (MIT), paramagnetic-antiferromagnetic (PM-AFM) transition 
and dimerization~\citep{Hubbard1963}. Hydrogen chain has a similar 
structure as the one-dimensional
Hubbard model which has been studied for decades. Compared to the Hubbard 
model where electron-electron
interactions are of short range, the Coulomb interaction in the hydrogen
chain is long-ranged. Moreover, calculations beyond the minimal basis set
(STO-6G) will introduce a multi-band effect into the hydrogen chain, which 
is absent in the one-band model systems.

In the following, we compute the complex polarization at both ground state
and finite temperature for the hydrogen chain system with atoms equally
spaced along the $z$-direction. The H-H bond length $R$ is introduced as 
the parameter and adjusted to show different phases. The Hamiltonian of
this problem is
\begin{equation}
\hat{H} = -\frac{1}{2}\sum_{\mu = 1}^N\nabla^2_{\mu} + 
\sum_{\mu < \nu}^N \frac{1}{|\mbf{r}_{\mu} - \mbf{r}_{\nu}|}
- \sum_{\mu, a}^{N} \frac{1}{|\mbf{r}_{\mu} - \mbf{R}_{a}|}
+ \sum_{a < b}^N  \frac{1}{|\mbf{R}_{a} - \mbf{R}_{b}|}
\end{equation}
where $(\mbf{r}_1, ..., \mbf{r}_N)$ are the electron positions in the 
Cartesian coordinates, $\mbf{R}_{a} = aR\hat{\mbf{e}}_z$ is the 
position of the $a$th atom on $z$-axis. In this work, energies 
and the temperature ($k_B=1$) are measured 
in Hartree ($me^4/\hbar^2$) and lengths in Bohr radius $a_B = \hbar^2/(me^2)$.
In one supercell, $30$ hydrogen atoms are included and only the $\Gamma$ point
in the reciprocal space is taken into account. The basis set is 6-31G,
where the 1$s$ orbital and the 2$s$ orbital are included. We evaluate the 
complex polarization $Z_N$, staggered magnetic moment $m$, electron
population on 2s orbital, and the HOMO-LUMO gap of 
above hydrogen chain system
at ground state, $T = 0.01, 0.02, 0.03$ and $0.04$ Hartree.
We present the results from unrestricted Hartree-Fock (UHF)
and DFT (GGA/PBE and B3LYP) calculations in Fig.~\ref{fig:hchain_uhf_cp},
Fig.~\ref{fig:hchain_pbe_cp}, and Fig.~\ref{fig:hchain_b3lyp_cp}.
All calculations are performed within the framework of the quantum chemistry
package \texttt{PySCF}~\citep{PYSCF2017,PYSCF2020}.

\begin{figure}[th!]
\centering
\justify
\includegraphics[width=1\textwidth]{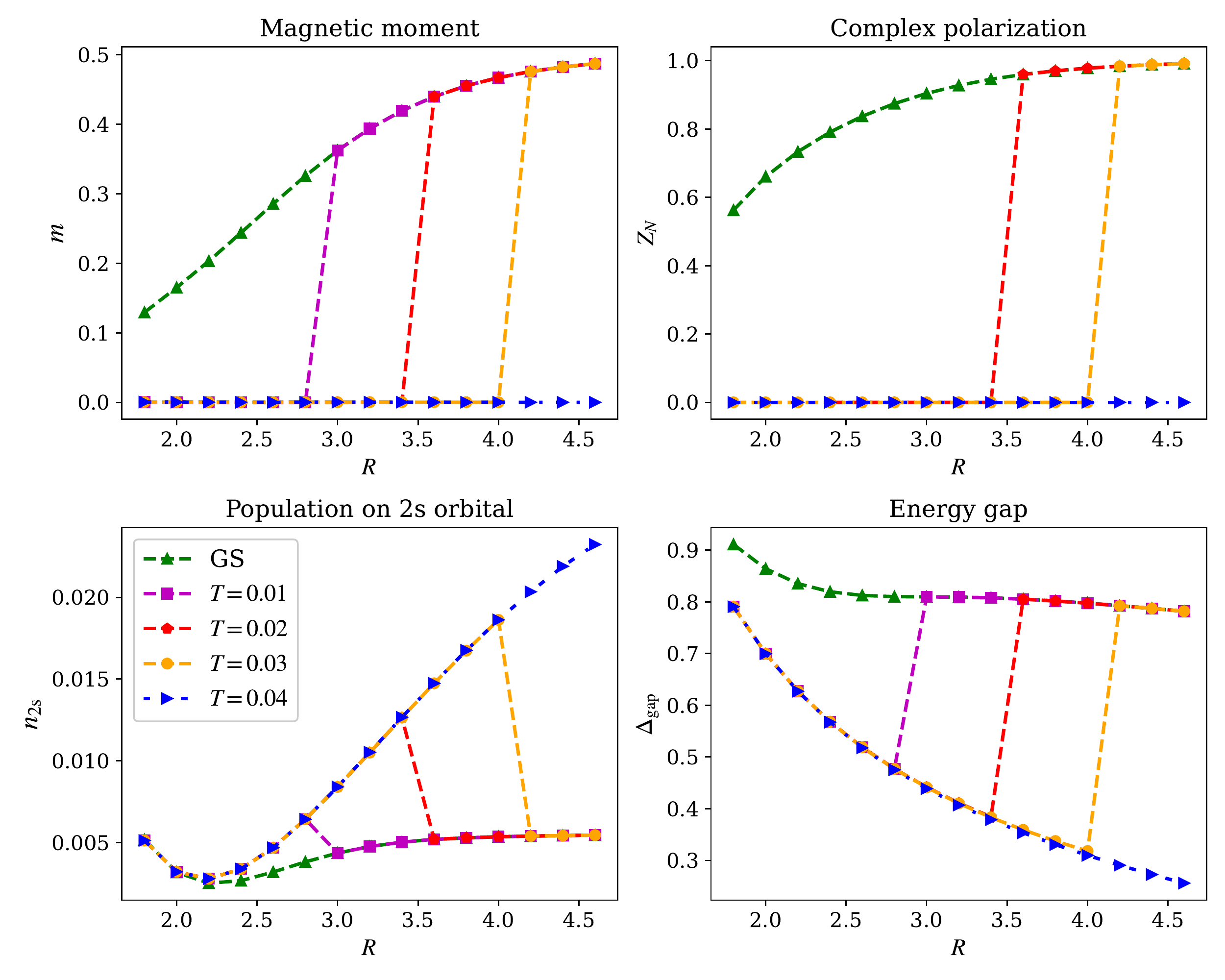}
\caption{Complex polarization, magnetic moment, population on 2s orbital
and energy gap of hydrogen chain with unrestricted Hartree-Fock method.
Note that the complex polarization at $T=0.01$ is not presented here due to 
overflow.
}
\label{fig:hchain_uhf_cp}
\index{figures}
\end{figure}

\begin{figure}[th!]
\centering
\justify
\includegraphics[width=1\textwidth]{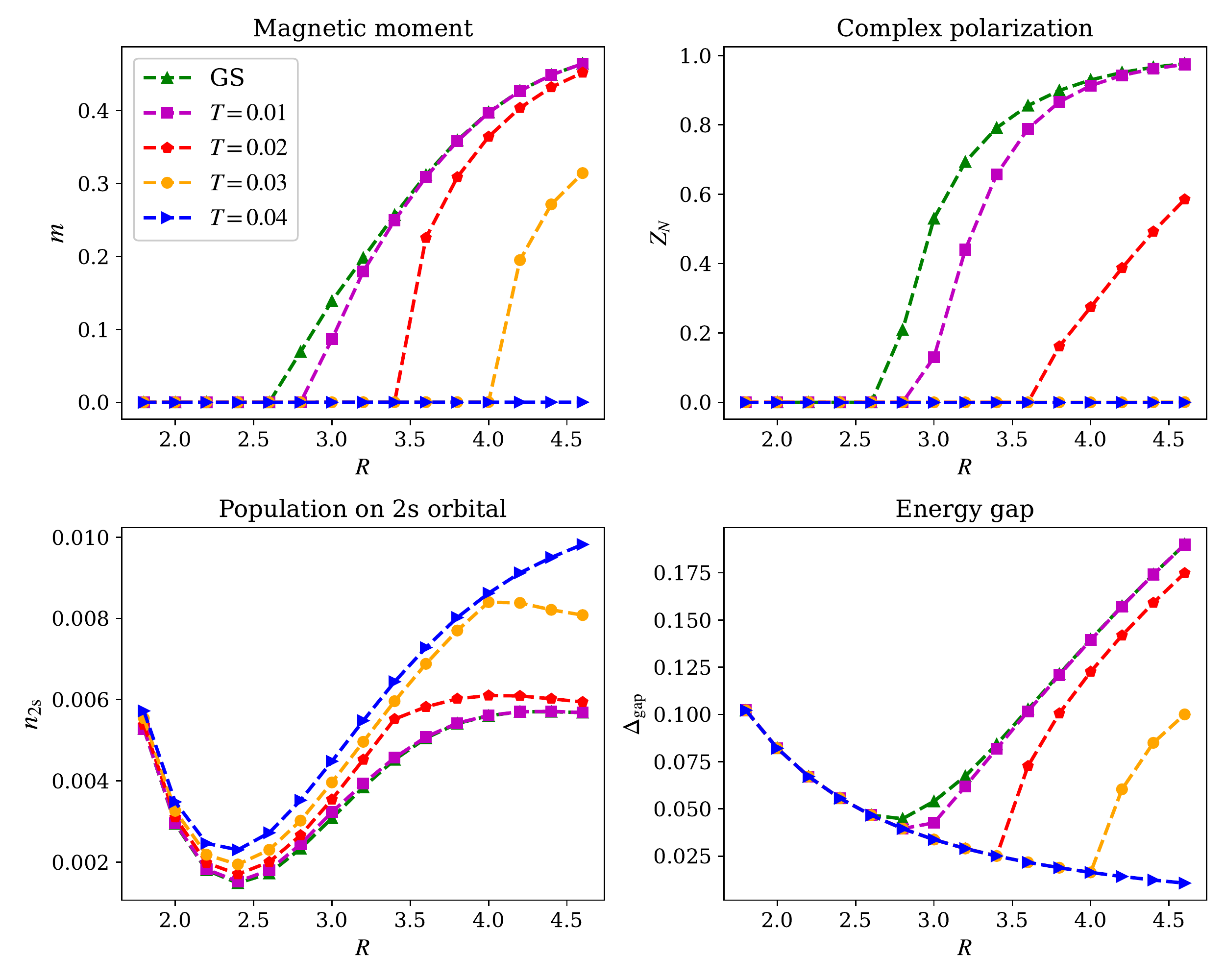}
\caption{Complex polarization, magnetic moment, population on 2s orbital,
and energy gap of hydrogen chain from DFT with PBE functional.}
\label{fig:hchain_pbe_cp}
\end{figure}

\begin{figure}[th!]
\centering
\justify
\includegraphics[width=1\textwidth]{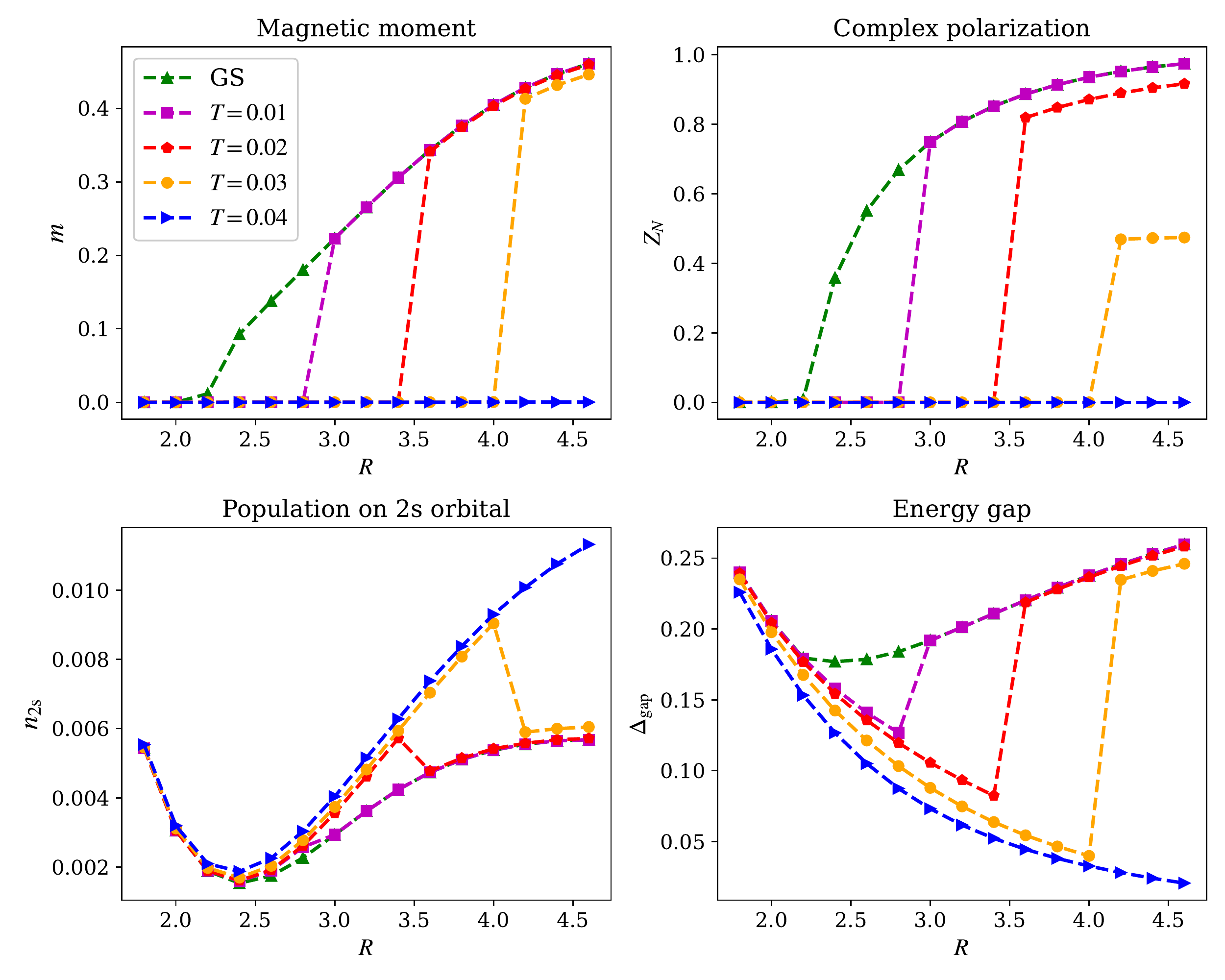}
\caption{Complex polarization, magnetic moment, population on 2s orbital,
and energy gap of hydrogen chain from DFT with B3LYP functional.}
\label{fig:hchain_b3lyp_cp}
\index{figures}
\end{figure}

\begin{table}[hbt!]
\centering
\begin{tabular}{l|ccc}
\hline
$T$/Hartree & Hartree-Fock & PBE  & B3LYP\\
\hline
Ground state & $\sim$1.0 & 2.6 & 2.2\\
 0.01& 2.8 & 2.8 &2.8 \\
0.02 & 3.4 & 3.4 & 3.4\\
0.03 & 4.0 & 4.0 & 4.0\\
\hline
\end{tabular}
\caption{PM-AFM transition bond length $R$ (in Bohr) at ground state and low temperature.}\label{tab:mag_trans}
\index{tables}
\end{table}

\begin{table}[hbt!]
\centering
\begin{tabular}{l|ccc}
\hline
$T$/Hartree & Hartree-Fock & PBE  & B3LYP\\
\hline
Ground state & $\sim$1.0 & 2.6 & 2.2\\
 0.01& -  & 2.8 &2.8 \\
0.02 & 3.4 & 3.6 & 3.4\\
0.03 & 4.0 & - & 4.0\\
\hline
\end{tabular}
\caption{Metal-insulator transition bond length $R$ (in Bohr) at ground state and low temperature.}\label{tab:cp_trans}
\index{tables}
\end{table}

All of the three methods predicted metal-insulator transition and 
PM-AFM transition at ground state and low temperature. The transition
$R$ predicted by the above methods are summarized in Table~\ref{tab:mag_trans}
and Table~\ref{tab:cp_trans}. The two
transitions happened nearly simultaneously, which provided evidence
for the hypothesis that the insulator at large $R$ regime is an 
antiferromagnetic (AFM) insulator. With a metal to insulator transition 
happening
with raising $R$, the population of the 2s orbital experienced a sudden 
drop, which indicates that the origin of the metal phase at small $R$ regime
is caused by the crossover between 1$s$ and 2$s$ bands. Although the three
methods all predicted the transitions, the behaviors of the order parameters
against $R$ are quite different between UHF and PBE calculations. UHF predicted
a much smaller transition $R$ at ground state ($\sim 1 Bohr$), while PBE
predicted $R_c$ to be $\sim 2.6 Bohr$. The finite temperature predictions
of $R_c$ for PM-AFM transitions from the two methods are similar, while the
transition behaviors are quite different: UHF described that the finite 
temperature curves experienced  a sudden jump from 
zero to the ground state curve; PBE predicted that the finite temperature
curves grew from zero at $R_c$ and reached to a peak which decreases with
temperature. Moreover, at $T=0.03$ and $R > 4.0$, PBE predicted an 
AFM metal ($m > 0$ and $Z_N = 0$). This observation confirmed that the 
metal is mainly caused by the crossover of 1$s$ and 2$s$ bands, and the
existence of the AFM order does not necessarily guarantee an insulating 
phase. However, the AFM metal phase is not observed with the other two methods.
The B3LYP results in Fig.~\ref{fig:hchain_b3lyp_cp} are closer to those
from UHF, except that the peaks of $m$ and $Z_N$ drop as the temperature
increases.

\section{Conclusion}\label{sec:cpconc}
In this chapter, we presented the finite temperature formulation of complex polarization
under the scheme of thermal field theory. The complex polarization
has a direct relationship with the electron localization at ground state: 
when the complex polarization is zero, the electrons are delocalized,
and thus the system is metallic; when the complex polarization is 
nonzero, the electrons are localized and thus the system is insulating.
At finite temperature, the complex polarization can also be used as
the indicator of metal-insulator transition. We applied the thermofield
implementation of the complex polarization to the tight binding model 
 with staggered potential $u$, where as $u$ increases, the electrons
tend to sit on the site with lower potential, and thus are localized. 
We observed the increase of the complex polarization with $u$ for both
ground state and finite temperature. Moreover, we found that the
transition temperature predicted by the complex polarization is linearly
 dependent on the staggered 
potential $u$ at intermediate to large $u$ regime, 
which is consistent with the linear dependence of the energy gap on $u$
at this intermediate to large $u$ regime. Therefore, the energy gap,
electron localization, and complex polarization provide the same
predictions of the metal-insulator transition behaviors. We further 
studied the metal-insulator and paramagnetic-antiferromagnetic (PM-AFM) transition
for the hydrogen chain system by computing the complex polarization
and magnetic moment against temperature and H-H bond length. 
Along with the results of the population of 2s orbitals and the energy gap,
we confirmed that the origin of the metallic phase is the crossover
of $1s$ and $2s$ (or higher, e.g, $2p$) bands. The antiferromagnetic (AFM)
phase is usually accompanied by the insulating phase, but at finite 
temperature, we saw that PBE predicted an AFM metallic phase, which 
indicates that the disappearance of the insulating phase is not necessarily
due to the loss of AFM phase. With complex polarization proving to be
a good indicator of metal-insulator transition at both ground 
state and finite temperature, further applications are anticipated to
bring more insights into this intriguing phenomena. 


\chapter{Quantum imaginary time evolution and quantum thermal simulation\label{chp:qite}}

\newcommand{\oper}[1]{\hat{#1}}
\newcommand{\ts}{{\Delta \tau}}
\newcommand{\idx}{m}
\newcommand{\choi}[2]{ \ket{#1} \bra{#2} }
\newcommand{\adj}[1]{ {#1}^\dagger }
\newcommand{\vett}[1]{ {\bf{#1} }}
\newtheorem{theorem}{Theorem}

\section{Abstract}
  An efficient way to compute Hamiltonian ground-states on a quantum computer stands to impact
  many problems in the physical and computer sciences, ranging from quantum simulation to machine learning. Unfortunately, existing techniques, such
  as phase estimation and variational algorithms, display potential disadvantages, such as
  requirements for deep circuits with ancillae and high-dimensional optimization. We describe the quantum imaginary time evolution and quantum Lanczos algorithms,
  analogs of classical algorithms for ground (and excited) states, but with exponentially reduced space and time requirements per iteration,
  and avoiding deep circuits with ancillae and high-dimensional optimization. We further discuss
  quantum imaginary time evolution as a natural subroutine to generate Gibbs averages through an analog of minimally entangled typical thermal
  states. We implement these algorithms with exact classical emulation as well as in prototype circuits on
  the Rigetti quantum virtual machine and  Aspen-1 quantum processing unit,
  demonstrating the power of quantum elevations of classical algorithms.

\section{Introduction}
An important application for a quantum computer is to compute the ground-state $\Psi$ of a Hamiltonian $\oper{H}$
\cite{Feynman_IJTP_1982,Abrams_PRL_1997,Abrams_PRL_1999}. 
This arises in simulations, for example, of the electronic structure of molecules and materials
\cite{Lloyd_Science_1996,Aspuru_Science_2005},
as well as in 
optimization when the cost function is encoded in a Hamiltonian. 
While efficient ground-state determination
cannot be guaranteed for all Hamiltonians, as this is a QMA complete
problem 
\cite{Kempe_SIAM_2004}, several heuristic quantum algorithms 
have been proposed, such as adiabatic state preparation with quantum phase estimation (QPE)
\cite{Farhi_MIT_2000,Kitaev_arxiv_1995} and quantum-classical variational algorithms,
including the quantum approximate optimization algorithm (QAOA)
\cite{Farhi_MIT_2014,Otterbach_arxiv_2017,Moll_QST_2018} and variational quantum eigensolver (VQE)
\cite{Peruzzo_Nature_2013,McClean_NJP_2016,grimsley2018adapt}.
While there have been many advances with these algorithms, they also have potential disadvantages, 
especially in the context of near-term quantum computing architectures and limited quantum resources. 
For example, phase estimation produces a nearly exact eigenstate, but appears impractical without error correction, while
variational algorithms, although somewhat robust to coherent errors, are limited in accuracy for a fixed variational form,
and involve a high-dimensional noisy classical optimization~\cite{mcclean2018barren}.

In classical simulations, different strategies are employed to numerically determine exact ground-states of Hamiltonians.
One popular approach is imaginary-time evolution, 
which expresses the ground-state as the long-time limit of the imaginary-time Schr\"odinger equation 
$- \partial_\beta |\Phi(\beta)\rangle 
= \oper{H} |\Phi(\beta)\rangle$, $|\Psi\rangle 
= \lim_{\beta \to \infty} \frac{|\Phi(\beta)\rangle}{ \| \Phi(\beta) \|}$ (for $\langle \Phi(0) | \Psi \rangle \neq 0$). Unlike variational algorithms with a fixed ansatz, imaginary-time evolution always converges to the ground-state 
(as distinguished from imaginary-time ansatz optimization, which can be trapped in local minima 
\cite{McArdle_arxiv_2018}).
Another common exact algorithm is the iterative Lanczos algorithm 
\cite{Lanczos_somewhere_1950,Arnoldi_somewhere_1951}
and its variations.
The Lanczos iteration constructs the Hamiltonian matrix $\mathbf{H}$ in a successively enlarged Krylov subspace
$\{ |\Phi\rangle, \oper{H} |\Phi\rangle, \oper{H}^2|\Phi\rangle \ldots \}$; diagonalizing $\mathbf{H}$
yields a variational estimate of the ground-state which tends to $|\Psi\rangle$ for a large number
of iterations. For a Hamiltonian on $N$ qubits, the
classical complexity of imaginary time evolution and the Lanczos algorithm
scales as $\sim 2^{\mathcal{O}(N)}$ in space as well as time.
The exponential space comes from storing $\Phi(\beta)$ or the Lanczos vector, while exponential time
comes from the cost of Hamiltonian multiplication $\oper{H} |\Phi\rangle$, as well as, in principle,
though not in practice, the $N$-dependence of the number of propagation steps and propagation time, or number of Lanczos iterations.
Thus it is natural to consider quantum versions of these algorithms that can overcome the exponential bottlenecks.

In this work, we will describe the quantum imaginary time evolution (QITE) and the quantum Lanczos (QLanczos) algorithms
to determine ground-states (as well as excited states in the case of QLanczos) on a quantum computer.
Compared to their classical counterparts, these achieve an exponential reduction in space for a fixed number of propagation
steps or number of iterations, and for a given iteration or time-step offer an exponential reduction in time.
They also offer advantages over existing ground-state quantum algorithms; compared to quantum phase estimation, they do not require
deep circuits, 
and compared to variational ground-state algorithms with a fixed ansatz, they are guaranteed to converge to the ground-state, avoiding non-linear
optimization. A crucial component of our algorithms is the efficient implementation of the non-Hermitian operation of an imaginary time step propagation
$e^{-\ts \oper{H} }$ (for small $\ts$), assuming a finite correlation length in the state.
Non-Hermitian operations are not natural on a quantum computer and are usually achieved using ancillae and postselection.
We will describe how to implement imaginary time evolution on a given state, without ancillae or postselection.
The lack of ancillae and complex circuits make QITE and QLanczos potentially
suitable for near-term quantum architectures. 
Using the QITE algorithm, we further show how we can sample from thermal (Gibbs) states, also without deep circuits or ancillae 
as is usually the case, via a quantum analog of the minimally entangled typical thermal states (QMETTS) algorithm~\cite{White_PRL_2009,Miles_NJP_2010}. We demonstrate the algorithms on spin and fermionic Hamiltonians (short- and long-range spin and Hubbard models, MAXCUT optimization, and dihydrogen
minimal molecular model) using exact classical emulation, and demonstrate proof-of-concept implementations on the
Rigetti quantum virtual machine (QVM) and Aspen-1 quantum processing units (QPUs).

\section{Quantum imaginary-time evolution}
Define a geometric $k$-local Hamiltonian $\oper{H} = \sum_\idx \oper{h}_\idx$ (where each term $\oper{h}_\idx$ acts on at most $k$ neighbouring qubits on an underlying graph)
and a Trotter decomposition of the corresponding imaginary-time evolution,
\begin{align}
  e^{-\beta \oper{H}} = (e^{-\ts \oper{h}_1} e^{-\ts \oper{h}_2} \ldots)^n + \mathcal{O}\left( {\ts} \right); \ n= \frac{\beta}{\ts},
\end{align}
applied to a state $|\Psi\rangle$. After a single Trotter step, we have
\begin{align}
  |\Psi^\prime \rangle = e^{-\ts \oper{h}_\idx} |\Psi\rangle \quad.
\end{align}
The basic idea is that the normalized state $|\bar{\Psi}^\prime \rangle = |\Psi^\prime \rangle / \| \Psi^\prime \|$ can be 
generated from $|\Psi\rangle$ by a unitary operator $e^{-i \ts \oper{A}[\idx]}$ (which also depends on imaginary-time step) acting in the neighbourhood
  of the qubits acted on by $\oper{h}_\idx$, 
where the Hermitian operator $\oper{A}[\idx]$ can be determined from tomography of $|\Psi\rangle$ in this neighbourhood up to controllable errors. 
This is illustrated by the simple example where $|\Psi\rangle$ is a product state. Then, the squared norm
$c = \| \Psi^\prime \|^2$ can be calculated from the expectation value of $\oper{h}_\idx $, which requires measurements 
over $k$ qubits,
\begin{align}
c = \langle \Psi | e^{-2\ts \oper{h}[\idx]} | \Psi\rangle = 1 - 2\ts \langle \Psi|  \oper{h}_ \idx |\Psi\rangle + \mathcal{O}(\ts^2).
\end{align}
Because $|\Psi \rangle$ is a product state, $|\Psi^\prime \rangle$ is obtained by acting the unitary operator 
$e^{-i\ts \oper{A}[\idx]}$ also on $k$ qubits. 
$\oper{A}[\idx]$ can be expanded in terms of an operator basis, such as the Pauli basis 
$\{ \sigma_i\}$ on $k$ qubits,
\begin{align}
\oper{A}[\idx] = \sum_{i_1i_2 \ldots i_k} a[\idx]_{i_1i_2 \ldots i_k} \sigma_{i_1}\sigma_{i_2} \ldots \sigma_{i_k},
\label{eq:aoperator}
  \end{align}
where $I$ denotes the index $i_1i_2 \ldots i_D$. Then, up to $\mathcal{O}(\ts)$, the vector of coefficients $a[\idx]_{i_1i_2 \ldots i_k}$ can be determined from the linear 
system 
\begin{equation}\label{eq:lineareq1}
\mathbf{S} \mathbf{a}[\idx] = \mathbf{b},
\end{equation}
where the elements of $\mathbf{S}$ and $\mathbf{b}$ are expectation values over $k$ qubits of $\Psi$, namely
\begin{align}
  S_{i_1i_2 \ldots i_k, i_1'i_2' \ldots i_k'} &= \langle \Psi|\sigma_{i_1}^\dag\sigma_{i_2}^\dag \ldots \sigma_{i_k}^\dag \sigma_{i_1'}\sigma_{i_2'} {\color{blue} \dots} \sigma_{i_k'}|\Psi\rangle \notag \\
  b_{i_1i_2 \ldots i_k} &= -i \, c^{-\frac{1}{2}} \, \langle \Psi|\sigma_{i_1}^\dag\sigma_{i_2}^\dag \ldots \sigma_{i_k}^\dag \oper{h}[\idx] |\Psi\rangle.
\end{align}
In general, $\mathbf{S}$ will have a null space; to ensure $\mathbf{a}[\idx]$ is real, we minimize
$\|  c^{-1/2}\Psi^\prime -(1-i\ts \oper{A}[\idx]) \Psi\|$ w.r.t. real variations in $\mathbf{a}[\idx]$. Note that the solution
is determined from a linear problem, thus there are no local minima. 

In this simple case, the normalized result of the imaginary time evolution step could be represented by a
unitary over $k$ qubits, because $|\Psi\rangle$ had a zero correlation length. After the initial step, this is no longer the case.
However, for a more general $|\Psi\rangle$ with finite correlation length extending over $C$ qubits (meaning
that the correlations between two observables separated by distance $l$ are bounded by $\exp(- l /C )$),
$|\Psi^\prime \rangle$
can be generated by a unitary acting on a domain of width  $D:= \log(1/\delta) C$ qubits surrounding
the qubits acted on by $h_i$ (this follows from Uhlmann's theorem~\cite{uhlmann}; see Appendix for a proof), with $\delta$ the
approximation error for that time step.
The unitary $e^{-i \ts A[i]}$ can then be determined by measurements and solving the least squares problem over $D$ qubits. For example, if we consider a nearest-neighbor local Hamiltonian on a $d$-dimension square lattice, the number of qubits $D$ where the unitary acts is bounded by $(2\log(1/\delta) C)^d$. 
Because correlations are induced only by the gates applied at previous time steps, the correlation length increases at most with a velocity bounded by
a constant $\alpha_v$ which depends on the geometry of the lattice and the locality of interactions. Consequently, each successive imaginary time step can be simulated by a unitary over an increasingly large neighborhood whose size propagates with velocity bounded by $\alpha_v$ (Fig. 1). 

 The number of measurements and classical storage at an imaginary time $\beta$ (starting the propagation from a product state) is bounded by $\exp(O((\alpha_v \beta)^d))$ for each unitary update, since each unitary at that level acts on at most $(2 \alpha_{v} \beta)^d$ sites; classical solution of the least squares equation has the same scaling $\exp(O((\alpha_v \beta)^d))$, as does the synthesis and application  of the unitary $e^{-i \ts A[i]}$. Thus, space and time requirements are bounded by exponentials of $\beta^d$, but are polynomial in $N$ (the polynomial in $N$ comes from the number of terms in $H$ and from the control of the Trotter error).

\begin{figure}[t!]
\centering
\includegraphics[width=1\textwidth]{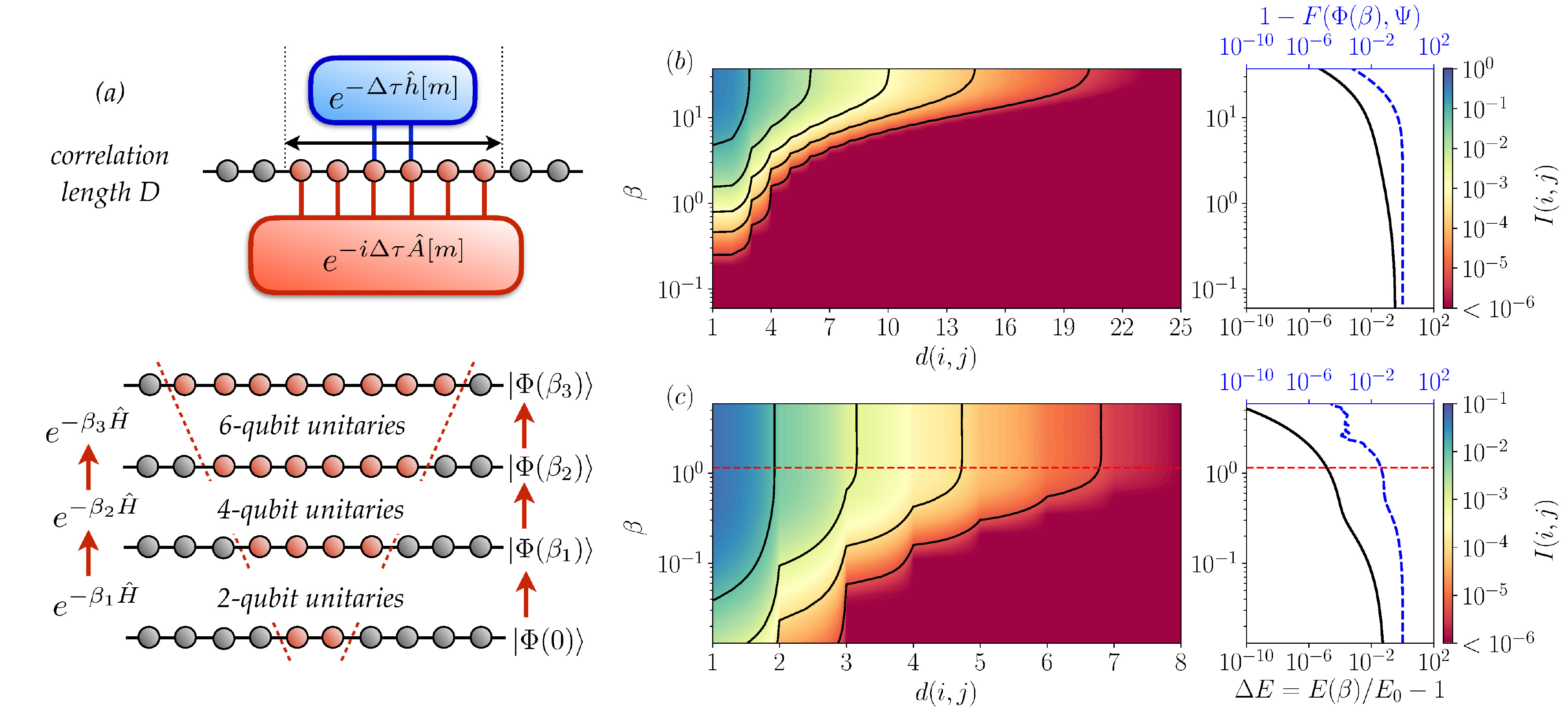} 
\caption{Quantum imaginary time evolution algorithm and correlation length.
(a) Schematic of the QITE algorithm.
Top: imaginary-time evolution under a geometric $k$-local operator $\hat{h}[m]$
can be reproduced by a unitary operation acting on a group of $D>k$ qubits.
Bottom: 
exact imaginary-time evolution starting from a product state requires 
unitaries acting on a domain $D$ that grows with $\beta$.
(b,c) Left: mutual information $I(i,j)$ between qubits $i$, $j$ as a function of 
distance $d(i,j)$ and imaginary time $\beta$, for a 1D (b) and a 2D (c) 
FM transverse-field Ising model, with $h=1.25$ (1D) and $h=3.5$ (2D). The mutual information
is seen to saturate at longer times.
Right: relative error in the energy $\Delta E$ and fidelity
$F= |\langle \Phi(\beta) | \Psi \rangle |^2$ between
the finite-time state $\Phi(\beta)$ and infinite-time state $\Psi$ as a function
of imaginary time. The noise in the 2D fidelity error at large $\beta$ arises from
the approximate nature of the algorithm used. }
\label{fig:1}
\end{figure}

\noindent \emph{Saturation of correlations}.
Note that the correlation volume cannot be larger than $N$. In many physical systems, we expect the correlation volume to
increase with $\beta$ and saturate for $C^d\ll N$ \cite{Hastings_CMP_2005}. As an example, in Fig.~\ref{fig:1} we plot the mutual information between qubits $i$ and $j$ for
the 1D and 2D FM transverse field Ising models computed by tensor network simulation 
which shows a monotonic increase and clear saturation. If saturation occurs before the ground-state is attained, the
cost of the algorithm for subsequent time-steps becomes linear in $\beta$,
and exponential in $C^d$. \\
\noindent \emph{Comparison to classical algorithm}. Unlike classical imaginary time evolution, QITE 
is bounded by an exponential in $\beta$, rather than an exponential in $N$. Thus for fixed $\beta$ (and
the same number of Trotter steps),
we achieve an exponential reduction in cost in space and time in $N$ compared to the classical algorithm. \\
\noindent \emph{Comparison to tensor networks}. If $|\Psi\rangle$ is represented by a tensor network in a classical simulation,
then $e^{-\ts \oper{h}[\idx]}|\Psi\rangle$ can be obtained directly as a classical tensor network with an increased bond dimension
\cite{VidalTEBD,Schollwock_Annals_2011}.
This bond dimension
increases exponentially with imaginary time $\beta$, thus the storage of the tensors, as well as the
cost of applying the imaginary time step $e^{-\ts \oper{h}[\idx]}$ to the tensors  grows  exponentially with $\beta$, similar
to the quantum algorithm. The key distinction is that, other than in one dimension, we cannot guarantee that contracting the resulting
classical tensor network to evaluate observables is efficient; it is a \#P-hard problem
in the worst case in two dimensions (and even in the average case for Gaussian distributed tensors)~\cite{Schuch_PRL_2007,PEPShard2018};
no such problem exists in the quantum algorithm.\\
\noindent \emph{Fermionic Hamiltonians}. For fermions, a non-local mapping to spins (e.g. through the Jordan-Wigner transformation) would
violate the $k$-locality of the Hamiltonian. In principle, this can be bypassed by using a local mapping to spins {\cite{Verstraete_JSM_2005}.
  Alternatively, we conjecture that by using
a fermionic unitary, where the Pauli basis in Eq.~\eqref{eq:aoperator} is replaced by the fermionic operator basis 
$\{ 1, \oper{a}, \oper{a}^\dag, \oper{a}^\dag \oper{a} \}$}, 
 the area of support for the fermionic unitary grows in the same fashion as the standard unitary for geometric $k$-local 
Hamiltonians described above. This can be tested in numerical simulations.\\
\noindent \emph{Long-range Hamiltonians}. Consider a $k$-local Hamiltonian with long-range terms on a lattice, such as
a general pairwise Hamiltonian. Then the action of $e^{-\ts \oper{h}[\idx]}$, if $\oper{h}[\idx]$ acts
on qubits $i$ and $j$, can be emulated by a unitary constructed in the neighborhood of $i$ and $j$, over $(2C \log(1/\delta))^k$ sites.\\
\noindent \emph{Inexact time evolution}. Given limited resources, we can choose to measure and construct
the unitary over a reduced number of sites $D' < D(\beta)$. For example, if $D' = 1$, this gives a mean-field
approximation of the imaginary time evolution. While the unitary is no longer an exact
representation of the imaginary time evolution, there is no issue of a local minimum in its construction, although
the energy is no longer guaranteed to decrease in every time step. In this case, one might apply inexact imaginary time evolution
simply until the energy stops decreasing. Alternatively, with limited resources, one may apply the quantum Lanczos algorithm described below.\\
\noindent\emph{Stabilization}.
Sampling noise in the expectation values of the Pauli operators can affect the solution to Eq.~\ref{eq:lineareq1} that sometimes leads to numerical instabilities. We regularize $\mathbf{S}+\mathbf{S}^T$ against such statistical errors by adding a small $\delta$ to its diagonal. To generate the data presented in Fig.~\ref{fig:4} and Fig.~\ref{fig:5} of the main text, we used $\delta=0.01$ for 1-qubit calculations and $\delta=0.1$ for 2-qubits calculations.

\section{Quantum Lanczos algorithm}
Given the QITE subroutine, we now consider
how to formulate a quantum version of the Lanczos algorithm. A significant practical motivation is that the Lanczos algorithm
typically converges much more quickly than imaginary time evolution, and often in physical simulations only tens of iterations
are needed to converge to good precision. In addition, Lanczos provides a natural way to compute excited states.

In quantum Lanczos, we generate a set of wavefunctions for different imaginary-time projections of
an initial state $| \Psi \rangle$, using QITE as a subroutine. The normalized states are
\begin{equation}
|\Phi_l \rangle = \frac{ e^{- l \Delta \tau \hat{H} } | \Psi_T \rangle }{\| e^{- l \Delta \tau \hat{H} } \Psi_T \|} 
\equiv n_l \, e^{- l \Delta \tau \hat{H} } | \Psi_T \rangle \quad 0 \leq l < L_\text{max} \quad .
\end{equation}
where $n_l$ is the normalization constant.
For the exact imaginary-time evolution and $l$, $l^\prime$ both even (or odd) the matrix elements
\begin{equation}
S_{l,l^\prime} = \langle \Phi_l | \Phi_{l^\prime} \rangle 
\quad,\quad
H_{l,l^\prime} = \langle \Phi_l | \hat{H} | \Phi_{l^\prime} \rangle 
\end{equation}
can be computed in terms of expectation values (i.e. experimentally accessible quantities) only. Indeed, defining
$2r = l+l^\prime$, we have
\begin{equation}
S_{l,l^\prime} = n_l n_{l^\prime} \, \langle \Psi_T | e^{- l \Delta \tau \hat{H} } e^{- l^\prime \Delta \tau \hat{H} } | \Psi_T \rangle
= \frac{n_l n_{l^\prime}}{n_{r}^2} \quad ,
\end{equation}
and similarly
\begin{equation}
H_{l,l^\prime} = n_l n_{l^\prime} \, \langle \Psi_T | e^{- l \Delta \tau \hat{H} } \hat{H} e^{- l^\prime \Delta \tau \hat{H} } | \Psi_T \rangle
= \frac{n_l n_{l^\prime}}{n_{r}^2} \, \langle \Phi_r | \hat{H} | \Phi_r \rangle = S_{l,l^\prime} \, \langle \Phi_r | \hat{H} | \Phi_r \rangle \quad .
\end{equation}
The quantities $n_r$ can be evaluated recursively, since
\begin{equation}
\frac{1}{n^2_{r+1}} = \langle \Psi_T | e^{- (r+1) \Delta \tau \hat{H} } e^{- (r+1) \Delta \tau \hat{H} } | \Psi_T \rangle = 
\frac{ \langle \Phi_r | e^{-2 \Delta \tau \hat{H} } | \Phi_r \rangle }{n_r^2} \quad.
\end{equation}

For inexact time evolution, the quantities $n_r$ and $\langle \Phi_r | \hat{H} | \Phi_r \rangle$ can still be used to
approximate $S_{l,l^\prime}$, $H_{l,l^\prime}$.

Given these matrices, we then solve the generalized  eigenvalue equation $\mathbf{H}\mathbf{x} = E \mathbf{S}\mathbf{x}$ to find an approximation
to the ground-state $| \Phi' \rangle = \sum_l x_l | \Phi_l \rangle$ for the ground state of $\oper{H}$. This eigenvalue equation
can be numerically ill-conditioned, as $S$ can contain small and negative eigenvalues for several reasons: (i)
as $m$ increases the vectors $|\Phi_l \rangle$ become linearly dependent; (ii) simulations have finite
precision and noise; (iii) $S$ and $H$ are computed approximately when inexact time evolution is performed.

To regularize the problem, out of the set of time-evolved states we extract a well-behaved sequence as follows:
(i) start from $|\Phi_\text{last}\rangle = |\Phi_0\rangle$, (ii) add the next $|\Phi_l\rangle$ in the set
of time-evolved states s.t. $|\langle \Phi_l | \Phi_\text{last}\rangle| < s$, where $s$
is a regularization parameter $0<s<1$, (iii) repeat, setting the $|\Phi_\text{last}\rangle=\Phi_l$ (obtained from (ii)), until
the desired number of vectors is reached.

We then solve the generalized eigenvalue equation $\tilde{\mathbf{H}}\mathbf{x} = E \tilde{\mathbf{S}}\mathbf{x}$ spanned by this regularized sequence,
removing any eigenvalues of $\tilde{\mathbf{S}}$ less than a threshold $\epsilon$.
The QLanczos calculations reported in Fig.~\ref{fig:2}  (lower panel) of the main text were stabilized with this algorithm,
in both cases using stabilization parameter $s=0.95$ and $\epsilon = 10^{-14}$. The stabilization parameters used in the QLanczos calculations reported in Fig.~\ref{fig:4} are $s=0.75$ and $\epsilon = 10^{-2}$.

We demonstrate the QLanczos algorithm using classical emulation on the 1D Heisenberg Hamiltonian,
as used for the QITE algorithm above in Fig. \ref{fig:2}.
Using exact QITE (large domains) to generate the matrix elements,
quantum Lanczos converges much more rapidly than imaginary time evolution. Using inexact QITE (small domains), the convergence
is usually faster and also reaches a lower energy.
We also assess the feasibility of QLanczos in the presence of noise, using emulated noise on the Rigetti QVM as well as
on the Rigetti Aspen-1 QPUs.
In Fig. \ref{fig:4}, we see
that QLanczos also provides more rapid convergence than QITE with both noisy classical emulation as well as on the physical device
for 1- and 2-qubits.

\section{Quantum thermal averages}
The QITE subroutine can be used in a range of other algorithms. As one 
example, we now discuss how to compute thermal averages of operators i.e. $\mathrm{Tr}\big[ \oper{O}  e^{-\beta \oper{H}} \big]
/ \mathrm{Tr} \big[ e^{-\beta \oper{H}} \big]$ using imaginary time evolution.
Several procedures have been proposed for quantum thermal averaging \cite{Terhal_PRA_2000}, ranging from generating the 
finite-temperature state explicitly with the help of ancillae, to a quantum analog of Metropolis sampling \cite{Temme_Nature_2011} 
that relies heavily on phase estimation. However, given a method for imaginary time evolution, one can generate thermal averages 
of observables without any ancillae or deep circuits. This can be done by adapting to the quantum setting the classical minimally entangled typical thermal 
state (METTS) algorithm \cite{White_PRL_2009,Miles_NJP_2010}, which generates a Markov chain from which the thermal average can be sampled.

 Consider the thermal average of
an observable $\hat{O}$
\begin{equation}
\langle \hat{O}\rangle = \frac{1}{Z}\mathrm{Tr}[e^{-\beta \hat{H}}\hat{O}] 
=  \frac{1}{Z}\sum_{i}\langle i|e^{-\beta \hat{H}/2} \hat{O} e^{-\beta \hat{H}/2}|i\rangle
\end{equation}
where $\{|i\rangle\}$ is an orthonormal basis set, and $Z$ is the partition
function. Defining $|\phi_i\rangle = P_i^{-1/2}e^{-\beta \hat{H}/2}|i\rangle$,
we obtain
\begin{equation}\label{eq:thermal_sum}
\langle \hat{O}\rangle = \frac{1}{Z} \sum_i P_i \langle \phi_i|\hat{O}|\phi_i\rangle 
\end{equation}
where $P_i = \langle i|e^{-\beta H}|i\rangle$. The summation in Eq.(\ref{eq:thermal_sum}) can be estimated by sampling
 $|\phi_i\rangle$ with probability $P_i/Z$, and summing
the sampled $\langle \phi_i|\hat{O}|\phi_i\rangle$.

In standard Metropolis sampling for thermal states, one starts from $|\phi_i\rangle$ and obtains the next state
$|\phi_j\rangle$ from randomly proposing and accepting based an
acceptance probability. However, rejecting and resetting  in the quantum analog of Metropolis~\cite{Temme_Nature_2011} is complicated to
implement on a quantum computer, requiring deep circuits.
The METTS algorithm provides an alternative way to sample
$|\phi_i\rangle$ distributed with probability $P_i/Z$ without this complicated procedure.
The algorithm is as follows
\begin{enumerate}
\item Choose a classical product state (PS) $|i\rangle$.
\item Compute $|\phi_i\rangle = P_i^{-1/2}e^{-\beta H/2}|i\rangle$ and
calculate observables of interest.
\item Collapse $|\phi_i\rangle$ to a new PS $|i'\rangle$ with probability
$p(i\rightarrow i') = |\langle i'|\phi_i\rangle|^2$ and repeat Step 2.
\end{enumerate}

In the above algorithm, $|\phi_i\rangle$ is named a minimally entangled typical
thermal state (METTS).
One can easily show that the set of METTS sampled following the above
procedure has the correct Gibbs distribution~\cite{Stoudenmire2010}.
Generally, $\{|i\rangle\}$ can be any orthonormal basis.
For convenience when implementing METTS on a quantum computer,
$\{|i\rangle\}$ are chosen to be product states.

On a quantum emulator or a quantum computer, the METTS algorithm is carried out as following
\begin{enumerate}
\item Prepare a product state $|i\rangle$.
\item Imaginary time evolve $|i\rangle$ with the QITE algorithm to
$|\phi_i\rangle = P_i^{-1/2}e^{-\beta H/2}|i\rangle$, and measure
the desired observables.
\item Collapse $|\phi_i\rangle$ to another product state by measurement.
\end{enumerate}

In practice, to avoid  long statistical correlations between samples, we used
the strategy of collapsing METTS onto alternating basis sets~\cite{Stoudenmire2010}.
For instance, for the odd METTS steps, $|\phi_i\rangle$ is collapsed
onto the $X$-basis (assuming a $Z$ computational basis, tensor products of $|+\rangle$ and $|-\rangle$), and for
the even METTS steps, $|\phi_i\rangle$ is collapsed onto the $Z$-basis
(tensor products of $|0\rangle$ and $|1\rangle$). The statistical error is then estimated by block analysis~\cite{Flyvbjerg1989}.
%
In Fig. \ref{fig:5}a we show the results of  quantum METTS (using exact classical emulation) for the thermal average $\langle \oper{H}\rangle$
as a function of temperature $\beta$,
for the 6-site 1D AFM transverse-field Ising model for several temperatures and domain sizes; sufficiently
large $D$ converges to the exact thermal average at each $\beta$; error bars reflect only the finite samples in QMETTS.
We also show an implementation of quantum METTS
on the Aspen-1 QPU and QVM with a 1-qubit field model (Fig.~\ref{fig:5}b),
and using the QVM for a 2-qubit AFM transverse field Ising model (Fig.~\ref{fig:5}d);
while the noise introduces additional error
including a systematic shift (Fig.~\ref{fig:5}c), the correct behaviour of the thermal average with temperature is reproduced on the
emulated and actual quantum device.

\section{Results}
To illustrate the QITE algorithm, we have carried out exact classical emulations (assuming perfect
expectation values and perfect gates) for several Hamiltonians: short-range 1D Heisenberg;
1D AFM transverse-field Ising; long-range 1D Heisenberg with spin-spin coupling
$J_{ij} ={|i-j|+1}^{-1}; i\neq j$; 1D Hubbard at half-filling (mapped by Jordan-Wigner transformation to a spin model); a 6-qubit MAXCUT 
\cite{Farhi_MIT_2014,Otterbach_arxiv_2017,Moll_QST_2018} instance, and a minimal basis 2-qubit dihydrogen molecular Hamiltonian~\cite{OMalley2015}. We 
describe the models below.

\noindent\textbf{1D Heisenberg and transverse field Ising model}.
The 1D short-range Heisenberg Hamiltonian is defined as
\begin{align}
  \hat{H} =\sum_{\langle ij\rangle} \hat{\mathbf{S}}_i \cdot \hat{\mathbf{S}}_j \quad,
\end{align}
the 1D long-range Heisenberg Hamiltonian as
\begin{align}
\hat{H} =\sum_{i \neq j} \frac{1}{|i-j|+1} \, \hat{\mathbf{S}}_i \cdot \hat{\mathbf{S}}_j \quad,
\end{align}
and the AFM transverse-field Ising Hamiltonian as
\begin{align}
  \hat{H} = \sum_{\langle ij\rangle} \hat{{S}}^z_i  \hat{{S}}^z_j + \sum_i h \hat{S}^x_i \quad .
\end{align}

\noindent\textbf{1D Hubbard model}.
The 1D Hubbard Hamiltonian is defined as
\begin{align}
  \hat{H} = - \sum_{\langle ij \rangle \sigma} a^\dag_{i\sigma} a_{j\sigma} + U \sum_i \hat{n}_{i\uparrow} \hat{n}_{i\downarrow}
\end{align}
where $\hat{n}_{i \sigma} = a^\dag_{i\sigma} a_{i\sigma}$, $\sigma \in \{ \uparrow,\downarrow\}$, and $\langle \cdot \rangle$
denotes summation over nearest-neighbors, here with open-boundary conditions. We label the $n$
lattice sites with an index $i=0 \dots n-1$, and the $2n-1$ basis functions as $|\varphi_0 \rangle = |0 \uparrow \rangle$,
$|\varphi_1 \rangle = |0 \downarrow \rangle$, $|\varphi_2 \rangle = |1 \uparrow \rangle$,
$|\varphi_3 \rangle = |1 \downarrow \rangle$ $\dots$.
Under Jordan-Wigner transformation, recalling that
\begin{align}
\hat{n}_{p} = \frac{1-Z_p}{2} \quad,\quad
\hat{a}^\dag_p \hat{a}_q + \hat{a}^\dag_q \hat{a}_p =  \frac{X_p X_q \prod_{k=q+1}^{p-1} Z_k \left( 1- Z_p Z_q \right)}{2} \quad,
\end{align}
with $p=0 \dots 2n-2$ and $q<p$, the Hamiltonian takes the form
\begin{align}
\hat{H} = - \sum_p \frac{X_{p} X_{p+2} Z_{p+1} \left( 1- Z_{p} Z_{p+2} \right)}{2}
+ U \sum_{p \, \mathrm{even}} \frac{(1-Z_{2i}) (1-Z_{2i+1})}{4} + \mu \sum_p \frac{(1-Z_p)}{2}
\end{align}

\noindent\textbf{H$_2$ molecule minimal basis model}.
We use the hydrogen molecule minimal basis model at the STO-6G level of theory. This is a common minimal model of hydrogen chains
\cite{hachmann2006multireference,Motta_PRX_2017} and has previously been studied in quantum simulations, for example
in~\cite{OMalley2015}. Given a molecular
geometry (H-H distance $R$) we perform a restricted Hartree-Fock calculation and express the second-quantized Hamiltonian
in the orthonormal basis of RHF molecular orbitals as~\cite{szaboostlund}
\begin{equation}
\label{eq:H2}
\hat{H} = H_0 + \sum_{pq} h_{pq} \hat{a}^\dag_p \hat{a}_q + \frac{1}{2} \sum_{prqs} v_{prqs}
\hat{a}^\dag_p \hat{a}^\dag_q \hat{a}_s \hat{a}_r
\end{equation}
where $a^\dag$, $a$ are fermionic creation and annihilation operators for the molecular orbitals.
The Hamiltonian \eqref{eq:H2} is then encoded by a Bravyi-Kitaev transformation into the 2-qubit operator
\begin{equation}
\hat{H} = g_0 I \otimes I + g_1 Z \otimes I + g_2 I \otimes Z + g_3 Z \otimes Z + g_4 X \otimes X + g_5 Y \otimes Y \quad,
\end{equation}
with coefficients $g_i$ given in Table I of \cite{OMalley2015}.

\noindent\textbf{MAXCUT Hamiltonian}.
The MAXCUT Hamiltonian encodes the solution of the MAXCUT problem.
Given a graph $\Gamma = (V,E)$, where $V$
is a set of vertices and $E \subseteq V \times V$ is a set of links between vertices in $V$, a cut of $\Gamma$ is a subset
$S \subseteq V$ of $V$. The MAXCUT problem consists in finding a cut $S$ that maximizes the number of edges between $S$ and $S^c$ (the complement of $S$).
We denote the number of links in a given cut $S$ as $C(S)$.
The MAXCUT problem can be formulated as a Hamiltonian ground-state problem, by (i) associating a qubit to every vertex in $V$, (ii) associating to every
partition $S =$ an element of the computational basis (here assumed to be in the $z$ direction) of the form $| z_0 \dots z_{n-1} \rangle$, where $z_i = 1$ if
$i \in S$ and $z_i = 0$ if $i \in S^c$, and finding the minimal (most negative) eigenvalue of the $2$-local Hamiltonian
\begin{equation}
\hat{C} = -\sum_{(ij) \in E} \frac{1 - \hat{S}^z_i \hat{S}^z_j}{2} \quad .
\end{equation}
The spectrum of $\hat{C}$ is a subset of numbers $C \in \{ 0,1 \dots |E| \}$.

To assess the feasibility of implementation on near-term quantum devices, we have also carried out noisy
classical emulation (sampling expectation values and with an error model) using the Rigetti quantum virtual machine (QVM) and a
physical simulation using the Rigetti Aspen-1 QPUs, for a single qubit field model ($2^{-1/2}(X+Z)$)\cite{lamm2018simulation} and a
1D AFM transverse-field Ising model. 
We carry out QITE using different fixed domain sizes $D$ for the unitary or fermionic unitary.

For quantum simulations, we used pyQuil, an open source Python library, to write quantum circuits that interface with both Rigetti's quantum virtual machine (QVM) and
the Aspen-1 quantum processing units (QPUs). 
pyQuil provides a way to include noise models in the QVM simulations. Readout error can be included in a
high-level API provided in the package and is characterized by $p_{00}$ (the probability of reading $|0\rangle$ given that the qubit
is in state $|0\rangle$) and $p_{11}$ (the probability of reading $|1\rangle$ given that the qubit is in state $|1\rangle$). Readout errors
can be mitigated by estimating the relevant probabilities and correcting the estimated expectation values. We do so by using a high level
API present in pyQuil.

A general noise model can be applied to a gate in the circuit by applying the appropriate Kraus maps. Included in the package is
a high level API that applies the same decoherence error attributed to energy relaxation and dephasing to every gate in the circuit.
This error channel is characterized by the relaxation time $T_{1}$ and coherence time $T_{2}$. We also include in our emulation
our own high-level API that applies the same depolarizing noise channel to every single gate by using the appropriate Kraus maps.
The depolarizing noise is characterized by $p_{1}$, the depolarizing probability for single-qubit gates and $p_{2}$, the
depolarizing probability for two-qubit gates.

\subsection{Benchmarks}

\begin{figure}[t!]
\centering
\includegraphics[width=1\textwidth]{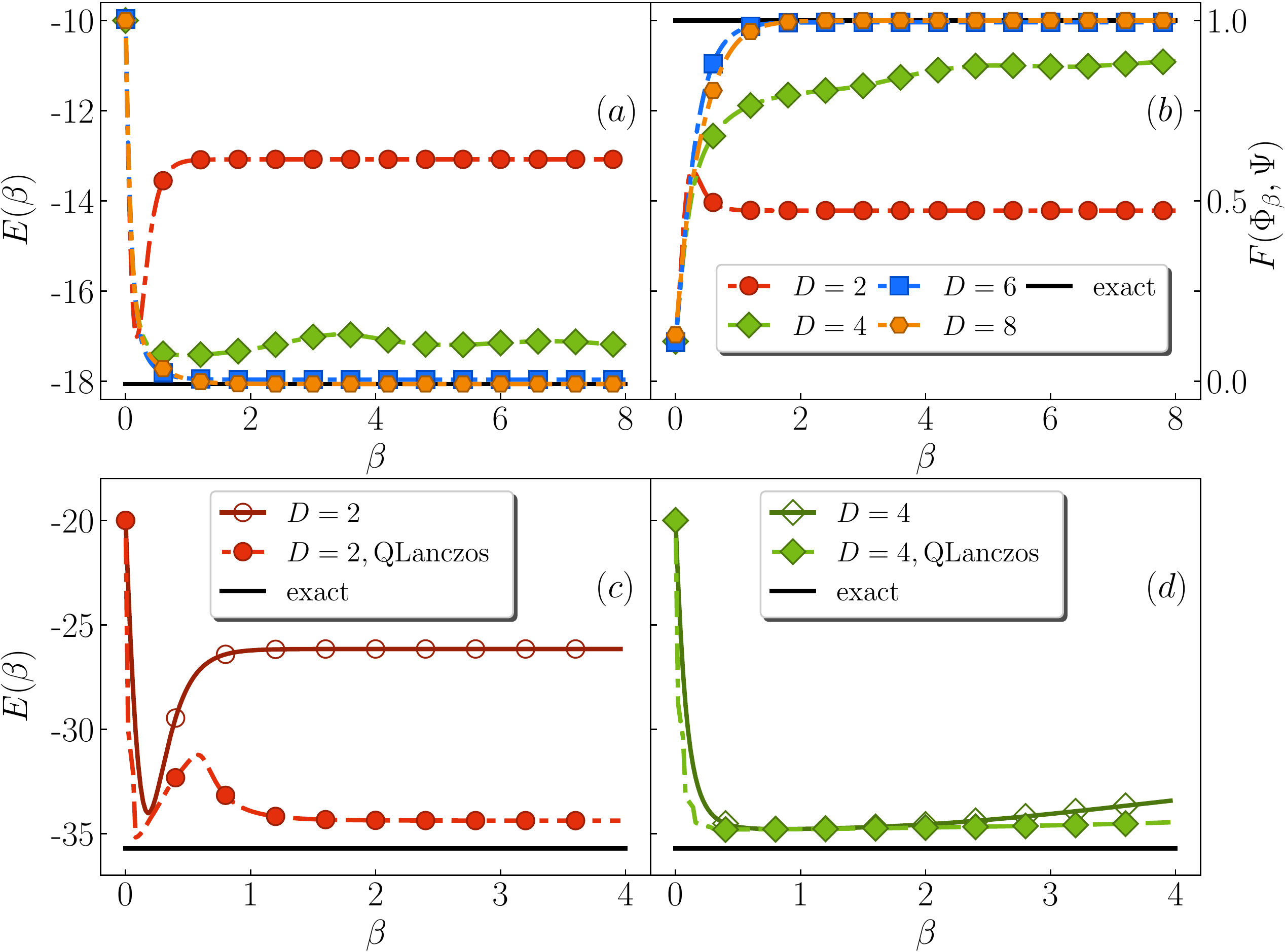} 
\caption{Energy calculations with QITE and QLanczos algorithms. Top: QITE energy $E(\beta)$ (a) and fidelity (b) between finite-time state
$\Phi(\beta)$ and exact ground state $\Psi$ as function of imaginary time $\beta$, 
  for a 1D 10-site Heisenberg model, showing the convergence with increasing unitary domains of $D=2-8$ qubits.  
Bottom: QITE (dashed red, dot-dashed green lines) and QLanczos (solid red, solid green lines) energies 
as function of imaginary time $\beta$, for a 1D Heisenberg model with $N=20$ qubits, using domains 
of $D=2$ (c) and $4$ qubits (d), showing improved convergence of QLanczos over QITE. Black line
  is the exact ground-state energy/fidelity.}
\label{fig:2}
\end{figure}

Figs.~\ref{fig:2} and \ref{fig:3} show the energy obtained by QITE as a function of $\beta$ and $D$ for the various models. 
As we increase $D$, the asymptotic ($\beta \to \infty$) energies rapidly converge to the exact ground-state. For
small $D$, the inexact QITE tracks the exact QITE for a time until the 
correlation length exceeds $D$. Afterwards, it may go down or up. The non-monotonic behavior is strongest
for small domains; in the MAXCUT example, the smallest domain $D=2$ gives an oscillating energy. 
In such cases, we consider a reasonable estimate of the ground-state energy to be the point at which the energy stops decreasing. In all 
models, increasing $D$ past a maximum value (less than $N$) no longer 
affects the asymptotic energy, showing that the correlations have saturated (this is true even in the MAXCUT instance).

\begin{figure}[t!]
\centering
\includegraphics[width=1\textwidth]{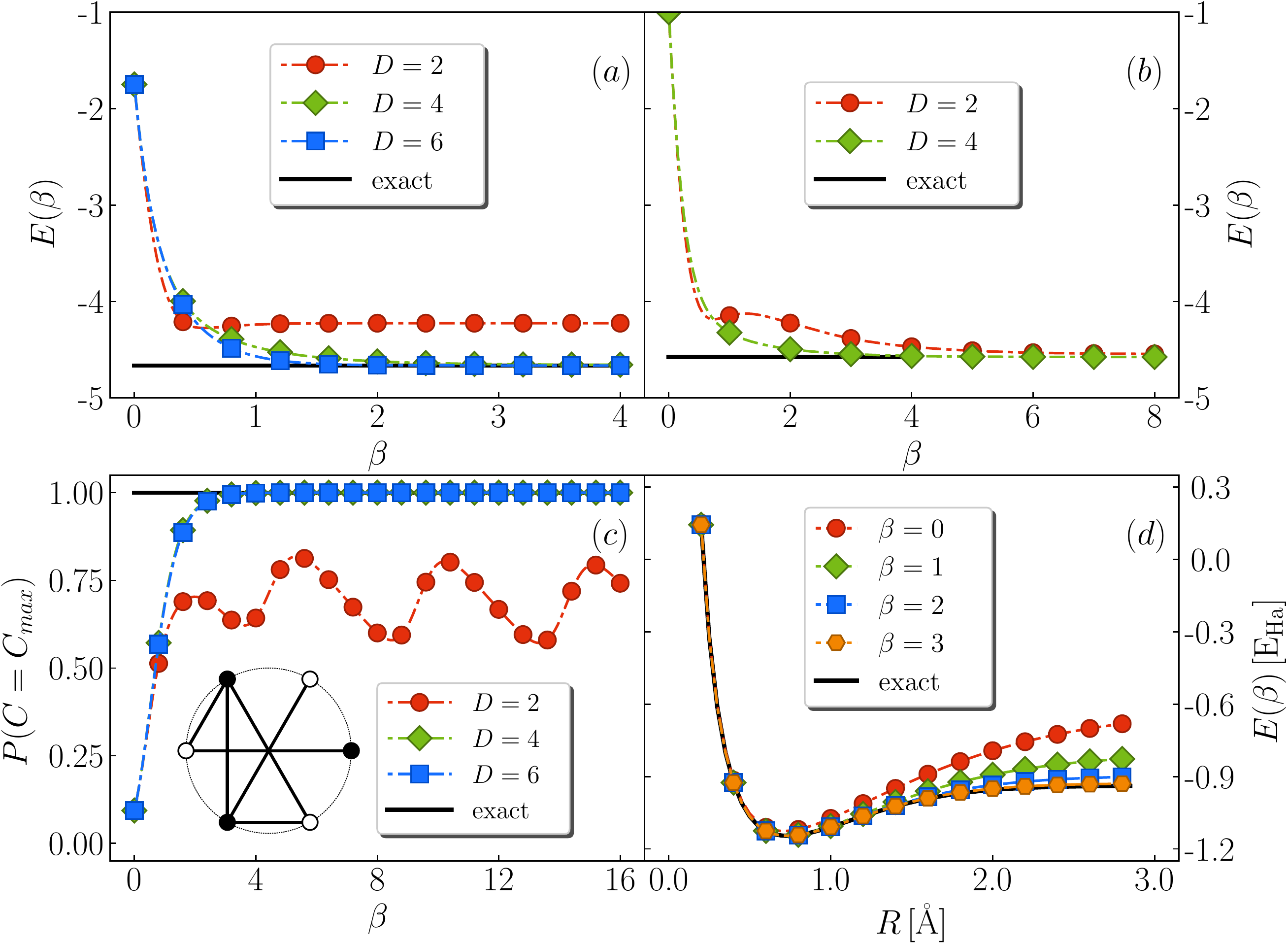} 
\caption{QITE energy evaluations. (a) QITE energy $E(\beta)$ as a function of imaginary time $\beta$ for 
a 6-site 1D long-range Heisenberg model, for unitary domains $D=2-6$;
(b) a 4-site 1D Hubbard model with $U/t = 1$, for unitary domains $D=2,4$; (d)
the H$_2$ molecule in the STO-6G basis. (c) Probability of MAXCUT detection, $P(C=C_{max})$ as a function of imaginary time $\beta$, for the
$6$-site graph in the panel. Black line is the exact ground-state energy/probability of detection.}
\label{fig:3}
\end{figure}

\begin{figure}[t!]
\centering
\includegraphics[width=1\textwidth]{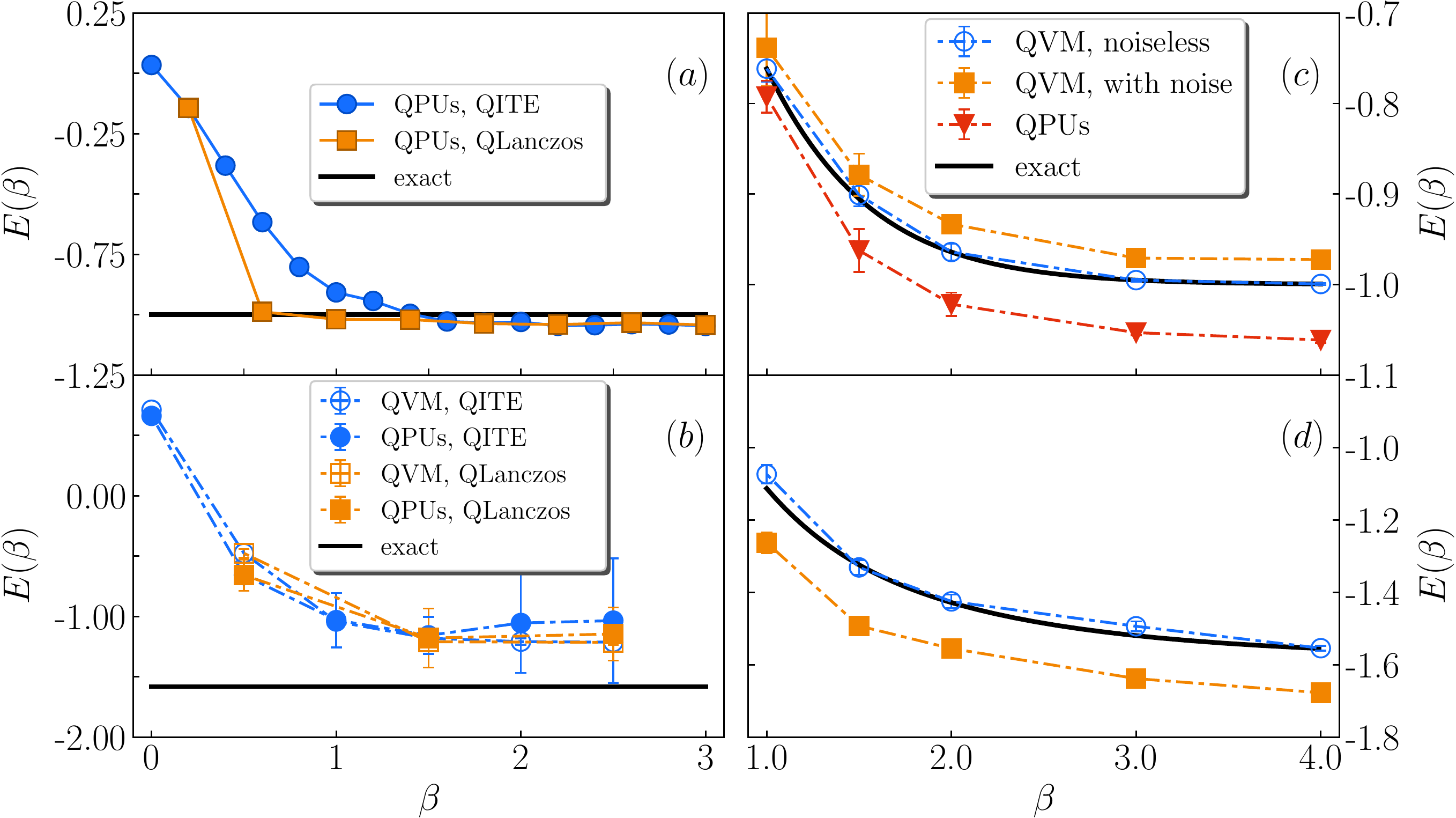}
\caption{QITE, QLanczos, and QMETTS energies $E(\beta)$ as a function of imaginary time $\beta$ for
  1-qubit field model using the QVM and QPU (qubit 14 on Aspen-1) and 2-qubit AFM transverse field Ising model
  using the QVM and QPU (qubit 14, 15 on Aspen-1). (a) Ground state energies
for 1-qubit field model using the QVM and QPU (qubit 14 on Aspen-1); (b) 
ground state energies for 2-qubit AFM transverse field Ising model
  using the QVM and QPU (qubit 14, 15 on Aspen-1); (c) finite temperature 
energies for  1-qubit field model using the QVM and QPU (qubit 14 on Aspen-1)
; and (d) finite temperature energies for 2-qubit AFM transverse field Ising model
  using the QVM.
Black lines are the exact solutions.}
\label{fig:4}
\end{figure}

Figs.~\ref{fig:4} shows the results of running the QITE algorithm on Rigetti's QVM and Aspen-1 QPUs for 1- and 2- qubits, respectively.
Encouragingly for near-term simulations, despite sampling errors and other errors such as gate, readout and incoherent errors present in the device, it is possible to converge to a ground-state energy close to the exact energy for the 1-qubit case. 
This result reflects a robustness that is sometimes informally observed in imaginary time evolution algorithms in which the ground state energy is approached even if the imaginary time step is not perfectly implemented. In the 2-qubit case, although the QITE energy converges,
there is a systematic shift which is reproduced on the QVM using available noise parameters for readout, decoherence and depolarizing noise~\cite{Rigetti}. (Remaining discrepancies between the emulator and hardware are likely attributable to cross-talk between parallel gates not included in the noise model) . However, reducing decoherence and depolarizing errors in the QVM or using different sets of
qubits with improved noise characteristics all lead to improved convergence to the exact ground-state energy.

\begin{figure}[t!]
\centering
  \includegraphics[width=0.8\textwidth]{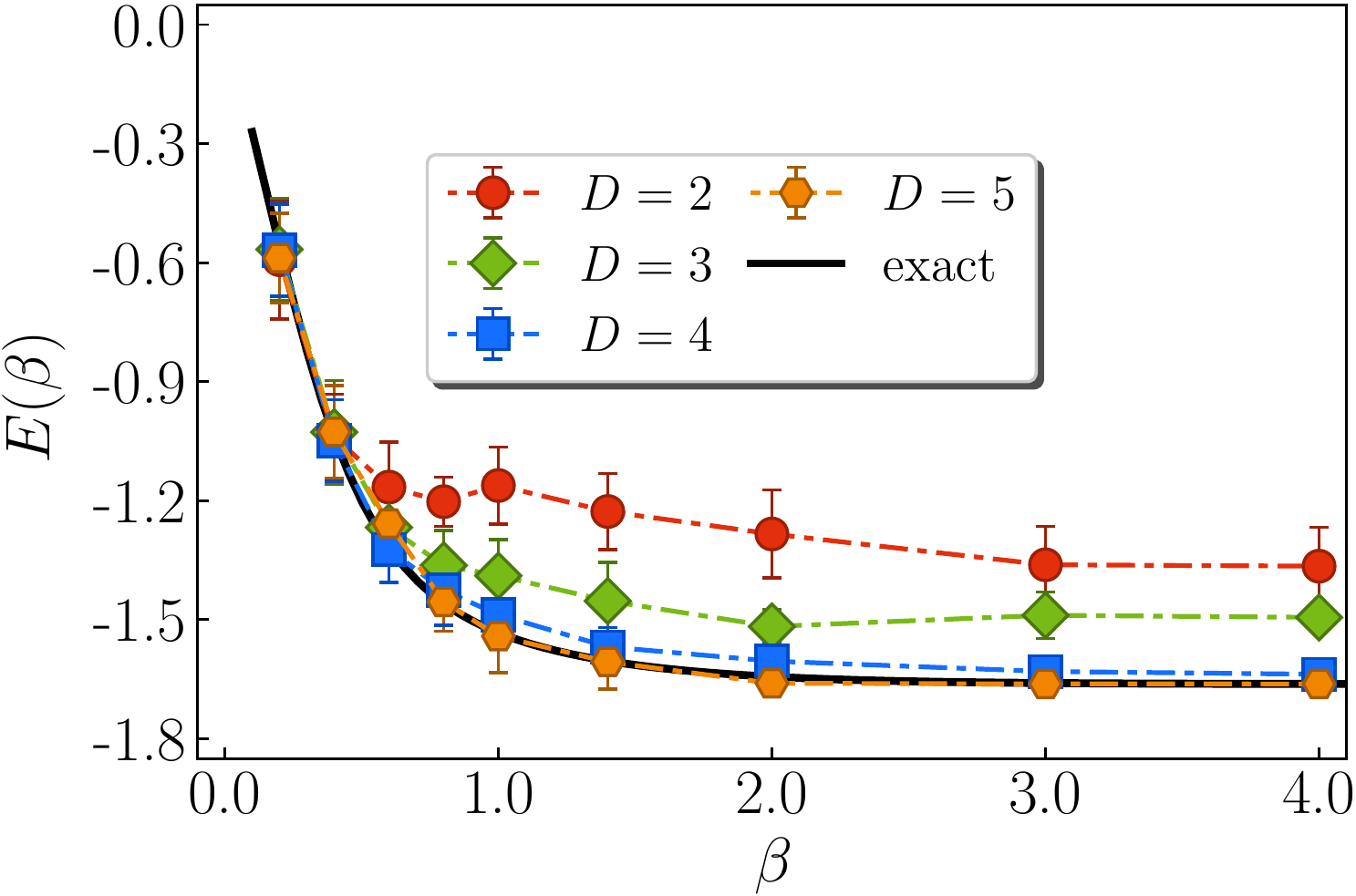}
  
\caption{Thermal (Gibbs) average $\langle E \rangle$ at temperature $\beta$ from QMETTS for a 1D 6-site 
  Heisenberg model (exact emulation).}
\label{fig:5}
\end{figure}

\section{Conclusions}.
 We have introduced quantum analogs
of imaginary time evolution (QITE) and the Lanczos algorithm (QLanczos), that can be carried out without ancillae or deep circuits,
and which achieve exponential reductions in space and time per iteration relative to their classical counterparts.
They provide new quantum routes to approximate ground-states of Hamiltonians in both physical simulations and in optimization
that avoid some of the disadvantages of phase estimation based approaches and variational algorithms. The QLanczos iteration
appears especially powerful if sufficient sampling can be done, as in practice it obtains accurate estimates of ground-states 
from only a few iterations, and also provides an estimate of excited states. Additionally, further algorithms that use QITE and QLanczos as subroutines can 
be formulated, such as a quantum version of the METTS algorithm to compute thermal averages. Encouragingly, these algorithms 
appear useful in conjunction with near-term quantum architectures, and serve to demonstrate the power of quantum elevations of 
classical simulation techniques, in the continuing search for quantum supremacy.

%


\newpage

\begin{appendices}
\chapter{Appendix for Chapter~\ref{chp:dmet} and Chapter~\ref{chp:hlatt}}
\section{Proof of the finite temperature bath formula}

Let $M$ be an arbitrary $N\times N$ full rank square matrix, and $Q_k$ be the $Q$ derived from the QR decomposition of the first $n$ columns of $M^k$, i.e., $M^k[:,:n] = Q_k R_k$, with $k = 0, 1, ..., K$. Let $S$ ($|S| < N$) be a space spanned by $\{Q_0, Q_1, ..., Q_K\}$, and $P$ be the projector onto $S$. The following equality holds
\begin{equation}\label{eq:2prove}
    P^{\dagger}M^lP[:,:n] = (P^{\dagger}MP)^l[:,:n], \hspace{0.2cm} l \leq K+1
.
\end{equation}

We prove the statement by mathematical induction. First write $M$ in the following form
\begin{equation}
    M = \begin{bmatrix}
    A & B \\
    C & D 
    \end{bmatrix},
\end{equation}
where $A$ and $B$ are the first $n$ rows of $M$, $A$ and $C$ are the first $n$
 columns of $M$. 
The projector has the form
\begin{equation}
    P = \begin{bmatrix}
        I & 0\\
        0 & V
        \end{bmatrix},
\end{equation}
where $I$ is an $n\times n$ matrix, and $V$ is an $(N-n)\times (K-1)n$ matrix with $(K-1)n < (N-n)$. The columns of $V$ are derived from the QR decomposition of $M^k[n:, :n]$, $k = 1, ..., K$ and then orthogonalized.
We can write $V$ in the form
\begin{equation}
    V = \begin{bmatrix} V_1 & V_2 & \cdots & V_K \end{bmatrix}
\end{equation}
where $V_k$ is from the QR decomposition of $M^k[n:, :n]$. $\PP M P$ has the 
form
\begin{equation}
\PP M P = \begin{bmatrix}
        A & BV\\
        \V C & \V DV
        \end{bmatrix}.
\end{equation}

The mathematical induction consists of two parts:

(i) We start with $l=2$. The first $n$ columns of $P^{\dagger}M^2P$ and $(P^{\dagger}MP)^2$ are
\begin{equation}
\begin{split}
   P^{\dagger}M^2P[:,:n] &=  \begin{bmatrix}
    A^2 + BC \\ \V CA + \V DC
    \end{bmatrix}\\
    (P^{\dagger}MP)^2[:,:n] &= \begin{bmatrix}
    A^2 + BV\V C \\ \V CA + \V DV\V C
    \end{bmatrix}.
\end{split}
\end{equation}
The two are equal when 
\begin{equation}\label{eq:VVC}
V\V C = V\V (VR) = V I R = VR = C
\end{equation}
which is true since $V$ is the $Q_1$ from the QR decomposition of $C$. (Note that $\V V = I$, but $V\V \neq I$).
 Therefore, Eq.~(\ref{eq:2prove}) holds for $l=2$ when $K \geq 1$. 

(ii) Now let us inspect Eq.~(\ref{eq:2prove}) for the $l$th order, assuming that Eq.~(\ref{eq:2prove}) holds for the $(l-1)$th order, i.e. $\PP M^{l-1}P = (\PP M P)^{l-1}$. Let 
\begin{equation}
    M^{l-1} = \begin{bmatrix}
    W & X \\
    Y & Z \\
    \end{bmatrix}
\end{equation}
and $M^l = MM^{l-1}$ has the form
\begin{equation}
    M^l = \begin{bmatrix}
    AW+BY  &  AX+BZ\\
    CW+DY  &  CX+DZ
    \end{bmatrix}
\end{equation}
and 
\begin{equation}
    \PP M^{l-1}P = (\PP MP)^{l-1} = \begin{bmatrix}
    W  &  XV\\
    \V Y  &  \V Z V
    \end{bmatrix}
\end{equation}
One can prove that $CW$ and $C$ share the same $Q$ space from the QR decomposition: let $C = QR$, then $CW = QRW$, where $R$ and $W$ are square matrices; we then perform another QR decomposition of $RW$, $RW = U\tilde{R}$, where $U$ is a unitary matrix, then $CW = \tilde{Q}\tilde{R}$ with $\tilde{Q} = QU$. Therefore, $Q$ and $\tilde{Q}$ span the same space. 

The first $n$ columns of $\PP M^lP$ and $(\PP M P)^l$ are
\begin{equation}
    \PP M^l P[:,:n]  = \begin{bmatrix}
    AW + BY \\
    \V CW + \V DY 
    \end{bmatrix},
\end{equation}

\begin{equation}
\begin{split}
    (\PP M P)^l[:,:n] =& \left((\PP M P)(\PP M P)^{l-1}\right)[:,:n]  \\
=& \begin{bmatrix}
     AW + BV\V Y\\
    \V CW + \V DV\V Y 
    \end{bmatrix}.
\end{split}
\end{equation}
Since $V$ contains $V_{l-1}$, which is derived from the QR decomposition of $Y$, we have $V\V Y  = Y$ as in Eq.~(\ref{eq:VVC}). 

Combining (i) and (ii) we then see that Eq.~(\ref{eq:2prove}) holds for the $l$th order with $K\geq l-1$ for $\forall l$. \QEDB

\section{Analytic gradient of the cost function for correlation potential fitting in DMET at finite temperature}
We rewrite the gradient of the cost function Eq.~\eqref{eq:cost_func_dmet} here
\begin{equation}\label{eq:gradient_cost_apdx}
\frac{\mrm{d}f}{\mrm{d}u_{kl}} = \sum_{i,j\in \text{imp}}2(D_{ij}^{\text{low}} - D_{ij}^{\text{high}}) 
\frac{\mrm{d}D_{ij}^{\text{low}} }{\mrm{d}u_{kl}},
\end{equation}
where $D^{\text{low}}$ is the single-particle density matrix from the
 mean-field (low-level) calculations, $D^{\text{high}}$ is the high-level
single partile density matrix, and $u$ is the correlation potential
matrix. The key to evaluate Eq.~\eqref{eq:gradient_cost_apdx} is to 
calculate $\frac{\mrm{d}D_{ij}^{\text{low}} }{\mrm{d}u_{kl}}$. For simplicity,
we will drop the superscript on  $D^{\text{low}}$. 

At finite temperature, $D$ is given by 
\begin{equation}
D = \frac{1}{1+e^{\beta (h - \mu + \delta u}},
\end{equation}
where $h$ is the one-body Hamiltonian, $\mu$ is the chemical potential 
(Fermi level),
and $\delta u$ is a small perturbation added to the Hamiltonian. 
Then $\frac{\mrm{d}D_{ij}^{\text{low}} }{\mrm{d}u_{kl}}$ has two parts:
\begin{equation}\label{eq:grad_D_apdx}
\frac{\mrm{d}D_{ij}(u, \mu(u)) }{\mrm{d}u_{kl}} = 
\frac{\partial D_{ij}}{\partial u_{kl}}\biggr\vert_{\mu}
+ \frac{\partial D_{ij}}{\partial \mu}\frac{\partial\mu}{\partial u_{kl}},
\end{equation}
where the second part comes from the change of Fermi level due to the 
change of correlation potential.

The first part of Eq.~\eqref{eq:grad_D_apdx} is evaluated by
\begin{equation}
\frac{\partial D_{ij}}{\partial u_{kl}}
= \sum_{pq} C_{ip}C^*_{kp}K_{pq}C_{lq}C_{jq}^*,
\end{equation}
where $C$ is the molecular orbital (MO) coefficient matrix with $ijkl$ the
site indices and $pq$ the MO indices, and 
\begin{equation}
K_{pq} = n_p (1-n_q)\frac{1-e^{\beta(\varepsilon_p-\varepsilon_q)}}{\varepsilon_p-\varepsilon_q},
\end{equation}
where $n_p$ is the occupation number on the $p$th orbital and $ \varepsilon_p$
is the energy of $p$th orbital. Note that when $\varepsilon_p = \varepsilon_q$,
both the denominator and numerator goes to zero and the value of $K_{pq}$ 
depends on $\beta$. When $\beta = \inf$, $\varepsilon_p = \varepsilon_q$ means
$n_p = n_q = 0$ or $1$, so $K_{pq} = 0$ is bounded. 

The second part is evaluated by
\begin{equation}
\begin{split}
\frac{\partial D_{ij}}{\partial \mu} &= \sum_p\beta C_{ip}n_p(1 - n_p)
C^*_{jp}\\
\frac{\partial\mu}{\partial u_{kl}} &= \frac{\sum_p n_p(1 - n_p)
C_{kp}^* C_{lp}}
{\sum_{p}n_p(1-n_p)}.
\end{split}
\end{equation}
The contribution of this part is usually small at low temperature and
becomes non-neglegible at higher temperature.

\section{Davidson diagonalization}\label{sec:apdx_davidson}
The Davidson diagonalization~\citep{Davidson1975} algorithm is an efficient
way to find the lowest/highest eigenvalues of a Hermitian matrix. 
In quantum chemistry, this method is widely used to get the ground 
state or low-lying excited states. This method constructs a subspace of the
Hilbert space from an initial vector as the guess of the ground state,
and diagonalize the Hamiltonian in this subspace. A preconditioner 
is used to make the algorithm more stable and converge fast.
The steps to evaluate $m$ lowest eigenvectors are listed below:
\begin{enumerate}
\item Select initial guess vectors $\mbf{v}^i, i = 1,...,n\geq m$ to form a 
 subspace $\mathcal{S}$. 
\item Construct the matrix representation of the Hamiltonian in the 
subspace $\mathcal{S}$: $\tilde{H}_{ij} = \mbf{v}_i^\dag\tilde{H}\mbf{v}_j$.
\item Diagonalize $\tilde{H}$ to obtain the lowest $m$ eigenvalues and
corresponding eigenvectors, $\tilde{H}\mbf{x}^p = \lambda_p \mbf{x}^p$.
The current approximated eigenvectors are $\mbf{c}_p = \sum_i x^p_i
\mbf{v}_i$.
\item Starting from the ground state ($p=1$), compute the residual vector 
$\mbf{r}_r = \sum_{i=1}^p\left(H - \lambda_i
 \right)\mbf{c}_i$. If $||\mbf{r_p}|| < \epsilon$, then move on to the next
excited state ($p\rightarrow p+1$). Otherwise,
compute the rescaled correction vector $\mbf{\sigma}^k_i = \left(\lambda_k 
- A_{ii}\right)r^k_i$.
\item Orthogonalize $\mbf{\sigma}^k$ with respect to $\mathcal{S}$ and normalize it. Add $\mbf{\sigma}^k$  to $\mathcal{S}$. If the size of $\mathcal{S}$
exceeds the preset maximum size, discard the earlest vectors.
\item Go back to Step 2 until the algorithm converges.
\end{enumerate}
The above algorithm iteratively finds the lowest $m$ eigenvectors of the 
Hamiltonian. Compared to other subspace methods such as Lanczos algorithm
mentioned in Chapter~\ref{chp:intro}, the Davidson algorithm is more 
accurate for both ground state and low-lying excited states. Note that
when updating the excited state, the already converged ground state might
be perturbed, therefore in Step 4, we recommend that one should always start
from calculating the residual of the ground state. To make the algorithm
faster, one could not worry about the ground state for a moment until 
all $m$ eigenvectors are derived, and then reexamine the residual of the 
ground state to make sure it is not perturbed.

\chapter{Appendix for Chapter~\ref{chp:qite}}
\section{Representing imaginary-time evolution by unitary maps}

As discussed in the main text, we  map the scaled non-unitary action of $e^{-\Delta\tau \hat{h}_m}$
  on a state $\Psi$ 
to that of a unitary $e^{-i\Delta\tau \hat{A}[m] }$, i.e.
\begin{align}
| \Psi^\prime \rangle \equiv c^{-1/2} \, e^{-\Delta \tau \hat{h}_m } |\Psi\rangle = e^{-i\Delta\tau \hat{A}[m]} |\Psi\rangle \quad .
\end{align}
where $c = \langle \Psi | e^{-2 \Delta \tau \hat{h}_m } |\Psi\rangle$.
$\hat{h}_m$ acts on $k$ geometrically local qubits; 
$\hat{A}$ is Hermitian and acts on a domain of $D$ qubits around the support of $\hat{h}_m$,
and is expanded as a sum of Pauli strings acting on the $D$ qubits,
\begin{align}
\hat{A}[m] &= \sum_{i_1i_2 \ldots i_D} a[m]_{i_1i_2 \ldots i_D} \sigma_{i_1}\sigma_{i_2} \ldots \sigma_{i_D} \notag \\
   &= \sum_I a[m]_I \sigma_I \label{eq:pauli}
\end{align}
where $I$ denotes the index $i_1i_2 \ldots i_D$. Define $  |\Delta_0\rangle = \frac{| \Psi^\prime \rangle - | \Psi\rangle}{\Delta \tau}$
and $|\Delta\rangle = -i \hat{A}[m] |\Psi\rangle$.
Our goal is to minimize the difference $||\Delta_0 - \Delta||$. If the unitary $e^{-i\Delta\tau \hat{A}[m] }$ is defined over a 
sufficiently large domain $D$, then this error minimizes at $\sim 0$, for small $\Delta \tau$. Minimizing for real $a[m]$
corresponds to minimizing the quadratic function $f(a[m])$
\begin{align}
f(a[m]) = f_0 + \sum_I b_I a[m]_I + \sum_{IJ} a[m]_I S_{IJ} a[m]_J
\end{align}
where
\begin{align}
f_0 &= \langle \Delta_0 | \Delta_0 \rangle \quad , \\
S_{IJ} &= \langle \Psi | \sigma^\dag_I \sigma_J | \Psi\rangle \quad ,\\
b_I &= i \, \langle \Psi | \sigma^\dag_I | \Delta_0 \rangle -  i \, \langle  \Delta_0 | \sigma_I | \Psi \rangle \quad ,
\end{align}
whose minimum obtains at the solution of the linear equation
\begin{align}
\left( \mathbf{S}+\mathbf{S}^T \right) \mathbf{a}[m] = -\mathbf{b} \label{eq:lineareq}
\end{align}
In general, $\mathbf{S}+\mathbf{S}^T$ may have a non-zero null-space. Thus, we solve Eq.~\eqref{eq:lineareq}
either by applying the generalized inverse of $\mathbf{S}+\mathbf{S}^T$ or by an iterative algorithm such as conjugate gradient.

For fermionic Hamiltonians, we replace the Pauli operators in Eq.~(\ref{eq:pauli}) by fermionic field operators.
For a number conserving Hamiltonian, such as the fermionic Hubbard Hamiltonian treated in Fig. 3 in the main text, 
we write
\begin{align}
\hat{A}[m] &= \sum_{i_1i_2 \ldots i_D} a[m]_{i_1i_2 \ldots i_D} \hat{f}^\dag_{i_1}\ldots \hat{f}^\dag_{i_{D/2}} \hat{f}_{i_{D/2+1}} \ldots \hat{f}_{i_D}
\end{align}
where $\hat{f}^\dag$, $\hat{f}$ are fermionic creation, annihilation operators respectively.

\section{Proof of correctness from finite correlation Length}

Here we present a more detailed analysis of the running time of the algorithm. Consider a $k$-local Hamiltonian 
\begin{equation}
H = \sum_{l=1}^m h_l 
\end{equation}
acting on a $d$-dimensional lattice with $\Vert h_i \Vert \leq 1$, where $\Vert * \Vert$ is the operator norm. In imaginary time evolution (used e.g. in Quantum Monte-Carlo or in tensor network simulations) one typically applies Trotter formulae to approximate 
\begin{equation}
\frac{e^{- \beta H} | \Psi_0 \rangle}{ \Vert   e^{- \beta H} | \Psi_0 \rangle \Vert   }
\end{equation}
for an initial state $| \Psi_0 \rangle$ (which we assume to be a product state) by 
\begin{equation}  \label{trotterdecomp}
 \frac{  \left (  e^{- t h_1 / n} \ldots  e^{- t h_m / n}  \right)^{n} | \Psi_0 \rangle   }   { \Vert  \left (  e^{- t h_1 / l} \ldots  e^{- t h_m / n}  \right)^{n} | \Psi_0 \rangle   \Vert  }.
\end{equation}
This approximation leads to an error which can be made as small as one wishes by increasing the number of time steps $n$.

Let $| \Psi_s \rangle$ be the state (after renormalization) obtained by applying $s$ terms $e^{- t h_i / n} $ from $\left(  e^{- t h_1 / n} \ldots  e^{- t h_m / n}  \right)^{n}$; with this notation $| \Psi_{mn} \rangle$ is the state given by Eq. (\ref{trotterdecomp}). In the QITE algorithm, instead of applying each of the operators $ e^{- t h_i / n}$ to $| \Psi_0 \rangle$ (and renormalizing the state), one applies local unitaries $U_s$ which should approximate the action of the original operator. Let $| \Phi_s \rangle$ be the state after $s$ unitaries have been applied. 

Let $C$ be an upper bound on the correlation length of $| \Psi_s \rangle$ for every $s$: we assume that for every $s$, and every observables $A$ and $B$ separated by $\text{dist}(A, B)$ sites, 
\begin{equation} \label{correlationdecay}
\langle \Psi_s | A \otimes B | \Psi_s \rangle - \langle \Psi_s | A   | \Psi_s \rangle \langle \Psi_s |  B  | \Psi_s \rangle \leq \Vert A \Vert \Vert B \Vert e^{- \text{dist}(A, B) / C}.
\end{equation}

\begin{theorem}  \label{bounderrors}
For every $\varepsilon > 0$, there are unitaries $U_s$ each acting on 
\begin{equation}
k  (2 C)^d \ln^d(2 \sqrt{2} n m \varepsilon^{-1})
\end{equation}
qubits such that
\begin{equation}
\left \Vert   | \Psi_{mn} \rangle  -  | \Phi_{mn} \rangle  \right \Vert \leq \varepsilon 
\end{equation}

\end{theorem}

\begin{proof}

We have
\begin{eqnarray} \label{boundingerror1}
\left \Vert  | \Psi_{s} \rangle  -| \Phi_{s} \rangle   \right  \Vert &=&
\left \Vert  | \Psi_{s} \rangle   - U_s | \Phi_{s-1} \rangle   \right  \Vert \nonumber \\
&\leq& \left \Vert   | \Psi_{s} \rangle   - U_s | \Psi_{s-1} \rangle   \right  \Vert    +  \left \Vert | \Psi_{s-1} \rangle  - | \Phi_{s-1} \rangle \right \Vert    
\end{eqnarray}

To bound the first term we use our assumption that the correlation length of $| \Psi_{s-1} \rangle$ is smaller than $C$. Consider a region $R_{v}$ of all sites that are at most a distance $v$ (in the Manhattan distance on the lattice) of the sites in which $h_{i_s}$ acts. Let $\text{tr}_{\backslash R_v}(| \Psi_s \rangle \langle \Psi_s | )$ be the reduced state on $R_v$, obtained by partial tracing over the complement of $R_v$ in the lattice. Since
\begin{equation}
 | \Psi_{s} \rangle = \frac{ e^{-\beta h_{i_s}/n}  | \Psi_{s-1} \rangle }{ \Vert e^{ - \beta h_{i_s}/n} | \Psi_{s-1} \rangle \Vert },
\end{equation}
it follows from Eq. (\ref{correlationdecay}) and Lemma 9 of \cite{brandao2015exponential} that 
\begin{equation} \label{boundmarginal}
\left \Vert \text{tr}_{\backslash R_v}(| \Psi_s \rangle \langle \Psi_s | ) -  \text{tr}_{\backslash R_v}(| \Psi_{s-1} \rangle \langle \Psi_{s-1} | ) \right \Vert_1 \leq   \Vert e^{h_{i_s}/n} \Vert^{-1} e^{- \frac{v}{C}} \leq 2 e^{- \frac{v}{C}},
\end{equation}
where we used that for $n \geq 2\beta$, $\Vert e^{- \beta h_{i_s}/n} \Vert \geq \Vert I - \beta h_{i_s}/n \Vert \geq 1 - \beta/n \geq 1/2$. Above $\Vert * \Vert_1$ is the trace norm. 

The key result in our analysis is Uhlmann's theorem (see e.g. Lemmas 11 and 12 of \cite{brandao2015exponential}). It says that two pure states with nearby marginals must be related by a unitary on the purifying system. In more detail, if $| \eta \rangle_{AB}$ and $| \nu \rangle_{AB}$ are two states s.t. $\Vert \eta_A - \nu_A \Vert_1 \leq \delta$, then there exists a unitary $V$ acting on $B$ s.t. 
\begin{equation} \label{uhlmannstatement}
\Vert | \eta \rangle_{AB} - (I \otimes V) | \nu \rangle_{AB}    \Vert \leq 2 \sqrt{\delta}.
\end{equation}

Applying Uhlmann's theorem to $| \Psi_s \rangle$ and $| \Psi_{s-1} \rangle$, with $B = R_v$, and using Eq. (\ref{boundmarginal}), we find that there exists a unitary $U_s$ acting on $R_{v}$ s.t. 
\begin{equation}
 \left \Vert  | \Psi_{s} \rangle   - U_s | \Psi_{s-1} \rangle   \right  \Vert \leq 2 \sqrt{2}   e^{- \frac{v}{2C}},
\end{equation}
which by Eq. (\ref{boundingerror1}) implies 
\begin{equation}
 \left \Vert  | \Psi_{s} \rangle   - U_s | \Psi_{s-1} \rangle   \right  \Vert \leq 2 \sqrt{2}  m  n e^{- \frac{v}{2C}},
\end{equation}

Choosing $\nu = 2 C \ln(2 \sqrt{2} n m \varepsilon^{-1})$ as the width of the support of the approximating unitaries, the error term above is $\varepsilon$. The support of the local unitaries is $k \nu^d$ qubits (as this is an upper bound on the number of qubits in $R_d$). Therefore each unitary $U_s$ acts on at most
\begin{equation}
k  (2 C)^d \ln^d(2 \sqrt{2} n m \varepsilon^{-1})
\end{equation}
qubits. 



\end{proof}

\vspace{0.4 cm}

\noindent \textit{Finding $U_s$:} In the algorithm we claim that we can find the unitaries $U_s$ by solving a least-square problem. This is indeed the case if we can write them as $U_s = e^{i A[s] / n}$ with $A[s]$ a Hamiltonian of constant norm. Then for sufficiently large $l$, $U_s = I + i A[s]/n + O((1/n)^2)$ and we can find $A[s]$ by performing tomography of the reduced state over the region where $U_s$ acts and solving the linear problem given in the main text. Because we apply Uhlmann's Theorem to $ | \Psi_{s-1} \rangle$ and
\begin{equation}
\frac{ e^{- \beta h_{i_s}/n}  | \Psi_{s-1} \rangle }{ \Vert e^{ - \beta  h_{i_s}/n} | \Psi_{s-1} \rangle \Vert },
\end{equation}
using $e^{ - \beta h_{i_s}/n} = I - \beta h_{i_s}/n + O((1/n)^2)$ and following the proof of the Uhlmann's Theorem, we find that the unitary can indeed be taken to be close to the identity, i.e. $U_s $ can be written as $e^{i A[s] / n}$ 

\vspace{0.4 cm}

\noindent \textit{Total Running Time:} Theorem \ref{bounderrors} gives an upper bound on the maximum support of the unitaries needed for a Trotter update, while tomography of local reduced density matrices gives a way to find the unitaries. The cost for tomography is quadratic in the dimension of the region, so it scales as $\exp(O( k  (2 C)^d \ln^d(2 \sqrt{2} n m \varepsilon^{-1})))$. This is also the cost to solve classically the linear system which gives the associated Hamiltonian $A[s]$ and of finding a circuit decomposition of $U_s = e^{i A[s] / n}$ in terms of 
two qubit gates. As this is repeated $mn$ times, for each of the $mn$ terms of the Trotter decomposition, the total running time (of both quantum and classical parts) is
\begin{equation}
ml \exp(O( k  (2 C)^d \ln^d(2 \sqrt{2} n m \varepsilon^{-1}))).
\end{equation}
This is exponential in $(C)^d$, with $C$ the correlation length, and quasi-polynomial in $n$ (the number of Trotter steps) and $m$ (the number of local terms in the Hamiltonian. Note that typically $m = O(N)$, with $N$ the number of sites). While this an exponential improvement over the $\exp(O(N))$ scaling classically, the quasi-polynomial dependence on $m$ is still prohibitive in practice. Below we show how to improve on that.

\vspace{0.4 cm}

\noindent \textit{Local Approximation:} We expect in practice to substantially beat the bound on the support of the unitaries given in Theorem \ref{bounderrors} above. Indeed, if one is only interested in a local approximation of the state (meaning that all the local marginals of $|\Phi_{nm} \rangle$ are close to the ones of $e^{- \beta H} |\Psi_0 \rangle$, but not necessarily the global states), then we expect the support of the unitaries to be independent of the number of terms of the Hamiltonian $m$ (while for global approximation we get a polylogarithmic dependence on $m$). 

The scaling with $m$ in the bound comes from the additive accumulation of error from each of the $ml$ steps (Eq. (\ref{boundingerror1})). The assumption of a correlation length $C$ ensures that the errors of replacing each local term in the Trotter decomposition by a unitary do not all add up if one is interested in local observables. Indeed, the contribution of the local error for a region $S$ from the replacement of $e^{-  \beta h_{j_s} / n}$ by $U_s$ is $\exp(- l / C)$, with $l$ the distance of the support of $h_{j_s}$ to $S$. Then we can substitute Eq.  (\ref{boundmarginal}) by 
\begin{equation}  \label{localerror term}
\left  \Vert   \text{tr}_{\backslash S} (  | \Psi_{mn} \rangle \langle \Psi_{mn} | ) -  \text{tr}_{\backslash S} (  | \Phi_{mn} \rangle \langle \Phi_{mn} | )    \right \Vert  \leq 2\sqrt{2}  n (C + |S|)  e^{- \frac{v}{2C}}. 
\end{equation}
with $|S|$ the size of the support of $S$. This gives a bound on the size of the support of the unitaries $U_s$ of 
\begin{equation}
k  (2 C)^d \ln^d(2 \sqrt{2} n (C + |S|) \varepsilon^{-1})
\end{equation}

Using this improved bound, the total running time becomes
\begin{equation}
ml \exp(O(  k  (2 C)^d \ln^d(2 \sqrt{2} n (C + |S|) \varepsilon^{-1})  )).
\end{equation}
As $m = O(N)$, we find  the scaling with the number of sites $N$ to be linear.

\vspace{0.4 cm}

 \noindent \textit{Non-Local Terms:} Suppose the Hamiltonian has a term $h_q$ acting on qubits which are not nearby, e.g. on two sites $i$ and $j$. Then $e^{- \beta h_q /n}$ can still be replaced by an unitary, which only acts on sites $i$  and $j$ and qubits in the neighborhoods of the two sites. This is the case if we assume that the state has a finite correlation length and the proof is again an application of Uhlmann's theorem (we follow the same argument from the proof of Theorem \ref{bounderrors}  but define $R_v$ in that case as the union of the neighborhoods of $i$ and $j$). Note however that the assumption of a finite correlation length might be less natural for models with long range interactions.

\section{Spreading of correlations}

In the main text, we argued that the correlation volume $V$ of the state $e^{-\beta H}|\Psi\rangle$ is bounded
for many physical Hamiltonians and saturates at the ground-state with $V \ll N$ where $N$ is the system size.
To numerically measure correlations, we use the mutual information between two sites, defined as
\begin{align}
I(i,j) =   S(i)+S(j) - S(i,j)
\end{align}
where $S(i)$ is the von Neumann entropy of the density matrix of site $i$ ($\rho(i)$) and similarly for $S(j)$, and $S(i,j)$
is the von Neumann entropy of the two-site density matrix for sites $i$ and $j$ ($\rho(i,j)$).

To compute the mutual information in Fig. 1 in the main text, we used matrix product state (MPS) and finite projected entangled pair state (PEPS) imaginary time evolution for the spin-$1/2$ 1D and 2D FM transverse field Ising model (TFI)
\begin{align}
  H_{TFI} = - \sum_{\langle ij \rangle} \sigma^z_i \sigma^z_j  - h \sum_{i} \sigma^x_i
\end{align}
where the sum over $\langle i, j \rangle$ pairs are over nearest neighbors. We use the parameter $h=1.25$ for the 1-D calculation and $h=3.5$ for the 2-D calculations as the ground-state is gapped in both cases.  It is known that the ground-state correlation length is finite.

\noindent \textbf{MPS}.
We performed MPS imaginary time evolution (ITE) on a 1-D spin chin with $L=50$ sites with open boundary conditions. We start from an initial state that is a random product state, and perform ITE using time evolution block decimation (TEBD) \cite{vidal2004TEBD,schollwock2011mps} with a first order Trotter decomposition. In this algorithm, the Hamiltonian is separated into terms operating on even and odd bonds. The operators acting on a single bond are exponentiated exactly. One time step is given by time evolution of odd and even bonds sequentially, giving rise to a Trotter error on the order of the time step $\Delta \tau$. In our calculation, a time step of $\Delta \tau = 0.001$ was used. 

We carry out ITE simulations with maximum bond dimension of $D=80$, but truncate singular values less than 1.0e-8 of the maximum singular value.  
In the main text, the ITE results are compared against the ground state obtained via the density matrix renormalization group (DMRG)). This should be equivalent to comparing to a long-time ITE ground state.  The long-time ITE ($\beta=38.352$) ground state reached an energy per site of -1.455071, while the DMRG ground-state energy per site is -1.455076. The percent error of the nearest neighbor correlations are on the order of 1.0e-4\% to 1.0e-3\%, and about 1.0e-2\% for correlations between the middle site and the end sites (a distance of 25 sites). The error in fidelity between the two ground states was about 5.0e-4.


\noindent \textbf{PEPS}. We carried out finite PEPS \cite{nishino1996corner,verstraete2004renormalization,
verstraete2006criticality,orus2014practical} imaginary time evolution for the two-dimensional transverse
field Ising model on a lattice
size of $21 \times 31$. The size was chosen to be large enough to see the spread of mutual information in the bulk 
without significant effects from the boundary. The mutual information was calculated
along the long (horizontal) axis in the center of the lattice.
The standard Trotterized
imaginary time evolution scheme for PEPS~\cite{VerstraeteITimeReview} was used with a time step $\Delta \tau = 0.001$,
up to imaginary time $\beta = 6.0$, starting from a random product state. To reduce computational cost from the
large lattice size, the PEPS
was defined in a translationally invariant manner with only 2 independent tensors \cite{VerstraeteIPEPS}
updated via the so-called ``simple update'' procedure \cite{XiangSimpleUpdate}. 
The simple update has been shown to be sufficiently accurate for capturing correlation 
functions (and thus $I(i,j)$) for ground states with relatively short correlation 
lengths (compared to criticality) \cite{LubaschAlgos,LubaschUnifying}. 
We chose a magnetic field value $h=3.5$ which is 
detuned from the critical field ($h \approx 3.044$) but         
still maintains a correlation length long enough to see
interesting behaviour.

\noindent \textit{Accuracy:} Even though the simple update procedure was used for the tensor update,
we still needed to contract the $21 \times 31$ PEPS at at every imaginary time step $\beta$ for a range of
correlation functions, amounting to a large number of contractions.
To control the computational cost, we limited  our bond dimension to $D=5$ and used an optimized
contraction scheme \cite{XiangContract}, with maximum allowed bond dimension 
of $\chi = 60$ during the contraction.
Based on converged PEPS ground state correlation functions
with a larger bond dimension of $D=8$, our $D=5$ PEPS yields $I(i,i+r)$ (where $r$ denotes horizontal separation) at large $\beta$ with
a relative error of $\approx 1\%$ for $r=1-4$, $5\%$ or less for $r=5-8$, and $10\%$ or greater for
$r > 8$. At smaller values of $\beta$ ($< 0.5$) the errors up to $r=8$ are much smaller because 
the bond dimension of 5 is able to completely support the smaller correlations (see Fig. 1, main text).
While error analysis on the 2D Heisenberg model \cite{LubaschAlgos} suggests
that errors with respect to $D=\infty$ may be larger, such analysis also confirms that 
a $D=5$ PEPS captures the qualitative behaviour of
correlation in the range $r=5-10$ (and beyond).
Aside from the bond dimension error, 
the precision of the calculations is governed by $\chi$ and the lattice size. Using the $21 \times 31$
lattice and $\chi = 60$, we were able to converge entries of single-site
density matrices $\rho(i)$ to a precision of $\pm 10^{-6}$ (two site density matrices $\rho(i,j)$
had higher precision). For $\beta = 0.001-0.012$,
the smallest eigenvalue of $\rho(i)$ fell below this precision threshold, leading
to significant noise in $I(i,j)$. Thus, these values of $\beta$ are omitted from Fig. 1 (main text) 
and the smallest reported values of $I$ are $10^{-6}$, although with more precision we expect $I \to 0$ as $r \to \infty$.
 
Finally, the energy and fidelity errors were computed with respect to the PEPS
ground state \textit{of the same bond dimension} at $\beta = 10.0$ (10000 time steps).
The convergence of the these quantities shown in Fig. 1 (main text) thus isolates the convergence of the
imaginary time evolution, and does not include effects of other errors that 
may result from deficiencies in the wavefunction ansatz.

\section{Parameters used in QVM and QPUs simulations}

In this section, we include the parameters used in our QPUs and QVM simulations.
Note that all noisy QVM simulations (unless stated otherwise in the text) were performed with noise parameters from noise model 1.

\begin{table}[h!]
  \begin{center}
    \caption{QPUs: 1-qubit QITE and QLanczos.}
    \label{tab:table1}
    \begin{tabular}{l|c|c|c|r} 
      \textbf{Trotter stepsize} & \textbf{nTrials} & \textbf{$\delta$} & \textbf{s} & \textbf{$\epsilon$}\\
      \hline
      0.2 & 100000 & 0.01 & 0.75 & $10^{-2}$\\
    \end{tabular}
  \end{center}
\end{table}

\begin{table}[h!]
  \begin{center}
    \caption{QPUs: 2-qubit QITE and QLanczos.}
    \label{tab:table1}
    \begin{tabular}{l|c|c|c|r} 
      \textbf{Trotter stepsize} & \textbf{nTrials} & \textbf{$\delta$} & \textbf{s} & \textbf{$\epsilon$}\\
      \hline
      0.5 & 100000 & 0.1 & 0.75 & $10^{-2}$\\
    \end{tabular}
  \end{center}
\end{table}

\begin{table}[h!]
  \begin{center}
    \caption{QPUs: 1-qubit METTS.}
    \label{tab:table1}
    \begin{tabular}{l|c|c|c|r} 
      \textbf{$\beta$} & \textbf{Trotter stepsize} & \textbf{nTrials} & \textbf{nMETTs} & \textbf{$\delta$}\\
      \hline
      1.5 & 0.15 & 1500 & 70 & 0.01\\
      2.0 & 0.20 & 1500 & 70 & 0.01\\
      3.0 & 0.30 & 1500 & 70 & 0.01\\
      4.0 & 0.40 & 1500 & 70 & 0.01\\
    \end{tabular}
  \end{center}
\end{table}

\begin{table}[h!]
  \begin{center}
    \caption{QVM: 2-qubit QITE and QLanczos.}
    \label{tab:table1}
    \begin{tabular}{l|c|c|c|r} 
      \textbf{Trotter stepsize} & \textbf{nTrials} & \textbf{$\delta$} & \textbf{s} & \textbf{$\epsilon$}\\
      \hline
      0.5 & 100000  & 0.1 & 0.75  & $10^{-2}$\\
    \end{tabular}
  \end{center}
\end{table}

\begin{table}[h!]
  \begin{center}
    \caption{QVM: 1-qubit METTS.}
    \label{tab:table1}
    \begin{tabular}{l|c|c|c|r} 
      \textbf{$\beta$} & \textbf{Trotter stepsize} & \textbf{nTrials} & \textbf{nMETTs} & \textbf{$\delta$}\\
      \hline
      1.0 & 0.10 & 1500 &70 & 0.01\\
      1.5 & 0.15 & 1500 &70 & 0.01\\
      2.0 & 0.20 & 1500 &70 & 0.01\\
      3.0 & 0.30 & 1500 &70 & 0.01\\
      4.0 & 0.40 & 1500 &70 & 0.01\\
    \end{tabular}
  \end{center}
\end{table}

\begin{table}[h!]
 \begin{center}
    \caption{QVM: 2-qubit METTS.}
    \label{tab:table1}
    \begin{tabular}{l|c|c|c|r} 
      \textbf{$\beta$} & \textbf{Trotter stepsize} & \textbf{nTrials} & \textbf{nMETTs} & \textbf{$\delta$}\\
      \hline
      1.0 & 0.10 & 10000 & 200 & 0.1\\
      1.5 & 0.15 & 10000 & 200 & 0.1\\
      2.0 & 0.20 & 10000 & 200 & 0.1\\
      3.0 & 0.30 & 10000 & 200 & 0.1\\
      4.0 & 0.40 & 10000 & 200 & 0.1\\
    \end{tabular}
  \end{center}
\end{table}

%
%

\end{appendices}
\printbibliography[heading=bibintoc]


\end{document}